\documentclass[abstract,10pt]{scrartcl}
\pdfoutput=1
\usepackage[utf8]{inputenc}
\usepackage[english]{babel}
\usepackage{mathtools}
\usepackage[numbers,sort&compress]{natbib}
\usepackage{hyphenat}
\usepackage{booktabs}
\usepackage{amsthm}
\usepackage{paralist}
\usepackage{txfonts}
\usepackage{microtype}
\usepackage[textsize=footnotesize]{todonotes}
\usepackage[ruled,vlined,linesnumbered,algosection]{algorithm2e}
\usepackage{authblk}
\usepackage{flushend}
\usepackage[colorlinks, allcolors=blue]{hyperref}
\usepackage{doi}
\usepackage[sort&compress,nameinlink,noabbrev,capitalize]{cleveref}

\SetCommentSty{normalfont}

\usepackage{fullpage}
\addtokomafont{title}{\raggedright}
\setkomafont{paragraph}{\normalfont\itshape}
\addtokomafont{date}{\sffamily}
\setkomafont{subsection}{\normalfont\itshape}
\setkomafont{subsubsection}{\normalfont\itshape}
\setcapindent{0em}

\usepackage{tikz}
\usetikzlibrary{arrows,calc}
\usetikzlibrary{calc,shapes,decorations.shapes}

\numberwithin{equation}{section}
\numberwithin{figure}{section}
\numberwithin{table}{section}

\newtheorem{theorem}{Theorem}[section]
\theoremstyle{plain}
\newtheorem{lemma}[theorem]{Lemma}
\newtheorem{corollary}[theorem]{Corollary}
\newtheorem{example}[theorem]{Example}
\newtheorem{proposition}[theorem]{Proposition}
\newtheorem{observation}[theorem]{Observation}
\newtheorem{hypothesis}[theorem]{Hypothesis}
\newtheorem{definition}[theorem]{Definition}
\theoremstyle{definition}
\newtheorem{assumption}[theorem]{Assumption}
\newtheorem{construction}[theorem]{Construction}
\newtheorem{problem}[theorem]{Problem}
\newtheorem{rrule}[theorem]{Reduction Rule}

\theoremstyle{remark}

\newtheorem{remark}[theorem]{Remark}

\newtheoremstyle{iostuff}%
{0pt}%
{0pt}%
{\hangindent=\parindent}%
{}%
{\itshape}%
{:}%
{.5em}%
{}%

\theoremstyle{iostuff}
\newtheorem*{probinstance}{Instance}
\newtheorem*{probtask}{Question}

\crefname{construction}{Construction}{Constructions}
\crefname{paragraph}{Section}{Sections}
\crefname{rrule}{Reduction~Rule}{Reduction~Rules}
\Crefname{rrule}{RR}{RRs}
\crefname{observation}{Observation}{Observations}
\Crefname{observation}{Obs.}{Obs.}
\crefname{lemma}{Lemma}{Lemmas}
\Crefname{lemma}{Lem.}{Lem.}
\crefname{theorem}{Theorem}{Theorems}
\Crefname{theorem}{Thm.}{Thm.}
\crefname{proposition}{Proposition}{Propositions}
\Crefname{proposition}{Prop.}{Props.}
\crefname{remark}{Remark}{Remarks}
\Crefname{remark}{Rem.}{Rem.}
\crefname{prop}{Property}{Properties}
\crefname{corollary}{Corollary}{Corollaries}
\Crefname{corollary}{Cor.}{Cors.}
\creflabelformat{prop}{(#2#1#3)}
\crefname{line}{line}{lines}
\crefname{section}{Section}{Sections}
\Crefname{section}{Sec.}{Secs.}

\newcommand{\ssfree}{\hyp subgraph\hyp free}
\newcommand{\no}{no}
\newcommand{\rmc}{\text{rmc}}

\newcommand{\glue}{\text{glue}}
\newcommand{\shift}{\text{shift}}
\newcommand{\proj}{\text{proj}}
\newcommand{\join}{\text{join}}
\newcommand{\mincup}{%
  \mathbin{\ooalign{$\cup$\cr\hss\raisebox{0.5ex}{\scriptsize $\downarrow$}\hss}}}%
\newcommand{\bigmincup}{%
\mathop{\ooalign{$\displaystyle\bigcup$\cr\hidewidth{\Large$\downarrow$}\hidewidth\cr}}}%

\newcommand{\mcost}{\text{minK}}
\newcommand{\opt}{\text{Opt}}

\newcommand{\yes}{yes}

\newcommand{\Ne}{N}
\newcommand{\nei}{\Lambda}
\newcommand{\cost}{K}
\newcommand{\Dz}{\ensuremath{D_\text{z}}}
\newcommand{\De}{\ensuremath{D_\text{e}}}
\newcommand{\Di}{\ensuremath{D_\text{i}}}
\newcommand{\iz}{{i_\text{z}}}
\newcommand{\ie}{{i_\text{e}}}
\newcommand{\ii}{{i_\text{i}}}
\newcommand{\N}{\mathbb{N}}
\newcommand{\Q}{\mathbb{Q}}
\newcommand{\Z}{\mathbb{Z}}
\newcommand{\I}{I}
\newcommand{\calR}{\mathcal{R}}

\newcommand{\calQ}{\mathcal{Q}}
\newcommand{\calT}{\mathcal{T}}

\newcommand{\WK}[1]{\ensuremath{\text{WK}[#1]}}
\newcommand{\poly}{\ensuremath{\text{poly}}}
\newcommand{\NP}{\ensuremath{\text{NP}}}
\newcommand{\coNP}{\ensuremath{\text{coNP}}}
\newcommand{\unlessPK}{\ensuremath{\coNP\subseteq \NP/\poly}}

\newcommand{\tw}{\ensuremath{\text{tw}}}

\newcommand{\vc}{{\ensuremath{\text{vc}}}}
\newcommand{\prtn}{\ensuremath{p}}
\newcommand{\pts}{\ensuremath{\mathcal{P}}}
\newcommand{\feas}{\mathcal E}
\newcommand{\crn}{{\ensuremath{\text{cr}}}}
\newcommand{\fes}{\ensuremath{\text{fes}}}
\newcommand{\fvs}{\ensuremath{\text{fvs}}}

\newcommand{\ncc}[1]{\tilde{#1}}
\newcommand{\norm}[2]{\left\lVert#1\right\rVert_{#2}}

\DeclareMathOperator{\sign}{sign}
\DeclareMathOperator{\parent}{parent}

\newcommand{\decprob}[3]{
  \begin{problem}[\textsc{#1}]\leavevmode
    \begin{probinstance}
      #2
    \end{probinstance}
    \begin{probtask}
      #3
    \end{probtask}
  \end{problem}
}

\newcommand{\sspTsc}{\textsc{Short Secluded Path}}
\newcommand{\sspAcr}{\textsc{SSP}}

\newcommand{\WsspTsc}{\textsc{Vertex\hyp Weighted \sspTsc{}}}
\newcommand{\WsspAcr}{\textsc{VW-\sspAcr{}}}

\newcommand{\mcclique}{\textsc{Multicolored Clique}}

\newcommand{\bbT}{\mathbb{T}}

\newcommand{\LD}{($\Leftarrow$)}
\newcommand{\RD}{($\Rightarrow$)}

\newcommand{\simple}{simple}
\newcommand{\Ef}{F}

\newcommand{\thetitle}{Parameterized algorithms and data reduction for the short secluded $s$-$t$-path problem}

\date{}
\title{\boldmath\thetitle{}%
  \thanks{%
    A preliminary version of this work
    appeared
    in the Proceedings of the
    18th Workshop on Algorithmic Approaches for Transportation Modeling, Optimization, and Systems (ATMOS 2018),
    23--24 August, 2018, Helsinki, Finland
    \citep{BFT18}.
    This version contains full proof details,
    new kernelization results
    with respect to the feedback vertex number as
    parameter,
    and the algorithm
    for graphs of bounded treewidth
    has been generalized to a more general problem variant
    and accelerated.
  }
}

\author{René van Bevern\\
  Department of Mechanics and Mathematics,
  Novosibirsk State University, Novosibirsk, Russian Federation, \texttt{rvb@nsu.ru}\\
  \and
  Till Fluschnik\\
  Algorithmics and Computational Complexity, Faculty IV, TU Berlin, Berlin, Germany
  \texttt{till.fluschnik@tu-berlin.de}\\
  \and
  Oxana Yu.\ Tsidulko\\
  Sobolev Institute of Mathematics,
  Siberian Branch of the Russian Academy of Sciences,
  Novosibirsk, Russian Federation, \texttt{tsidulko@math.nsc.ru}\\
  Department of Mechanics and Mathematics,
  Novosibirsk State University, Novosibirsk, Russian Federation
}

\begin{document}
\maketitle
\begin{abstract}
  \looseness=-1
  Given a graph~\(G=(V,E)\),
  two vertices~\(s,t\in V\),
  and two integers~\(k,\ell\),
  the \textsc{Short Secluded Path} problem
  is to find a simple \(s\)-\(t\)-path
  with at most \(k\)~vertices
  and \(\ell\)~neighbors.
  We study the parameterized complexity
  of the problem with respect to four structural graph parameters:
  the vertex cover number, treewidth,
  feedback vertex number, and
  feedback edge number.
  In particular,
  we
  completely settle the question of the existence
  of problem kernels with size polynomial in
  these parameters and their combinations with~$k$ and~$\ell$.
  We also
  obtain a \(2^{O(\tw)}\cdot \ell^2\cdot n\)-time
  algorithm for graphs of treewidth~$\tw$,
  which yields subexponential\hyp time algorithms
  in several graph classes.
\end{abstract}

\paragraph{Keywords:} NP-hard problem · fixed-parameter tractability · problem kernelization · shortest path · kernelization lower bounds · treewidth · subexponential time

\section{Introduction}
Finding shortest paths
is a fundamental problem in route planning
and
has extensively been studied
with respect to efficient algorithms,
including data reduction and preprocessing
\citep{BDG+16}. 
In this work,
we study the following NP\hyp hard~\cite{LF18} variant
of finding shortest \(s\)-\(t\)-paths.

\decprob{\sspTsc~(\sspAcr)}
{An undirected, simple graph~$G=(V,E)$ with two distinct vertices~$s,t\in V$, and two integers~$k\geq2$ and~$\ell\geq0$.}
{Is there an \(s\)-\(t\)-path~$P$ in~$G$ such
  that $|V(P)|\leq k$ and $|N(V(P))|\leq \ell$?}

\noindent
The problem can be understood
as finding short and safe routes for a convoy
through a transportation network:
each neighbor of the convoy's travel path
requires additional precaution.
Thus,
we seek to minimize not only the length
of the convoy's travel path,
but also its number of neighbors.
In our work,
we study the
parameterized complexity of the
above basic, unweighted variant,
as well as weighted variants of the problem.
In particular,
given the effect that preprocessing and data reduction
had to fundamental routing problems
like finding shortest paths \citep{BDG+16},
we study the possibilities
of polynomial\hyp time data reduction
with \emph{provable performance guarantees}
for \sspAcr{}.

\paragraph{Fixed-parameter algorithms.}
Fixed\hyp parameter algorithms
have recently been applied
to numerous NP-hard routing problems
\citep{GWY17,GJW16,GMY13,GJS17,BKS17,SBNW12,SBNW11,BNSW14,DMNW13,BHKK07,GP16}.
In particular,
they led to subexponential\hyp time algorithms
for fundamental NP-hard routing problems
in planar graphs \citep{KM14}
and to algorithms  that work efficiently
on real\hyp world data \citep{BKS17}.

The main idea of fixed\hyp parameter algorithms
is to accept the exponential running time
seemingly inherent to solving NP-hard problems,
yet to restrict the combinatorial explosion
to a parameter of the problem,
which can be small in applications.
We call a problem \emph{fixed\hyp parameter tractable}
if it can be solved in \(f(k)\cdot n^{O(1)}\)~time
on inputs of length~\(n\)
and some function~\(f\)
depending only on some parameter~\(k\).
In contrast
to an algorithm
that merely runs in polynomial time
for fixed~\(k\), say, in \(O(n^k)\)~time,
which is intractable even for small values of~\(k\),
fixed\hyp parameter algorithms
can solve NP\hyp hard problems
quickly if \(k\)~is small.

\paragraph{Provably effective polynomial\hyp time data reduction.}
Parameterized complexity theory also
provides a framework for
data reduction with performance guarantees---\emph{problem kernelization}
\citep{DF13,FG06,Nie06,CFK+15}.

Kernelization allows
for provably effective polynomial\hyp time data reduction.
Note that a result of the form
``our polynomial\hyp time data reduction algorithm reduces
the input size by at least one bit, preserving optimality
of solutions''
is impossible for NP\hyp hard problems unless P\({}={}\)NP.
In contrast,
a kernelization algorithm
reduces a problem instance
into an equivalent one (the \emph{problem kernel})
whose size depends only (ideally polynomially)
on some problem parameter.
Problem kernelization has been successfully applied
to obtain effective polynomial\hyp time data
reduction algorithms
for many \NP-hard problems
\citep{GN07,Kra14}
and also led to techniques for proving
the limits
of polynomial\hyp time data reduction
\citep{BDFH09,MRS11,BJK14}.

\newcommand{\smtab}[1]{{\scriptsize (#1)}}
\renewcommand{\arraystretch}{1.25}
\begin{table*}[t]
  \centering
  \caption
  {
    Overview of our results.
    Herein, $n$, $\tw$, $\vc$, $\fes$, $\fvs$, $\crn$, and~$\Delta$
    denote the number of vertices,
    treewidth,
    vertex cover number,
    feedback edge number,
    feedback vertex number,
    the crossing number,
    and maximum degree of the input graph,
    respectively. %
  }
  \begin{tabular}{rp{0.5\textwidth}p{0.38\textwidth}}
    \toprule
    par.&positive results&negative results\\
    \midrule    
    \vc
             &size $\vc^{O(r)}$\hyp kernel in~$K_{r,r}$\ssfree{} graphs \smtab{\cref{thm:kernelkrr}}
                              &
                                No polynomial kernel
                                and WK[1]-hard w.\,r.\,t.\ \(\vc\) \smtab{\cref{thm:wk1hard}}\\
    \fes
             &size $\poly(\fes)$-kernel \smtab{\cref{thm:bikernelfes}}\\
    \fvs
             & $O(\fvs\cdot(k+\ell)^2)$-vertex kernel \smtab{\cref{thm:ssppk-fvskell}}&
                                                                                        No kernel with size~$\poly(\fvs+\ell)$ \smtab{\cref{thm:sspNoPKfvsell}}\\
    \tw
             & $2^{O(\tw)}\cdot\ell^2\cdot n$-time algorithm \smtab{\cref{thm:twsingexp}}
                              &
                                No kernel
                                with size \(\poly(\tw+k+\ell)\)
    even in planar graphs with
                                                     const.\ \(\Delta\) \smtab{\cref{thm:nopktwell}}\\
    \bottomrule
  \end{tabular}

  \label{tab:results}
\end{table*}	

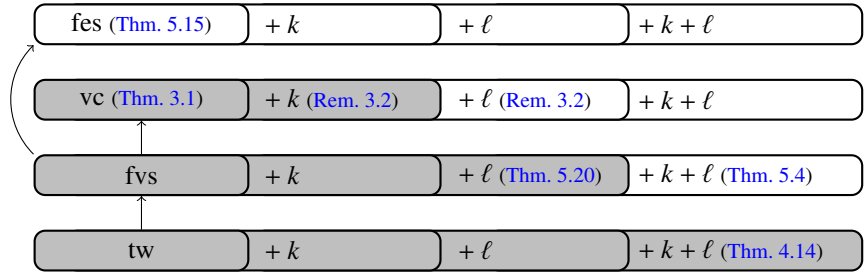
\begin{figure}
  \centering
  \begin{tikzpicture}

    \usetikzlibrary{arrows,patterns,calc}

    \def\xr{1}
    \def\yr{1}

    \tikzstyle{ppnode}=[rounded corners, thick, minimum width=140*\xr, minimum height=15*\xr, draw];
    \tikzstyle{pnode}=[rounded corners, thick, minimum width=80*\xr, minimum height=15*\xr, draw];
    \tikzstyle{nopk}=[fill=lightgray] %
    \tikzstyle{pk}=[fill=white] %
    \newcommand{\hasssize}{\footnotesize}
    \newcommand{\parabox}[8]{
      \node (#1-4) at (#2*\xr,#3*\yr)[ppnode,pk,label={[xshift=-90*\xr]0:{${}+k+\ell$~\hasssize#8}}]{};
      \node (#1-3) at ($(#1-4.west)+(-0.6*\xr,0)$)[ppnode,pk,label={[xshift=-70*\xr]0:{${}+\ell$~\hasssize#7}}]{};
      \node (#1-2) at (#1-3.west)[ppnode,pk,label={[xshift=-72*\xr]0:{${}+k$~\hasssize#6}}]{};
      \node (#1-1) at ($(#1-2.west)+(+1*\xr,0)$)[pnode,pk]{#4~\hasssize#5};
    }
    \parabox{fes}{7}{1}{$\fes$}{(\Cref{thm:bikernelfes})}{}{}{};
    \renewcommand{\parabox}[8]{
      \node (#1-4) at (#2*\xr,#3*\yr)[ppnode,pk,label={[xshift=-90*\xr]0:{${}+k+\ell$~\hasssize#8}}]{};
      \node (#1-3) at ($(#1-4.west)+(-0.6*\xr,0)$)[ppnode,pk,label={[xshift=-70*\xr]0:{${}+\ell$~\hasssize#7}}]{};
      \node (#1-2) at (#1-3.west)[ppnode,nopk,label={[xshift=-72*\xr]0:{${}+k$~\hasssize#6}}]{};
      \node (#1-1) at ($(#1-2.west)+(+1*\xr,0)$)[pnode,nopk]{#4~\hasssize#5};
    }
    \parabox{vc}{7}{0}{$\vc$}{(\Cref{thm:wk1hard})}{(\Cref{rem:vckr})}{(\Cref{rem:vckr})}{};
    \renewcommand{\parabox}[8]{
      \node (#1-4) at (#2*\xr,#3*\yr)[ppnode,pk,label={[xshift=-90*\xr]0:{${}+k+\ell$~\hasssize#8}}]{};
      \node (#1-3) at ($(#1-4.west)+(-0.6*\xr,0)$)[ppnode,nopk,label={[xshift=-70*\xr]0:{${}+\ell$~\hasssize#7}}]{};
      \node (#1-2) at (#1-3.west)[ppnode,nopk,label={[xshift=-72*\xr]0:{${}+k$~\hasssize#6}}]{};
      \node (#1-1) at ($(#1-2.west)+(+1*\xr,0)$)[pnode,nopk]{#4~\hasssize#5};
    }
    \parabox{fvs}{7}{-1}{$\fvs$}{}{}{(\Cref{thm:sspNoPKfvsell})}{(\Cref{thm:ssppk-fvskell})};
    \renewcommand{\parabox}[8]{
      \node (#1-4) at (#2*\xr,#3*\yr)[ppnode,nopk,label={[xshift=-90*\xr]0:{${}+k+\ell$~\hasssize#8}}]{};
      \node (#1-3) at ($(#1-4.west)+(-0.6*\xr,0)$)[ppnode,nopk,label={[xshift=-70*\xr]0:{${}+\ell$~\hasssize#7}}]{};
      \node (#1-2) at (#1-3.west)[ppnode,nopk,label={[xshift=-72*\xr]0:{${}+k$~\hasssize#6}}]{};
      \node (#1-1) at ($(#1-2.west)+(+1*\xr,0)$)[pnode,nopk]{#4~\hasssize#5};
    }
    \parabox{tw}{7}{-2}{$\tw$}{}{}{}{(\Cref{thm:nopktwell})};

    \draw[<-] (fes-1.south west) to [out=-135,in=135](fvs-1.north west);
    \draw[<-] (vc-1.south) -- (fvs-1.north);
    \draw[<-] (fvs-1.south) -- (tw-1.north);

  \end{tikzpicture}
  \caption{Overview on the existence of polynomial kernelization.
    Gray: no polynomial\hyp size kernel unless \unlessPK.
    White: polynomial\hyp size kernel exists.
    An arrow from parameter~$p$ to~$p'$
    means that the value of~$p$
    can be upper\hyp bounded
    by a polynomial in~$p'$ \citep{FJR13}.
    Thus,
    hardness results for~$p'$ also hold for~$p$
    and polynomial\hyp size kernels for $p$ also hold for~$p'$.}
  \label{fig:Hasse}
\end{figure}

\subsection{Our contributions}
We study the parameterized complexity of
\sspAcr{} (and a weighted variant)
with respect to four structural graph parameters:
the vertex cover number~$\vc$, the treewidth~$\tw$,
feedback vertex number~$\fvs$ and
feedback edge number~$\fes$.
Herein,
$\vc$~is interesting since lower bounds for it
are very strong:
$\vc$ bounds from above most
other known graph parameters \citep{FJR13}.
The other extreme is $\tw$,
which is $O(\sqrt{n})$
in many graph classes \citep{DH08}
and allows one to obtain subexponential\hyp time algorithms
in these.
Our results are summarized in \cref{tab:results} and \cref{fig:Hasse}.  The latter shows that
we completely settle the question of the existence
of problem kernels of size polynomial
with respect to $\vc$, $\tw$, $\fvs$, $\fes$, $k$ and~$\ell$
and all of their combinations.

In \cref{sec:apgs},
we show that SSP
has no problem kernel of size polynomial in \(\vc\)
unless \unlessPK.
In fact,
we even show that SSP is WK[1]-hard parameterized by~$\vc$;
WK[1]-hard problems are conjectured
to not even have polynomial\hyp size Turing kernels \citep{HKS+15b}.
We prove that SSP does have problem kernels
of size polynomial in~$\vc$ in $K_{r,r}$\ssfree{} graphs
for any constant~$r$.

In \cref{sec:tlgs},
we prove that (even the weighted version of) \sspAcr{}
is solvable in \(2^{O(\tw)}\cdot\ell^2\cdot n\)~time
in graphs of treewidth~\(\tw\).
This also gives subexponential $2^{O(\sqrt{n})}$-time
algorithms for many graph classes,
in particular for planar graphs.
Moreover,
we prove that \sspAcr{} is not solvable in
$2^{o(\sqrt{n})}$-time in planar graphs
unless the Exponential Time Hypothesis
fails.
We also prove that
there is no problem kernel
with size polynomial in~\(\tw+k+\ell\) unless \unlessPK{}.

Finally, in
\cref{sec:fes},
we show problem kernels
with \(O(\fes)\)~vertices
or
$O(\fvs\cdot(k+\ell)^2)$~vertices,
where \(\fes\)~is the feedback edge number
and \(\fvs\)~is the feedback vertex number of the input graph.
We also prove that,
unless \unlessPK{},
the latter kernel cannot be improved
to be of size polynomial in~\(\fvs+\ell\) or~\(\fvs+k\).

\subsection{Related work}

\looseness=-1
\citet{LF18} first defined \sspAcr{}
and analyzed its parameterized complexity
with respect to the parameters~\(k\) and~\(\ell\).
In contrast to their work,
we study problem parameters
that describe the structure of the input graphs.

\citet{CJPP17} introduced the similar
\textsc{Secluded Path} problem,
that,
given an undirected graph~$G=(V,E)$
with two designated vertices~$s,t\in V$, vertex-weights~$w:V\to\N$,
and two integers~$k,C\in\N$,
asks whether there is an $s$-$t$-path~$P$
such that the size of the
\emph{closed} neighborhood~$|N_G[V(P)]|\leq k$ and the weight of the closed neighborhood~$w(N_G[V(P)])\leq C$. 
\citet{FGKK17} studied the
parameterized complexity of the problem.
In particular,
they prove that \textsc{Secluded Path}
admits problem kernels
with size polynomial in~\(k\)
and the
feedback vertex number combined.
On the negative side,
they prove that \textsc{Secluded Path}
does not admit problem kernels
with size polynomial in the vertex cover number~\(\vc\).
Our negative results
on kernelization for \sspAcr{}
even show
WK[1]-hardness.

Van Bevern et al.~\citep{BFM+18} studied several classical graph optimization problems in both the ``secluded'' (small closed neighborhood) and the ``small secluded'' (small set with small open neighborhood) variants.
Amongst others, they prove that while finding a secluded~$s$-$t$ separator with small closed neighborhood remains solvable in polynomial time, finding a small secluded~$s$-$t$ separator is \NP-complete.

\citet{GHLM17} studied the
``small secluded'' scenario
for finding connected induced subgraphs
with given properties.
They prove that if
the requested property is
characterized through
finitely many forbidden induced subgraphs,
then the problem is fixed\hyp parameter tractable
when parameterized by the size~$\ell$ of the open neighborhood.
Their result
obviously does not generalize
to \sspAcr{},
since \sspAcr{} is NP-hard even for~\(\ell=0\)
\citep{LF18}.

\section{Preliminaries}

We use basic notation from graph theory~\cite{Die10} and parameterized algorithmics~\cite{DF13,FG06,Nie06,CFK+15}.
By $\mathbb N$ we denote the positive integers.
By $A\uplus B$,
we denote the union of two
sets~$A$ and~$B$
when we emphasize that \(A\) and~\(B\) are disjoint.
We denote by~$\log$ the logarithm with base 2.

\subsection{Graph theory}
\looseness=-1
We study simple, finite, undirected graphs~\(G=(V,E)\).
We denote by $V(G):=V$ the set of \emph{vertices of~$G$}
and by $E(G):=E$ the set of \emph{edges of~\(G\)}.
We denote \(n:=|V|\) and \(m:=|E|\).
For any subset~$U\subseteq V$ of vertices, we denote by~$N_G(U)=\{w\in V\setminus U\mid \exists v\in U:\{v,w\}\in E \}$ the \emph{open neighborhood} of~$U$ in~$G$.
When the graph~\(G\) is clear from the context,
we drop the subscript~\(G\).
A~set~$U\subseteq V$ of vertices is a \emph{vertex cover}
if every edge in~\(E\) has an endpoint in~\(U\).
The size of a minimum vertex cover
is called \emph{vertex cover number~\(\vc(G)\)} of~\(G\).
A set~\(F\subseteq E\) of edges is a \emph{feedback edge set}
if the graph~\((G,E\setminus F)\) is a forest.
The minimum size of a feedback edge set
in a connected graph is \(m-n+1\)
and is called
the \emph{feedback edge number \(\fes(G)\)} of~\(G\).
A set~\(V'\subseteq V\) of edges is a \emph{feedback vertex set}
if the graph~\(G-V':=(V\setminus V',\{e\in E(G)\mid e\cap V'=\emptyset\})\) is a forest.
The minimum size of a feedback vertex set is called
the \emph{feedback vertex number \(\fvs(G)\)} of~\(G\).
The \emph{crossing number~\(\crn(G)\)} of~\(G\)
is the minimum number of crossings
in any drawing of~\(G\) in the two-dimensional plane
(where only two edges are allowed to cross in each point). We say that graph~$G$ is $K_{r,r}$\ssfree{} if it does not contain a complete bipartite graph with parts of size $r$ as a subgraph.
A path~$P=(V,E)$ is a graph with vertex set~$V=\{x_0,x_1,\ldots,x_p\}$ and edge set~$E=\{\{x_i,x_{i+1}\}\mid 0\leq i<p\}$.
We say that~$P$ is an~$x_0$-$x_p$-path of length~$p$.
We also refer to~$x_0,x_p$ as the \emph{end points} of~$P$, and to all vertices~$V\setminus\{x_0,x_p\}$ as the \emph{inner} vertices of~$P$. 

\subsection{Fixed-parameter tractability and problem kernels}

Let~$\Sigma$ be a finite alphabet.
A \emph{parameterized problem}~$L$ is a subset~$L\subseteq \Sigma^*\times \N$.
An instance~$(x,k)\in \Sigma^*\times \N$ is a~\emph{\yes-instance} for~$L$ if and only if~$(x,k)\in L$.
We call \(x\)~the \emph{input} and \(k\)~the \emph{parameter}.

\begin{definition}[fixed-parameter tractability, FPT]
  A parameterized problem~\(L\subseteq\Sigma^*\times\N\)
  is \emph{fixed\hyp parameter tractable}
  if there is an algorithm
  deciding~\((x,k)\in L\)
  in time \(f(k)\cdot|x|^{O(1)}\) (we call such an algorithm a \emph{fixed\hyp parameter algorithm}).
  The complexity class~\emph{FPT} consists of all fixed\hyp parameter tractable
  problems.
\end{definition}

\begin{definition}[kernelization]
 \label{def:compression}
 Let $L\subseteq \Sigma^*\times \N$ be a parameterized problem.
 A~\emph{kernelization}
 is an algorithm
 that maps any instance~$(x,k)\in\Sigma^*\times \N$
 to an instance~$(x',k')\in\Sigma^*\times \N$
 in $\poly(|x|+k)$ time
 such that
 \begin{enumerate}[(i)]
  \item $(x,k)\in L \iff (x',k')\in L'$, and
  \item $|x'|+k'\leq f(k)$ for some computable function~\(f\).
  \end{enumerate}
  We call \((x',k')\) the \emph{problem kernel}
  and \(f\) its \emph{size}.
\end{definition}

\noindent
A generalization of problem kernels
are \emph{Turing kernels},
where one is allowed to generate multiple reduced
instances instead of a single one.

\begin{definition}[Turing kernelization]
  Let $L\subseteq \Sigma^*\times \N$ be a
  parameterized problem.
 A~\emph{Turing kernelization} for~$L$
 is an algorithm~$A$ that decides~$(x,k)\in L$
 in polynomial time
 given access to an oracle
 that answers $(x',k')\in L$ in constant time
 for any \((x',k')\in\Sigma^*\times \N\) with
 $|x'|+k\leq f(k)$,
 where $f$~is an arbitrary function
 called the \emph{size} of the Turing kernel.
\end{definition}

\subsection{WK[1]-hardness}
To obtain evidence for the nonexistence even of
Turing kernels of polynomial size,
we employ the recently
introduced concept of WK[1]-hardness~\citep{HKS+15b}.
Parameterized problems
that are WK[1]-hard
do not have problem kernels of polynomial size
unless \unlessPK{} (which would imply
a collapse of the polynomial\hyp time hierarchy),
and are conjectured not to have
Turing kernels of polynomial size either.
We prove WK[1]-hardness of a problem~$L$
by reducing a WK[1]-hard problem to~$L$ using
the following type of reduction.

\begin{definition}[polynomial parameter transformation]
  A \emph{polynomial parameter transformation (PPT)}
  of a parameterized problem~$L\subseteq \Sigma^*\times \N$
  into a parameterized problem~$L'\subseteq \Sigma^*\times\N$
  is an algorithm that
  maps any instance~$(x,k)$
  to an instance~$(x',k')$
  \begin{enumerate}[(i)]
    
  \item in $\poly(|x|+k)$~time such that
  \item $(x,k)\in L\iff (x',k')\in L'$ and
  \item $k'\in \poly(k)$.
 \end{enumerate}
\end{definition}

\subsection{Basic observations}
We may assume our input graph
to be connected due to the following
obviously correct
and linear\hyp time executable
data reduction rule.

\begin{rrule}
 \label{rrule:onecomponent}
 If $G$~has more than one connected component, then delete all but the component containing both~$s$ and~$t$ or return \no{} if such a component does not exist.
\end{rrule}

\section{Parameterizing by the vertex cover number}
\label{sec:apgs}

In this section,
we study the parameterized complexity of \sspAcr{}
with respect to the vertex cover number of the input graph.
The vertex cover number
bounds from above most other known graph parameters \citep{FJR13}
and is therefore a rather large parameter.
It thus comes at no surprise
that \sspAcr{} is fixed\hyp parameter tractable
parameterized by the vertex cover number:
this follows from the fact
that \sspAcr{} is fixed\hyp parameter tractable
parameterized by the treewidth,
which we show in \cref{sec:tw}.

However,
despite the vertex cover number being
one of the largest known graph parameters,
in \cref{sec:vcnokern}, we show that
\sspAcr{} is WK[1]-hard with respect to
the vertex cover number.
In contrast,
in \cref{sec:plankern},
we show that \sspAcr{} does have
a problem kernel with size
polynomial in the vertex cover number
in $K_{r,r}$\ssfree{} graphs for any constant~$r$,
a graph class that comprises, for example,
many road networks.

\subsection{Limits of data reduction}
\label{sec:plannokern}
\label{sec:vcnokern}

In this section,
we show lower bounds on kernel sizes
of \sspAcr{} parameterized by the vertex cover number.
Both of the following results
come at some surprise:
finding a standard shortest \(s\)-\(t\)-path
is easy,
whereas finding
a short secluded path
in general graphs is so hard that
not even preprocessing helps.

\begin{theorem}
  \label{thm:wk1hard}
  Even in bipartite graphs, \sspTsc{}
  is \WK{1}-hard when parameterized by~\(\vc\),
  where
  \(\vc\)~is the vertex cover number
  of the input graph.
\end{theorem}

\begin{remark}
  \label{rem:vckr}
  The hardness results of \cref{thm:wk1hard} also hold
  with respect to the parameter~$\vc+k$:
  a vertex cover contains
  at least $\lfloor k/2\rfloor$~vertices
  of a path with $k$~vertices.
  Thus,
  a polynomial parameter transformation
  of \sspAcr{} parameterized by~$\vc$
  to \sspAcr{} parameterized by~$\vc+k$
  can safely reduce~$k$ so that~$k\leq 2\vc+1$.
  However,
  there is a problem kernel with size polynomial in~$\vc+\ell$:
  we will show in \cref{sssec:fvskell} that
  \sspAcr{} allows for a problem
  kernel with size polynomial in~$\fvs+k+\ell\leq\vc+k+\ell\in O(\vc+\ell)$.
\end{remark}

\noindent
To prove \cref{thm:wk1hard},
we use a polynomial parameter transformation
of the following problem
parameterized by \(k\log n\) \citep{HKS+15b}
into \sspAcr{} parameterized by~$\vc$.

\decprob{\mcclique}%
{\label{prob:mcclique}
  A \(k\)-partite \(n\)-vertex graph~\(G=(V_1,V_2,\dots,V_k,E)\)
  with pairwise non\hyp intersecting independent sets~$V_i$.}%
{Does \(G\)~contain a clique of size~\(k\)?}

\noindent
Our polynomial parameter transformation
of \mcclique{}
into \sspAcr{}
uses the following gadget.

\begin{definition}[\(z\)-binary gadget]
  A \emph{$z$-binary gadget}
  for some power~\(z\) of two
  is a set~\(B=\{u_1,u_2,\dots,u_{2\log z}\}\)
  of vertices.
  We say
  that a vertex~\(v\) is \emph{$p$-connected to~$B$}
  for some~\(p\in\{0,\dots,z-1\}\)
  if \(v\)~is adjacent
  to~\(u_q\in B\) if and only if
  there is a ``1'' in position~\(q\)
  of the string
  that consists of
  the binary encoding of~\(p\)
  followed by its complement.
\end{definition}

\begin{example}
  \looseness=-1
  The binary encoding of~\(5\) followed by its complement
  is \(101010\).
  Thus,
  a vertex~\(v\) is \(5\)-connected
  to an \(8\)-binary gadget~\(\{u_1,\dots,u_6\}\)
  if and only if \(v\)~is adjacent
  to exactly~\(u_1,u_3\), and \(u_5\).
  Also observe that,
  if a vertex~\(v\)
  is \(q\)-connected
  to a \(z\)-binary gadget~\(B\),
  then \(v\)~is adjacent
  to exactly half of the vertices of~\(B\),
  that is,
  to \(\log z\)~vertices of~\(B\).
\end{example}

\noindent
The following reduction from \mcclique{}
to \sspAcr{} is illustrated in \cref{fig:ppt}.

 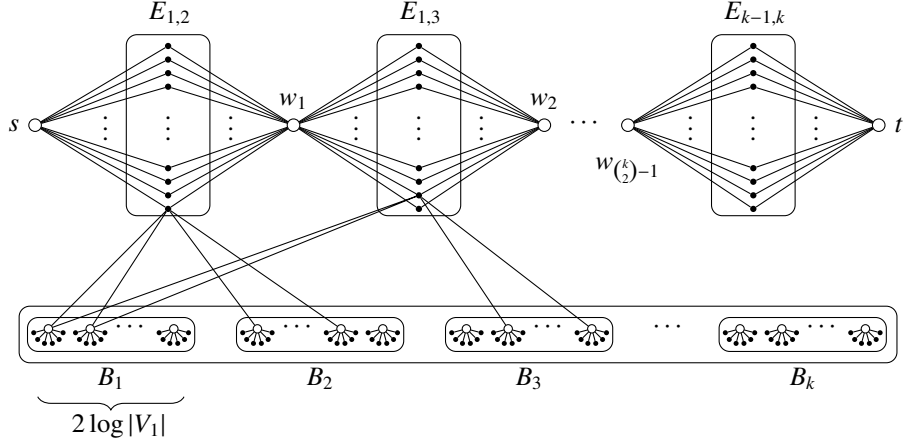
\begin{figure*}
  \centering
   \begin{tikzpicture}
	\usetikzlibrary{decorations.pathreplacing,calc}
	\tikzstyle{pnode}=[fill,circle,scale=1/4]
	\tikzstyle{lnode}=[fill,circle,scale=1/5]
	\tikzstyle{ppnode}=[fill=white,circle,draw,scale=1/3]
      \def\xr{1.1}
      \def\yr{0.6}

    \node (s) at (1,0)[fill=white,circle,draw,scale=1/2,label=180:{$s$}]{};
	\node (x12) at  (4*\xr,0*\yr)[fill=white,draw, circle, scale=1/2,label=90:{$w_1$}]{};
	\node (x13) at  (7*\xr,0*\yr)[fill=white,draw, circle, scale=1/2,label=90:{$w_2$}]{};
	\node (xkk) at  (8*\xr,0*\yr)[fill=white,draw, circle, scale=1/2,label={[label distance=6pt]-90:$w_{{k\choose 2}-1}$}]{};
	\node (t) at  (11*\xr,0*\yr)[fill=white,draw, circle, scale=1/2,label=0:{$t$}]{};

	\newcommand{\edgesets}[6]{
		
		\draw[rounded corners] (#3,#4) rectangle (#3+1*\xr,#4+4*\yr) node[above,xshift=-15pt]{$E_{#1,#2}$};
		\foreach \x in {1,2,3,4,10,11,12,13}{
    			\node (x#1#2x\x) at  (#3+0.5*\xr,\x*0.3*\yr+#4-0.15*\yr)[pnode]{};
    		}
		\node at (#3+0.5*\xr,0.1*\yr)[]{$\vdots$};
		 \foreach \x in {1,2,3,4,10,11,12,13}{
 		   \draw (x#1#2x\x) to (#5);
    		   \draw (x#1#2x\x) to (#6);
    		}
		 \node (txt) at (#3-0.25*\xr,0.1*\yr)[]{$\vdots$};
	     \node (txt) at (#3+1.25*\xr,0.1*\yr)[]{$\vdots$};
	}
    
	\edgesets{1}{2}{2*\xr}{-2*\yr}{s}{x12};
	\edgesets{1}{3}{5*\xr}{-2*\yr}{x12}{x13};
	 \node (txt) at  (7.5*\xr,0*\yr)[]{$\cdots$};
	\edgesets{k-1}{k}{9*\xr}{-2*\yr}{xkk}{t};

	\newcommand{\bstar}[3]{
	
		\def\noleaves{6};
		\def\distleaves{0.2};
	
		\node (#1) at (#2,#3)[ppnode,fill=white]{};
 		 \foreach \j in {1,...,\noleaves}{
			   \node (a\j) at ($(#1)+ (165+\j*180/\noleaves:\distleaves cm)$)[lnode]{};
				\draw (#1) -- (a\j);
 		 }
	}

	\def\sh{0.9}
	\def\bsh{-1} 
    \draw[rounded corners] (\sh -0.1*\xr,-4*\yr) rectangle (\sh+10.4*\xr,-5.25*\yr);
    
	\draw[rounded corners] (\sh +0*\xr,-4.25*\yr) rectangle node[below,yshift=-9pt]{$B_1$} (\sh+2*\xr,-5*\yr);
    \draw[rounded corners] (\sh + 2.5*\xr,-4.25*\yr) rectangle node[below,yshift=-9pt]{$B_2$} (\sh+4.5*\xr,-5*\yr) ;

    \draw[rounded corners] (\sh+5*\xr,-4.25*\yr) rectangle node[below,yshift=-9pt]{$B_3$} (\sh+7*\xr,-5*\yr) ;
    
    \draw[rounded corners] (\bsh + 10*\xr,-4.25*\yr) rectangle node[below,yshift=-9pt]{$B_k$} (\bsh + 12*\xr,-5*\yr);

	\bstar{v11}{\sh+0.25*\xr}{-4.5*\yr};
    \bstar{v12}{\sh+0.75*\xr}{-4.5*\yr};
    \node at (\sh+1.25*\xr,-4.5*\yr){$\cdots$};
    \bstar{v13}{\sh+1.75*\xr}{-4.5*\yr};

    \bstar{v21}{\sh+2.75*\xr}{-4.5*\yr};
    \node at (\sh+3.25*\xr,-4.5*\yr){$\cdots$};
    \bstar{v22}{\sh+3.75*\xr}{-4.5*\yr};
    \bstar{v23}{\sh+4.25*\xr}{-4.5*\yr};

    \bstar{v31}{\sh+5.25*\xr}{-4.5*\yr};    
    \bstar{v32}{\sh+5.75*\xr}{-4.5*\yr};    
    \node at (\sh+6.25*\xr,-4.5*\yr){$\cdots$};
    \bstar{v33}{\sh+6.75*\xr}{-4.5*\yr};    
    
    \node (txt) at  ( 0.55 +8*\xr,-4.5*\yr)[]{$\cdots$};

    \bstar{vk1}{\bsh + 10.25*\xr}{-4.5*\yr};
    \bstar{vk2}{\bsh + 10.75*\xr}{-4.5*\yr};    
    \node at (\bsh + 11.25*\xr,-4.5*\yr){$\cdots$};
    \bstar{vk3}{\bsh + 11.75*\xr}{-4.5*\yr};    

    \draw (v12) -- (x12x1) --  (v11);
    \draw (v22) -- (x12x1) --  (v21);
    \draw (v11) -- (x13x2) --  (v12);
    \draw (x13x2) to (v32);
    \draw (x13x2) to (v33);

    \draw [decorate,decoration={brace,amplitude=5pt},xshift=-4pt,yshift=0pt] (\sh+2*\xr,-6*\yr) -- (\sh+0.25*\xr,-6*\yr) node [black,midway,below,xshift=3pt,yshift=-3pt] {$2\log{|V_1|}$};
    \end{tikzpicture}
    \caption{Illustration of the polynomial parameter transformation. White vertices indicate the vertices in the vertex cover.} 
    \label{fig:ppt}
  \end{figure*}

\begin{construction}
 \label{constr:vcnopk}
  Let \(G=(V_1,V_2,\dots,V_k,E)\)~be an instance
  of \mcclique{}
  with $n$~vertices.
  Without loss of generality,
  assume that
  \(V_i=\{v_i^1,v_i^2,\dots,v_i^{\ncc{n}}\}\)
  for each~$i\in \{1,\dots,k\}$,
  where \(\ncc{n}\)~is some power of two
  (we can guarantee this by adding
  isolated vertices to~\(G\)).
  We construct an equivalent instance~\((G',s,t,k',\ell')\)
  of \sspAcr{},
  where
  \begin{align*}
   k'&:=2\cdot {k\choose 2}+1,&
   \ell'&:=|E|-{k\choose 2}+k \log\ncc{n},
 \end{align*}
 and the graph~\(G'=(V',E')\) is as follows.
 The vertex set~\(V'\)
 consists of
 vertices~\(s\), \(t\),
 a vertex~\(v_e\) for each edge~$e\in E$,
 vertices~\(w_h\)
 for \(h\in\{1,\dots, {k\choose 2}-1\}\),
 and
 mutually disjoint
 $\ncc{n}$-binary vertex gadgets~$B_1,\ldots,B_k$,
 each vertex in which has \(\ell'+1\)~neighbors
 of degree one.
 We denote
 \begin{align*}
   E^*&:=\{v_e\in V'\mid e\in E\},
   &
   B&:=B_1\uplus B_2\uplus\dots\uplus B_k,\\
   E_{ij}&:=\{v_{\{x,y\}} \in E^*\mid x\in V_i,y\in V_j\},\text{\qquad and}
           &
   W&:=\{w_h\mid 1\leq h\leq \tbinom k2-1\}.
 \end{align*}
 The edges of~\(G'\) are as follows.
 For each edge~\(e=\{v_i^p,v_j^q\}\in E\),
 vertex~\(v_e\in E_{ij}\) of~\(G'\)
 is $p$-connected to~$B_i$
 and $q$-connected to~$B_j$.
 Vertex~$s\in V'$ is adjacent
 to all vertices in~$E_{1,2}$
 and
 vertex~$t\in V'$ is adjacent
 to all vertices in~$E_{k-1,k}$.
 Finally,
 to describe the edges
 incident to vertices in~\(W\),
 consider
 the lexicographic ordering of the pairs 
 \(\{(i,j)\mid 1\leq i<j\leq k\}\).
 Then,
 vertex~\(w_h\in W\)
 is adjacent to all
 vertices in~\(E_{ij}\)
 and to all
 vertices in~\(E_{i'j'}\),
 where \((i,j)\) is the \(h\)-th pair
 in the ordering and \((i',j')\)~is the \((h+1)\)-st.
 This finishes the construction.
\end{construction}

\noindent
On our way proving \cref{thm:wk1hard},
we aim to prove that \cref{constr:vcnopk}
is a 
polynomial parameter transformation,
that is,
a polynomial\hyp time many\hyp one reduction
that
generates %
graphs
with sufficiently small vertex covers.

\begin{lemma}
  \label[lemma]{lem:ppt}
  \cref{constr:vcnopk} is a
  polynomial parameter transformation
  of \mcclique{}
  parameterized by \(k\log n\)
  into \sspAcr{}
  parameterized by \(\vc\).
\end{lemma}

\begin{proof}
  Let \(\I':=(G',s,t,k',\ell')\)~be the
  \sspAcr{} instance
  created by \cref{constr:vcnopk} %
  from a \mcclique{} instance~\(G=(V_1,V_2,\dots,V_k,E)\).
  Instance $\I'$~can obviously be computed in polynomial time.
  We show that $\vc\in\poly(k\log n)$.
  The vertex set of~$G'$ partitions into two independent sets
  \begin{align*}
    X=\{s,t\}\cup W\cup B&&\text{and}&& Y=N(B)\cup E^*.
  \end{align*}
  Hence,
  $X$~is a vertex cover of~$G'$.
  Its size is  \(2k\log n+\binom{k}{2}+2\).
  It remains to show that
  $G$~is a \yes-instance
  if and only if
  $\I'$~is.
 
  \RD{}
  Let~$C$
  be the edge set of a clique of size~$k$ in~$G$.
  For each \(1\leq i<j\leq k\),
  \(C\)~contains exactly one edge~\(e\)
  between~\(V_i\) and~\(V_j\).
  Thus,
  \(E_C:=\{v_e\in E^*\mid e\in C\}\)
  is a set of \(\binom k2\)~vertices---exactly
  one vertex of~\(E_{ij}\)
  for each \(1\leq i<j\leq k\).
  Thus,
  by \cref{constr:vcnopk},
  \(G'\)~contains
  an \(s\)-\(t\)-path~\(P=(V_P,E_P)\)
  with \(|V_P|\leq k'\):
  its inner vertices
  are~\(E_C\cup W\),
  alternating between the sets~$E_C$ and~$W$.
  To show that
  \((G',s,t,k',\ell')\)~is a yes\hyp instance,
  it remains to show \(|N(V_P)|\leq\ell'\).
  
  Since \(P\)~contains all vertices of~\(W\),
  one has \(N(V_P)\subseteq B\cup (E^*\setminus E_C)\),
  where
  \(|E^*\setminus E_C|= |E|-\binom k2\).
  To show \(|N(V_P)|\leq\ell'\),
  it remains to show that
  \(|N(V_P)\cap B|\leq k\log \ncc{n}\).
  To this end,
  we show that \(|N(V_P)\cap B_i|\leq\log\ncc{n}\)
  for each \(i\in\{1,\dots,k\}\).

  The vertices in~\(W\cup\{s,t\}\)
  have no neighbors in~\(B\).
  Thus,
  let \(i\in\{1,\dots,k\}\) be fixed and
  consider arbitrary vertices~\(v_{e_1},v_{e_2}\in E_C\)
  such that~\(N(v_{e_1})\cap B_i\ne\emptyset\) and~\(N(v_{e_2})\cap B_i\ne\emptyset\)
  (possibly, \(e_1=e_2\)).
  Then,
  \(e_1=\{v_i^p,v_j^q\}\)
  and
  \(e_2=\{v_{i}^{p'},v_{j'}^{q'}\}\).
  Since \(C\)~is a clique,
  \(e_1\) and \(e_2\)
  are incident to the same vertex of~\(V_i\).
  Thus,
  we have
  \(p=p'\).
  Both \(v_{e_1}\) and \(v_{e_2}\)
  are therefore \(p\)-connected to~\(B_i\)
  and hence have
  the same $\log \ncc{n}$~neighbors in~$B_i$.
  It follows that
  $N(V_P)\leq \ell'$
  and,
  consequently,
  that $\I'$~is a \yes-instance.
 
  \LD{}
  Let \(P=(V_P,E_P)\)~be an $s$-$t$-path in~\(G'\) with
  \(|V_P|\leq k'\) and \(|N(V_P)|\leq\ell'\).
  The path~\(P\)
  does not contain any vertex of~\(B\),
  since each of them has \(\ell'+1\)~neighbors
  of degree one.
  Thus,
  the inner vertices of~\(P\)
  alternate between
  vertices in~\(W\) and in~\(E^*\)
  and we get
  \(N(V_P)= (E^*\setminus V_P)\cup(N(V_P)\cap B)\).
  Since \(P\)~contains
  one vertex of~\(E_{ij}\)
  for each \(1\leq i<j\leq k\),
  we know
  \(|E^*\setminus V_P|= |E|-\binom k2\).
  Thus,
  since \(|N(V_P)|\leq\ell'\),
  we have \(|N(V_P)\cap B|\leq k\log\ncc{n}\).
  We exploit this
  to show that
  the set~\(C:=\{e\in E\mid v_e\in V_P\cap E^*\}\)
  is the edge set of a clique in~\(G\).
  To this end,
  it is enough to show that,
  for each \(i\in\{1,\dots,k\}\),
  any two edges~\(e_1,e_2\in C\)
  with \(e_1\cap V_i\ne\emptyset\)
  and \(e_2\cap V_i\ne\emptyset\)
  have the same endpoint in~\(V_i\):
  then \(C\)~is a set of \(\binom k2\)~edges
  on \(k\)~vertices
  and thus forms a \(k\)-clique.

  For each \(1\leq i<j\leq k\),
  \(P\)~contains exactly
  one vertex~\(v\in E_{ij}\),
  which has exactly $\log\ncc{n}$~neighbors
  in each of~\(B_i\) and~\(B_j\).
  Thus,
  from \(|N(V_P)\cap B|\leq k\log\ncc n\)
  follows
  \(|N(V_P)\cap B_i|=\log\ncc{n}\)
  for each~\(i\in\{1,\dots,k\}\).
  It follows that,
  if two vertices~\(v_{e_1}\)
  and~\(v_{e_2}\) on~\(P\)
  both have neighbors in~\(B_i\),
  then both are \(p\)-connected to~\(B_i\)
  for some~\(p\),
  which means that
  the edges~\(e_1\) and \(e_2\) of~\(G\)
  share endpoint~\(v_i^p\).
  We conclude that
  \(C\)~is the edge set of a clique of size~$k$ in~\(G\).
  Hence,
  \(G\)~is a \yes-instance.
\end{proof}

\noindent
To prove \cref{thm:wk1hard},
it is now a matter of
putting together \cref{lem:ppt}
and the fact that \mcclique{}
parameterized by~\(k\log n\)
is WK[1]-complete.

\begin{proof}[Proof of \cref{thm:wk1hard}]
  By \cref{lem:ppt},
  \cref{constr:vcnopk}
  is a polynomial parameter transformation
  from \mcclique{} parameterized by~\(k\log n\)
  to \sspAcr{} parameterized by~\(\vc\).

  \mcclique{} parameterized by~$k\log n$
  is known to be WK[1]-complete~\citep{HKS+15b}
  and hence,
  does not admit a polynomial\hyp size
  problem kernel unless~$\unlessPK{}$.
  From the
  polynomial parameter transformation
  in \cref{constr:vcnopk},
  it thus follows that
  \sspAcr{} is WK[1]-hard parameterized by~\(\vc\)
  and does not admit a polynomial\hyp size
  problem kernel unless~$\unlessPK{}$,
  either.
\end{proof}

\subsection{Polynomial-size kernels in $K_{r,r}$\ssfree{} graphs}
\label{sec:plankern}
\label{sec:planar-kern}

In this section,
we show that the hardness result
of the previous section
does not transfer to
$K_{r,r}$\ssfree graphs for constant~$r$,
which can be assumed to comprise,
for example,
road networks.

\begin{theorem}
  \label{thm:kernelkrr}
  For each constant~\(r\in\N\),
  \sspTsc{}
  in \(K_{r,r}\)\ssfree{} graphs
  admits a problem kernel
  with size polynomial
  in the vertex cover number
  of the input graph.
\end{theorem}

\noindent
Since the theorem would follow trivially
if \sspAcr{} was polynomial\hyp time solvable
in $K_{r,r}$\ssfree{} graphs,
before proving \cref{thm:kernelkrr},
we first show the following:

\begin{proposition}
  \label{prop:k22-hardness}
  \sspAcr{} in $K_{r,r}$\ssfree{} graphs
  is
  polynomial\hyp time solvable for~$r=1$
  and NP\hyp hard
  for each constant~$r\geq 2$
  even in graphs with maximum degree four
  and $\ell=0$.
\end{proposition}

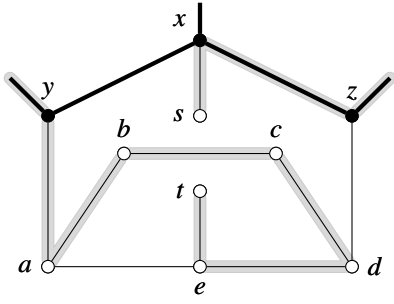
\begin{figure}
  \centering
  \begin{tikzpicture}

    \usetikzlibrary{calc,fadings}
    \def\xr{1}
    \def\yr{1}

    \tikzstyle{xnode}=[circle,fill,scale=1/2,draw];
    \tikzstyle{anode}=[circle,fill=white,scale=1/2,draw];
    \tikzstyle{xedge}=[-];
    \tikzstyle{xxedge}=[ultra thick,-];
    \tikzstyle{pedge}=[line cap=round, double distance = 10\pgflinewidth, double=gray!30!white,gray!40!white];

    \newcommand{\rrnodes}{
    \node (s) at (0,-1*\yr)[anode,label=180:{$s$}]{};
    \node (x) at (0,0)[xnode,label=135:{$x$}]{};
    \node (y) at (-2*\xr,-1*\yr)[xnode,label={$y$}]{};
    \node (z) at (2*\xr,-1*\yr)[xnode,label={$z$}]{};
    \node (a) at (-2*\xr,-3)[anode,label=180:{$a$}]{};
    \node (b) at (-1*\xr,-1.5)[anode,label={$b$}]{};
    \node (c) at (1*\xr,-1.5)[anode,label={$c$}]{};
    \node (d) at (2*\xr,-3)[anode,label=0:{$d$}]{};
    \node (e) at (0*\xr,-3)[anode,label=-90:{$e$}]{};
    \node (t) at (0*\xr,-2)[anode,label=180:{$t$}]{};
    }

    \rrnodes{}
    \draw[pedge] (s) to (x)  to (z) to ($(z)+(0.5*\xr,0.5*\yr)$);
    \draw[pedge]  ($(y)+(-0.5*\xr,0.5*\yr)$) to (y) to (a)  to (b) to (c) to (d) to (e) to (t);
    \rrnodes{}

    \draw[xedge] (s) to (x);
    \draw[xxedge] (y) to (x);
    \draw[xxedge] (z) to (x);
    \draw[xedge] (y) to (a);
    \draw[xedge] (b) to (a);
    \draw[xedge] (b) to (c);
    \draw[xedge] (d) to (c);
    \draw[xedge] (d) to (e);
    \draw[xedge] (d) to (z);
    \draw[xedge] (a) to (e);
    \draw[xedge] (t) to (e);
    \draw[xxedge] (x) to ($(x)+(0,0.5*\yr)$);
    \draw[xxedge] (y) to ($(y)+(-0.5*\xr,0.5*\yr)$);
    \draw[xxedge] (z) to ($(z)+(0.5*\xr,0.5*\yr)$);
  \end{tikzpicture}
  \caption{Construction for the proof of \cref{prop:k22-hardness}. 
  White-colored vertices and thin edges are added to the graph.
  A Hamiltonian cycle is sketched by gray thick lines.}
  \label{fig:k22-hardness}
\end{figure}

\begin{proof}
  A $K_{1,1}$\ssfree{} graph has no edges,
  thus \sspAcr{} in such graphs is trivial.  We now prove that \sspAcr{} is NP\hyp hard in $K_{r,r}$\ssfree{} graphs
  of maximum degree four
  for $r=2$,
  which implies NP\hyp hardness for any $r\geq 2$.
  To this end,
  we present a polynomial\hyp time many\hyp one
  reduction from the NP\hyp complete problem \textsc{Hamiltonian Cycle}  
  in \emph{hexagonal grid graphs} \citep{AFI+09},
  which are subgraphs of a hexagonal grid.
  Note that such a graph
  has maximum degree three,
  does not contain cycles of length four
  as a subgraph, and is therefore $K_{2,2}$\ssfree{}.

  Let $G$~be a hexagonal grid graph and $x$~be an arbitrary vertex of~$G$
  with at least two neighbors~$y$ and~$z$
  (if there is no such vertex,
  then we conclude in polynomial\hyp time that $G$~is
  a no\hyp instance).
  We obtain a graph~$G'$
  by adding vertices~$s,t,a,b,c,d,e$ as shown in \cref{fig:k22-hardness}
  and return an instance~$(G',s,t,k,\ell)$ with $k=n+7$ and~$\ell=0$.
  Observe that~$G'$ does not contain
  a~$K_{2,2}$\hyp subgraph (or, equivalently, a cycle of length four),
  since $G$~does not contain them:
  all cycles that are not in~$G$
  contain both~$a$ and~$d$,
  and the shortest cycle containing them consists of five vertices.

  We now prove that $G$ admits a Hamiltonian cycle if and only if $(G',s,t,k,\ell)$ is a \yes-instance of \sspAcr{}.
  
  $(\Rightarrow)$ Assume that~$G$ has a Hamiltonian cycle~$H$.
  Then at least one of~$y$ and~$z$ is adjacent to~$x$ in~$H$.
  By symmetry,
  assume it is~$y$.
  Then an $s$-$t$-path
  on $n+7$~vertices and
  empty open neighborhood in~$G'$
  starts at~$s$,
  visits~$x$,
  follows~$H$ starting with the neighbor of~$x$ on~$H$ that is not~$y$
  until arriving at~$y$,  and finally ends in~$a,b,c,d,e,t$
  (illustrated by the dashed line in \cref{fig:k22-hardness}).

  $(\Leftarrow)$ Assume that~$G'$ has a simple $s$-$t$-path~$P$
  with at most~$n+7$ vertices and empty open neighborhood.
  Then $P$~contains \emph{all} of the $n+7$~vertices of~$G'$.
  Note that the only entry and exit
  points of the set of vertices~$\{a,b,c,d,e,t\}$
  are~$y$ and~$z$ and,  therefore,
  when~$P$ enters~$\{a,b,c,d,e,t\}$,
  it cannot leave this set of vertices anymore,
  since otherwise it will be impossible for $P$ to reach~$t$.
  Thus, path~$P$ starts with~$s,x$
  and ends with the vertices~$\{a,b,c,d,e,t\}$ (in some order),
  which, modulo symmetry, have been entered via~$y$.
  Thus, removing the vertices~$\{s,a,b,c,d,e,t\}$ from~$P$
  and adding an edge~$\{x,y\}$
  yields a Hamiltonian cycle for~$G$.
\end{proof}

\noindent
Note that \cref{thm:kernelkrr} is trivial for~$r=1$:
we can simply solve the problem in polynomial time (see \cref{prop:k22-hardness})
and return a constant\hyp size equivalent instance.

For $r\geq 2$,
the proof of \cref{thm:kernelkrr}
consists of three steps.
First,
in linear time,
we transform an
$n$-vertex instance of \sspAcr{}
into an equivalent instance
of an auxiliary vertex\hyp weighted version of \sspAcr{}
with $O(\vc^r)$ vertices.
Second,
using a theorem of \citet{FT87},
in polynomial time,
we reduce the vertex weights
to $2^{O(\vc^{3r})}$
so that the length of their
encoding is \(O(\vc^{3r})\).
Finally,
since \sspAcr{} is NP-complete
in~$K_{r,r}$\ssfree{} graphs for~$r\geq 2$
by \cref{prop:k22-hardness},
we can, in polynomial time,
reduce
the shrunk instance
back to an instance
of the unweighted
\sspAcr{} in \(K_{r,r}\)\ssfree{} graphs.
Due to the polynomial running time
of the reduction,
there is
at most a polynomial blow-up
of the instance size.

Our auxiliary variant
of \sspAcr{}
allows each vertex to have three weights:
weight~\(\kappa(v)\)
counts towards the length of the path,
the weights~\(\lambda(v)\) and~\(\eta(v)\)
count towards the number of neighbors
(in fact,
we will not use~\(\eta(v)\) yet,
but to derive other results later).

\decprob{\WsspTsc~(\WsspAcr)}
{%
  \label[problem]{prob:wssp}%
  An undirected, simple graph~$G=(V,E)$ with two distinct vertices~$s,t\in V$, two integers~$k\geq2$ and~$\ell\geq0$, and vertex weights~$\kappa:V\to\N$,~$\lambda:V\to\N\cup\{0\}$, and~$\eta:V\to\N\cup\{0\}$.}
{Is there a simple \(s\)-\(t\)-path~$P$ with
  $\sum_{v\in V(P)}\kappa(v)\leq k$ and $\sum_{v\in V(P)} \eta(v)+\sum_{v\in N(V(P))} \lambda(v)\leq \ell$ in~\(G\)?}

\noindent
Note that an instance of \sspAcr{}
can be considered
to be an instance of \WsspAcr{}
with unit weight functions~\(\kappa\) and~\(\lambda\)
and the zero weight function~$\eta$.
Our data reduction will be based
on removing \emph{twins}.
\begin{definition}[twins]
  Two vertices~\(u\) and~\(v\)
  are called \emph{(false) twins}
  if \(N(u)=N(v)\).
\end{definition}

\noindent
As the first step
towards proving \cref{thm:kernelkrr},
we will show
that the following data reduction rule,
when applied to a \(K_{r,r}\)\ssfree{} instance
of \sspAcr{} for constant~\(r\),
leaves us with
an instance of \WsspAcr{}
with \(O(\vc^r)\)~vertices.

\begin{rrule}
 \label{rr:reducenbs}
 Let \((G,s,t,k,\ell,\kappa,\lambda,\eta)\)
 be an \WsspAcr{} instance with unit weights~\(\kappa\)
 and~\(\lambda\), and zero weights~\(\eta\),
 where \(G=(V,E)\)~is a \(K_{r,r}\)\ssfree{} graph.
 
 For each maximal set~\(U\subseteq V\setminus\{s,t\}\)
 of twins such that \(|U|>r\),
 delete $|U|-r$ vertices of~$U$ from~\(G\),
 and, for an arbitrary remaining vertex~\(v\in U\),
 set $\lambda(v):=|U|-r+1$ and $\kappa(v):=k+1$.
\end{rrule}

\begin{lemma}
  \label[lemma]{obs:reducenbslintime}
  \cref{rr:reducenbs} is correct
  and can be applied in linear time.
\end{lemma}

\begin{proof}
  All maximal sets of twins
  can be computed in linear time \citep{HPV98}.
  It is now easy to check
  which of them has size larger than~\(r\)
  and to apply \cref{rr:reducenbs}.

  To prove that \cref{rr:reducenbs} is correct,
  we prove that its
  input instance~\(I=(G,s,t,k,\ell,\kappa,\lambda,\eta)\)
  is a yes\hyp instance
  if and only if
  its
  output instance~\(I'=(G',s,t,k,\ell,\kappa',\lambda',\eta)\)
  is.
  Herein, note that \(\eta\)~is the zero function,
  so we will ignore it in the rest of the proof.

  \RD{}
  Let $P$~be a simple $s$-$t$-path
  with \(\sum_{v\in V(P)}\kappa(v)\leq k\)
  and \(\sum_{v\in N(V(P))}\lambda(v)\leq\ell\)
  in~\(G\).
  Let $U\subseteq V\setminus \{s, t\}$
  be an arbitrary set of twins
  with \(|U|>r\).
  Since $G$~is $K_{r,r}$\ssfree{},
  $|N(U)| \leq r-1$.
  Thus,
  \(P\)~contains at most
  $|N(U)|-1\leq r-2$~vertices of~$U$.
  \cref{rr:reducenbs} reduces $U$~to a set~$U'$
  with $r$~vertices, where only
  one of the vertices~$v\in U'$
  has weight~\(\kappa'(v)>1\).
  Thus,
  without loss of generality,
  we can assume that \(P\)~uses
  only the \(r-1\)~vertices~\(v\in U\cap U'\)
  with \(\kappa'(v)=1\).
  Hence,
  \begin{equation}
    \sum_{v\in V(P)\cap U}\kappa(v)
    =
    \sum_{v\in V(P)\cap U}\kappa'(v)
    =|V(P)\cap U|.
    \label{kappas1}
  \end{equation}
  Moreover,
  if \(P\)~uses a vertex of~\(U\),
  then it also uses a vertex of~\(N(U)\)
  and, hence,
  \(U\setminus V(P)\subseteq N(V(P))\).
  Thus,
  \begin{equation}
    \begin{aligned}
      \sum_{v\in N_G(V(P))\cap U}\lambda(v)
      &=    \smashoperator{\sum_{v\in U\setminus V(P)}}\lambda(v)
      =    \smashoperator{\sum_{v\in U'\setminus V(P)}}\lambda'(v)
      = \sum_{v\in N_{G'}(V(P))\cap U'}\lambda'(v)
 \end{aligned}
    \label{lambdas1}
  \end{equation}
  since
  \(|U\setminus U'|=|U|-r\)
  and
  there is a vertex~\(v\in U'\cap U\)
  that has
  \(\lambda(v)=1\)
  on the left\hyp hand side of \eqref{lambdas1}
  but \(\lambda'(v)=|U|-r+1\)
  on the right\hyp hand side of \eqref{lambdas1}.
  From \eqref{kappas1}, \eqref{lambdas1},
  and the arbitrary choice of~\(U\),
  it follows that
  \(P\)~is an $s$-$t$-path
  with \(\sum_{v\in V(P)}\kappa'(v)\leq k\)
  and \(\sum_{v\in N(V(P))}\lambda'(v)\leq\ell\)
  in~\(G'\).
  Thus,
  \(I'\)~is a yes\hyp instance.

  \LD{}
  Let $P$~be a simple $s$-$t$-path
  with \(\sum_{v\in V(P)}\kappa'(v)\leq k\)
  and \(\sum_{v\in N(V(P))}\lambda'(v)\leq\ell\)
  in~\(G'\).
  Let $U\subseteq V\setminus \{s, t\}$
  be a set of twins in~\(G\)
  reduced to a subset~\(U'\)
  in~\(G'\)
  by \cref{rr:reducenbs}.
  The only vertex~$v\in U'$
  with weight~\(\kappa'(v)>1=\kappa(v)\)
  has \(\kappa'(v)=k+1\)
  and thus is not on~\(P\).
  Yet,
  if \(P\)~uses vertices of~\(U'\),
  then \(v\in U'\setminus V(P)\subseteq N_{G'}(V(P))\)
  and \(U\setminus V(P)\subseteq N_G(V(P))\).
  Thus,
  \eqref{kappas1} and \eqref{lambdas1} apply
  and,
  together with the arbitrary choice of~\(U\),
  show that
  \(P\)~is an $s$-$t$-path
  with \(\sum_{v\in V(P)}\kappa(v)\leq k\)
  and \(\sum_{v\in N(V(P))}\lambda(v)\leq\ell\)
  in~\(G\)
  and, thus,
  \(I\)~is a yes\hyp instance.
\end{proof}

\noindent
Having proved the correctness of \cref{rr:reducenbs},
we now prove a size bound
for the instances
remaining after \cref{thm:kernelkrr}.

\begin{proposition}
  \label[proposition]{lem:instreduce}
  Applied to an instance
  of \sspAcr{}
  with a \(K_{r,r}\)\ssfree{} graph
  with vertex cover number~\(\vc\),
  \cref{rr:reducenbs,rrule:onecomponent} yield
  an instance of
  \WsspAcr{} on at most $(\vc+2) + r(\vc +2)^r$~vertices in linear time.
\end{proposition}

\begin{proof}
  Let \((G',s,t,k,\ell,\lambda',\kappa',\eta)\)~be
  the instance obtained from applying
  \cref{rr:reducenbs,rrule:onecomponent}
  to an instance \((G,s,t,k,\ell,\lambda,\kappa,\eta)\).

  Let \(C\)~be
  a minimum\hyp cardinality vertex cover for~\(G'\)
  that contains~\(s\) and~\(t\),
  let the vertex set of~\(G'\)
  be~$V$, and let~$Y=V\setminus C$.
  Since \(G'\)~is
  a subgraph of~\(G\),
  one has
  \(|C|\leq\vc(G')+2\leq\vc(G)+2=\vc+2\).
  It remains to bound~\(|Y|\).
  To this end,
  we bound the number of vertices
  of degree at least~\(r\) in~\(Y\)
  and the number of vertices
  of degree exactly~\(i\) in~\(Y\)
  for each \(i\in\{0,\dots,r-1\}\).
  Note
  that vertices in~\(Y\)
  have neighbors only in~\(C\).
  
  Since \cref{rrule:onecomponent} has been applied,
  there are no vertices of degree zero in~\(Y\).

  Since \cref{rr:reducenbs}
  has been applied,
  for each~\(i\in\{1,\dots,r-1\}\)
  and each subset~\(C'\subseteq C\)
  with \(|C'|=i\),
  we find at most \(r\)~vertices in~\(Y\)
  whose neighborhood is~\(C'\).
  Thus,
  for each~\(i\in\{1,\dots,r-1\}\),
  the number of vertices with degree~\(i\)
  in~\(Y\) is at most
  \(
    r\cdot \binom{|C|}{i}.
  \)

  Finally,
  since \(G\)~is \(K_{r,r}\)\ssfree{},
  any $r$-sized subset of the vertex cover~$C$
  has at most $r-1$~common neighbors.
  Hence,
  since vertices in~\(Y\)
  have neighbors only in~\(C\),
  the number of vertices in~$Y$
  of degree greater or equal to~$r$
  is at most \((r-1)\cdot\binom{|C|}{r}.\)
  We conclude that
  \begin{align*}
    |V'|&\leq |C| + (r-1)\cdot\binom{|C|}{r} + r\cdot \sum_{i=1}^{r-1} \binom{|C|}{i}\le (\vc+2) + r(\vc +2)^r.\qedhere
  \end{align*}
\end{proof}

\noindent
Having shown how to reduce an instance of \sspAcr{}
on \(K_{r,r}\)\ssfree{} graphs
to an equivalent instance of \WsspAcr{}
on \(O(\vc^r)\)~vertices
for constant~\(r\),
we finished the first step
to proving \cref{thm:kernelkrr}.
However,
our data reduction works
by ``hiding'' an unbounded number of twins
in vertices of unbounded weights.
Therefore,
the second step on the proof of \cref{thm:kernelkrr}
is reducing the weights.

To reduce the weights of an \WsspAcr{} instance,
we are going to apply a theorem by \citet{FT87}.
The theorem is a key approach to
polynomial\hyp size kernels for weighted problems
\citep{EKMR17}.
Notably,
we are seemingly the first ones
to apply the theorem of \citet{FT87}
to eventually kernelize
an \emph{unweighted} problem---\sspAcr{}.

\begin{proposition}[\citet{FT87}]
  \label[proposition]{thm:FrankTardos}
  There is an algorithm that,
  on input~$w\in\Q^d$ and integer~$N$,
  computes in polynomial time
  a vector~$\bar{w}\in \Z^d$
  with~$\norm{\bar{w}}{\infty}\leq 2^{4d^3}N^{d(d+2)}$
  such that~$\sign(w^\top b)=\sign(\bar{w}^\top b)$
  for all~$b\in\Z^d$ with~$\norm{b}{1}\leq N-1$,
  where
  \[
    \sign(x)=
    \begin{cases}
      +1&\text{if $x>0$},\\
      \phantom{+}0&\text{if $x=0$, and}\\
      -1&\text{if $x<0$}.
    \end{cases}
  \]
\end{proposition}

\begin{observation}
  \label[observation]{obs:nonneg}
  \looseness=-1
  For \(N\geq 2\),
  \cref{thm:FrankTardos}
  gives
  \(\sign(w^\top e_i)=\sign(\bar w^\top e_i)\)
  for each~\(i\in\{1,\dots,d\}\),
  where \(e_i\in \Z^d\) is the vector that has 1 in the \(i\)-th coordinate
  and zeroes in the others.
  Thus,
  one has \(\sign(w_i)=\sign(\bar w_i)\) for each \(i\in\{1,\dots,d\}\).
  That is,
  when reducing a weight vector from~\(w\) to~\(\bar w\),
  \cref{thm:FrankTardos} maintains the signs of weights.
\end{observation}

\noindent
We apply \cref{thm:FrankTardos,obs:nonneg}
to the weights of \WsspAcr{}.

\begin{lemma}
\label[lemma]{lem:weightreduce}
An  instance~$I=(G, s,t,k,\ell,\lambda, \kappa,\eta)$
of \WsspAcr{}
on an \(n\)-vertex graph~\(G=(V,E)\)
can be reduced in polynomial time to
an instance~$I'=(G,s,t,k',\ell', \lambda', \kappa',\eta')$
of \WsspAcr{}
such that
\begin{enumerate}[i)]
\item\label{shrink1} $\{k',\kappa'(v), \ell', \lambda'(v),\eta'(v)\}\subseteq\{0,\dots, 2^{4(2n+1)^3}\cdot (n+2)^{(2n+1)(2n+3)}\}$, for each vertex~$v\in V$, and
\item\label{shrink2} $I$ is a \yes-instance if and only if $I'$ is a \yes-instance.
\end{enumerate}
\end{lemma}

\begin{proof}
  In this proof,
  we will conveniently
  denote the weight functions~\(\lambda,\lambda',\kappa,\kappa',\eta\), and \(\eta'\)
  as column vectors in~\(\N^n\)
  such that \(\lambda_v=\lambda(v)\)
  for each \(v\in V\),
  and similarly for the other weight functions.

  We apply \cref{thm:FrankTardos}
  with \(d=2n+1\) and \(N=n+2\)
  separately
  to the vectors~$(\eta,\lambda, \ell) \in \N^{2n+1}$
  and~$(\kappa,\{0\}^n, k) \in \N^{2n+1}$
  to obtain vectors $(\eta',\lambda',\ell') \in \Z^{2n+1}$ and $(\kappa',\{0\}^n,k')\in\Z^{2n+1}$
  in polynomial time.

  \eqref{shrink1} This follows from \cref{thm:FrankTardos}
  with \(d=2n+1\)
  and \(N=n+2\),
  and from \cref{obs:nonneg}
  since \((\eta,\lambda,\ell)\)
  and \((\kappa,\{0\}^n,k)\) are vectors of nonnegative numbers.

  \eqref{shrink2}
  Consider an arbitrary
  $s$-$t$-path $P$ in~$G$
  and two associated vectors~\(x,y\in\Z^n\),
  where
  \begin{align*}
    x_v&=
    \begin{cases}
      1&\text{ if $v\in V(P)$,}\\
      0&\text{ otherwise,}
    \end{cases}
       &
          y_v&=
    \begin{cases}
      1&\text{ if $v\in N(V(P))$ and}\\
      0&\text{ otherwise.}
    \end{cases}
  \end{align*}
  Observe that $\|(x,y,-1)\|_1 \le n+1$
  and $\|(x,\{0\}^n,-1)\|_1 \le n+1$.
  Since \(n+1\leq N-1\),
  \cref{thm:FrankTardos}
  gives
  \begin{align*}
  \sign((x,y,-1)^\top(\eta,\lambda,\ell))	&=\sign((x,y,-1)^\top(\eta',\lambda',\ell')) \text{\quad and}\\
  \sign((x,\{0\}^n,-1)^\top(\kappa,\{0\}^n, k))	&=\sign((x,\{0\}^n,-1)^\top(\kappa',\{0\}^n, k')),
  \end{align*}
  which is equivalent to
  \begin{align*}
    \smashoperator{\sum_{v\in V(P)}} \eta(v) + \sum_{v\in N(V(P))} \lambda(v) \leq \ell
    &\iff \smashoperator{\sum_{v\in V(P)}} \eta'(v) + \sum_{v\in N(V(P))} \lambda'(v) \leq \ell' \text{\quad and}\\
    \smashoperator{\sum_{v\in P}} \kappa(v) \leq k
    &\iff \smashoperator{\sum_{v\in P}} \kappa'(v) \leq k'.\qedhere
  \end{align*}
\end{proof}

\noindent
We have finished two steps towards
the proof of \cref{thm:kernelkrr}:
we reduced \sspAcr{} in \(K_{r,r}\)\ssfree{}
graphs for constant~\(r\)
to instances of \WsspAcr{} with \(O(\vc^r)\)~vertices
using \cref{lem:instreduce}
and shrunk
its weights to encoding\hyp length
\(O(\vc^{3r})\) using \cref{lem:weightreduce}.
To finish the proof of \cref{thm:kernelkrr},
it remains to reduce
\WsspAcr{} back to \sspAcr{}
on \(K_{r,r}\)\ssfree{} graphs.

\begin{proof}[Proof of \cref{thm:kernelkrr}]
  For $r=1$,
  the problem is solvable in polynomial time (see \cref{prop:k22-hardness})
  and thus has a kernel of constant size.
  Henceforth,
  assume that~$r\geq 2$.
  Using
  \cref{lem:instreduce}
  and
  \cref{lem:weightreduce},
  we reduce any \sspAcr{} instance~\(I\)
  on a \(K_{r,r}\)\ssfree{} \(n\)-vertex graph
  for constant~\(r\)
  with vertex cover number~\(\vc\)
  to an equivalent \WsspAcr{} instance~\(I'\)
  on \(O(\vc^r)\)~vertices
  whose weights are bounded by \(2^{O({\vc}^{3r})}\).
  Thus,
  the overall encoding length of~\(I'\)
  is \(O({\vc}^{4r})\).
  Since \sspAcr{} is NP-complete in~$K_{r,r}$\ssfree{}
  graphs by \cref{prop:k22-hardness},
  we can in polynomial time reduce~$I'$
  to an equivalent instance~\(I^*\)
  of \sspAcr{} on \(K_{r,r}\)\ssfree{} graphs.
  Since the running time
  of the reduction is polynomial,
  the size of~\(I^*\) is polynomial
  in the size of~\(I'\) and, hence,
  polynomial in~\(\vc\).
\end{proof}

\noindent
Finally,
observe that
we indeed have shown polynomial\hyp size problem kernels
for \sspAcr{} parameterized by~$\vc$ in $K_{r,r}$-free graphs
only for constant~$r$,
since $r$~appears in the degree of the size polynomial.
Note that,
unless \(\unlessPK\),
we cannot show a problem kernel with a size polynomial in both~$\vc$ and~$r$:

\begin{remark}
  \label{rem:vc+r}
  Since every graph is $K_{r,r}$\ssfree{} for $r > \vc$, from
  \cref{thm:wk1hard} it follows, that for \sspAcr{} in
  $K_{r,r}$\ssfree{} graphs, there is no problem kernel with size
  polynomial in~\(\vc+r\) unless \(\unlessPK\).
\end{remark}

\section{Graphs with small treewidth}
\label{sec:tlgs}

In this section,
we study \sspAcr{}
in graphs with small \emph{treewidth}
(formally defined below).
In \cref{sec:tw},
we first show an algorithm for \sspAcr{}
parameterized by treewidth.
We then prove in \cref{sec:notwkern}
that \sspAcr{} does not allow
for problem kernels with size
polynomial in treewidth.
Finally,
in \cref{sec:subexp},
we show how our fixed\hyp parameter algorithm
for treewidth can be used
to obtain subexponential\hyp time algorithms
for \sspAcr{}
in restricted graph classes,
for example almost planar
graphs like road networks.

\noindent

\subsection{A fixed\hyp parameter algorithm}
\label{sec:tw}

\looseness=-1
In this section,
we present a fixed\hyp parameter algorithm
for \sspAcr{}
parameterized by treewidth.
Before describing the algorithm,
we  introduce
the treewidth concept.
We roughly follow
the notation for tree decompositions
of \citet{BCKN15},
since we will be using some of their
results to make our algorithm
run in single\hyp exponential time.

\begin{definition}[tree decomposition, treewidth]
  \label[definition]{def:tw}
  A \emph{tree decomposition}~\(\mathbb T=(T,\beta)\)
  of a graph~\(G=(V,E)\)
  consists of a tree~\(T\)
  and a function~\(\beta\colon V(T)\to 2^V\)
  that associates each \emph{node}~\(x\)
  of the tree~\(T\)
  with a subset~\(B_x:=\beta(x)\subseteq V\),
  called a \emph{bag}, such that
  \begin{enumerate}[i)]
  \item\label{tw1} for each vertex~\(v\in V\),
    there is a node~\(x\) of~\(T\)
    with \(v\in B_x\),
    
  \item\label{tw2} for each edge~\(\{u,v\}\in E\),
    there is a node~\(x\) of~\(T\)
    with \(\{u,v\}\subseteq B_x\),
  \item\label{tw3} for each~\(v\in V\)
    the nodes~\(x\) with \(v\in B_x\)
    induce a subtree of~\(T\).
  \end{enumerate}
  The \emph{width of~\(\mathbb T\)} is
  \(w(\mathbb T):=\max_{x\in V(T)}|B_x|-1\).
  The \emph{treewidth} of~\(G\)
  is \(\tw(G):=\min\{w(\mathbb T)\mid \text{$\mathbb T$ is a tree decomposition of~$G$}\}\).
\end{definition}

\noindent
In this section,
we will prove the following result.

\begin{theorem}
  \label{thm:twsingexp}
  \sspTsc{} is solvable
  in \(2^{O(\tw)}\cdot \ell^2\cdot n\)~time
  in graphs of treewidth~\(\tw\).
\end{theorem}

\noindent
To prove \cref{thm:twsingexp},
we first need to compute a tree decomposition
of the input graph.
\citet{BDD+16}
proved that a tree decomposition
of width~\(O(\tw(G))\)
of a graph~\(G\)
is computable in
\(2^{O(\tw)}\cdot n\)-time.
Applying the following \cref{prop:wssptw}
to such a tree decomposition
yields \cref{thm:twsingexp}:

\begin{proposition}
  \label{prop:wssptw}
  \WsspTsc{} is solvable
  in \(n\cdot\ell^2\cdot\tw^{O(1)}\cdot (9+2^{(\omega+3)/2})^{\tw}\text{ time}\)
  when a tree decomposition of width~\(\tw\) is given,
  where \(\omega<2.2373\)
  is the matrix multiplication exponent.
\end{proposition}

\noindent
To prove \cref{thm:twsingexp},
it thus remains to prove \cref{prop:wssptw}.
Note that \cref{prop:wssptw}
actually solves the weighted
problem \WsspAcr{} (\cref{prob:wssp}),
where the term \(\ell^2\)
is only pseudo\hyp polynomial for \WsspAcr{}.
It is a true polynomial for \sspAcr{}
since we can assume \(\ell\leq n\).

\subsubsection{Assumptions on the tree decomposition}
Our algorithm for \cref{prop:wssptw}
will work on
the following simplified kind of tree decomposition,
which can be obtained
from a classical tree decomposition
of width~\(\tw\)
in \(n\cdot \tw^{O(1)}\)~time
without increasing its width \citep{BCKN15}.

\begin{definition}[nice tree decomposition]
  A \emph{nice tree decomposition}~\(\mathbb T\)
  is a tree decomposition
  rooted at one bag~\(r\)
  and in which each bag
  is of one of
  the following types.

  \begin{description}
  \item[Leaf node:] a leaf~\(x\) of \(\mathbb T\)
    with~\(B_x=\emptyset\).
  \item[Introduce vertex node:]
    an internal node~\(x\) of~\(\mathbb T\)
    with one child~\(y\)
    such that \(B_x=B_y\cup\{v\}\)
    for some vertex~\(v\notin B_y\).
    This node is said to \emph{introduce vertex~\(v\)}.
  \item[Introduce edge node:]
    an internal node~\(x\) of~\(\mathbb T\)
    labeled with an edge \(\{u,v\}\in E\)
    and with one child~\(y\)
    such that \(\{u,v\}\subseteq B_x=B_y\).
    This node is said to
    \emph{introduce edge~\(\{u,v\}\)}.
  \item[Forget node:]
    an internal node~\(x\) of~\(\mathbb T\)
    with one child~\(y\)
    such that \(B_x=B_y\setminus\{v\}\)
    for some node~\(v\in B_y\).
    This node is said to \emph{forget~\(v\)}.
    
  \item[Join node:]
    an internal node~\(x\) of~\(\mathbb T\)
    with two children~\(y\) and~\(z\)
    such that \(B_x=B_y=B_z\).
  \end{description}
\end{definition}
\noindent
We additionally require that
each edge is introduced
exactly once
and  make the following,
problem specific assumptions
on tree decompositions.
\begin{assumption}
  \label{ass:st}
  When solving \WsspAcr{},
  we will assume that
  the source~\(s\) and destination~\(t\)
  of the sought path
  are contained in
  all bags of the tree decomposition
  and that the root bag contains only~\(s\) and~\(t\).
  This ensures that
  \begin{itemize}
  \item every bag contains vertices of the sought solution, and that
  \item \(s\) and~\(t\) are never forgotten or introduced.
  \end{itemize}
  Such a tree decomposition
  can be obtained from a nice tree decomposition
  by adding \(s\) and \(t\) to all bags
  and adding forget nodes above the root node
  until one arrives at a new root containing only~$s$ and~$t$.
  This will increase
  the width of the tree decomposition
  by at most two.
\end{assumption}

\noindent
Our algorithm will be based
on computing partial solutions
for subgraphs corresponding to
a node of a tree decomposition
by means of
combining partial solutions
for the subgraphs
corresponding to its children.
Formally, these subgraphs
are the following.

\begin{definition}[subgraphs induced by a tree decomposition]
  Let \(G=(V,E)\)~be a graph
  and \(\mathbb T\)~be
  a nice tree decomposition for~\(G\)
  with root~\(r\).
  Then,
  for any node~\(x\) of~\(\mathbb T\),
  \begin{align*}
    V_x&:=\{v\in V\mid v\in B_y\text{ for a descendant $y$ of~$x$}\}\text{, and}\\
    G_x&:=(V_x,E_x)\text{, where }\\
    E_x&=\{e\in E\mid
    e\text{ is introduced in a descendant of~$x$}\}.
  \end{align*}
  Herein,
we consider each node~\(x\) of~\(\mathbb T\)
to be a descendent of itself.
\end{definition}
\noindent
Having defined subgraphs
induced by subtrees,
we can define partial solutions
in them.
\subsubsection{Partial solutions}

\begin{figure*}[t]
 \centering
  \begin{tikzpicture}[y=0.75cm]

    \usetikzlibrary{decorations.pathreplacing,calc}
    \tikzstyle{tnode}=[circle, fill, scale=1/2,draw]
    \tikzstyle{uedge}=[gray, very thick]
    \tikzstyle{ledge}=[ultra thick]
    \def\xr{1}
    \def\yr{1}

	\draw[densely dotted,thick] (-6.75*\xr,0) -- (6.75*\xr,0);    
	\node at (6.25*\xr,0.5*\yr){$G-V_x$};
	\node at (6.25*\xr,-0.5*\yr){$G_x$};
    \begin{scope}
      \clip  (0,0) ellipse (5*\xr cm and 1*\yr cm);
      \draw[thick,fill=white]  (0,0) ellipse (5*\xr cm and 1*\yr cm);
      \foreach \x in {-2.5,-0.6, 1.3, 3.2}{ \draw (\x*\xr,-3*\yr) -- (\x*\xr,4*\yr); }	
      \node at (-2.8*\xr, 0.75*\yr){$\Dz$};
      \node at (-0.9*\xr,0.75*\yr){$\De$};
      \node at ( 1*\xr, 0.75*\yr){$\Di$};
      \node at (2.9*\xr, 0.75*\yr){$N$};
    \end{scope}
		
    \node (s) at (-4.5*\xr,0*\yr)[tnode,label=180:{$s$}]{};
    \node (x0) at (-4*\xr,0*\yr)[tnode]{};
    \node (x1) at (-3.4*\xr,0*\yr)[tnode]{};
    \node (xo1) at (-2.2*\xr,2.1*\yr)[tnode]{};
    \node (xo2) at (-2*\xr,1.5*\yr)[tnode]{};

    \node (x22) at (-2*\xr,0*\yr)[tnode]{};
    \node (t) at (-1.5*\xr,0*\yr)[tnode,label=0:{$t$}]{};

    \node (xo3) at (0.9*\xr,-2.0*\yr)[tnode]{}{};
    \node (xo4) at (0.95*\xr,-1.5*\yr)[tnode]{}{};
    \node (x31) at (0*\xr, 0*\yr)[tnode]{};
    \node (x32) at (0.7*\xr, 0*\yr)[tnode]{};

    \node (x41) at (2*\xr,0.5*\yr)[tnode]{};
    \node (x42) at (2*\xr,-0.5*\yr)[tnode]{};

    \node (x51) at (4*\xr,0.33*\yr)[tnode]{};
    \node (x52) at (4*\xr,-0.33*\yr)[tnode]{};
 
    \draw[uedge] (s) to [out=90,in=180](-4.3*\xr, 1.8*\yr) to [out=0,in=140](x0);
    \draw[uedge] (x0) to [out=40,in=180](-3.5*\xr, 1.8*\yr) to [out=0,in=140](x1);
	\draw[uedge] (x1)to [out=50,in=180](-2.5*\xr, 2.3*\yr) to [out=-10,in=130](xo1);
	\draw[uedge] (xo1) to [out=-55,in=110](xo2);
	\draw[uedge] (xo2) to [out=-80,in=90](x22);
	\draw[ledge] (x22) to  [out=-100,in=180](-0.2*\xr,-3.3*\yr) to [out=0,in=-110](xo3);
	\draw[ledge] (xo3) to (xo4);
	\draw[ledge] (xo4) to [out=90,in=0](x32);
	\draw[ledge] (x32) to  [out=-130,in=0](-0.2*\xr,-2.1*\yr) to [out=170,in=-30](x31);
	
	\draw[ledge] (x31) to  [out=-150,in=10](-0.7*\xr,-1.8*\yr) to [out=180,in=-90](t);
    \draw[-] (xo4) to [out=0,in=-135](x42);
    \draw[-] (xo2) to [out=0,in=90](x41);
    \draw[-,dashed] (xo1) to [out=0,in=135](x51);
    \draw[-,dashed] (xo3) to [out=0,in=-90](x52);

  \end{tikzpicture}
  \caption{Illustration of a partial solution: the
    thick edges are an overall solution,
    where the darker edges are the part of the solution in~\(G_x\).  The dashed edges are forbidden to exist.}
 \label{fig:partialsol}
\end{figure*}
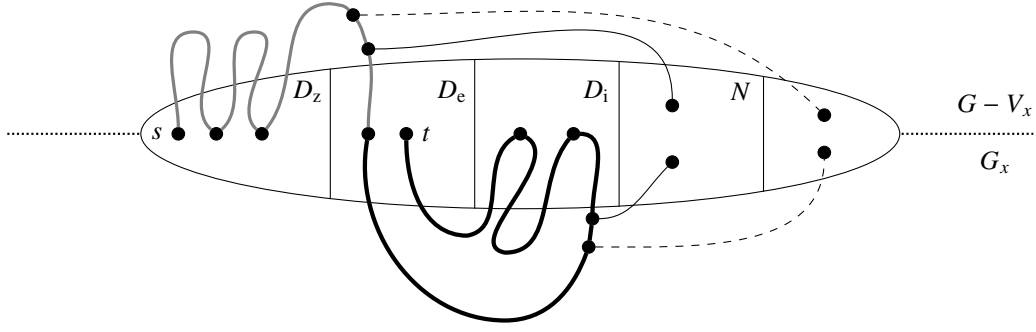

Assume that we have a solution path~\(P\)
to \WsspAcr{}.
Then, the part of~\(P\) in~\(G_x\)
is a collection~\(\pts\) of paths
(some might consist of a single vertex,
that is,
have length zero).
When computing a partial solution
for a parent~\(y\) of~\(x\),
we ideally want to check
which partial solutions for~\(x\)
can be continued to partial solutions for~\(y\).
However,
we cannot try all possible
partial solutions for~\(G_x\)%
---there might be too many.
Moreover,
this is not necessary:
by \cref{def:tw}\eqref{tw2}--\eqref{tw3},
vertices in bag~\(B_y\)
cannot be vertices of
and cannot have edges
to vertices of~\(V_x\setminus B_x\).
Thus,
it is enough to know the
states of vertices in bag~\(B_x\)
in order to know
which partial solutions of~\(x\)
can be continued to~\(y\).
The state of such vertices
is characterized by
\begin{itemize}
\item which vertices of~\(B_x\) are end points of paths in~\(\pts{}\), inner vertices of paths in~\(\pts{}\), or paths of zero length in~\(\pts{}\),
  
\item which vertices of~\(B_x\) are allowed
  to be neighbors of the solution path~\(P\),
\item how many neighbors the solution path~\(P\)
  is allowed to have in~\(G_x\), and
\item which vertices of~\(B_x\)
  belong to the same path of~\(\pts{}\).
\end{itemize}

\noindent
We formalize this as follows.

\begin{definition}[partial solution]
  \label[definition]{def:partsol}
  Let \((G,s,t,k,\ell,\kappa,\lambda,\eta)\)~be
  an instance of~\(\WsspAcr{}\),
  \(\mathbb T\)~be a tree decomposition for~\(G\),
  \(x\)~be a node of~\(\mathbb T\),
  \(\Dz\uplus \De\uplus \Di\uplus \Ne\subseteq B_x\)
  such that \(\{s,t\}\subseteq \Dz\cup\De\),
  \(\prtn\)~be a partition
  of~\(\De\),
  and \(l\leq \ell\).
  Then, we call \(S=(\Dz,\De,\Di,N,l)\)
  a \emph{pre-signature} and
  \((S,\prtn{})\)
  a \emph{solution signature} at~\(x\).
  For a set~\(\pts\) of paths in~\(G_x\)
  and a set~\(N\subseteq B_x\),  let
  \begin{align*}
    V(\pts)&:=\bigcup_{P\in\pts}V(P),&&&
    N_x(\pts)&:=V_x\cap N(V(\pts)),\\
    \nei_x(\pts,N)&:=\smashoperator{\sum_{v\in N_x(\pts)\cup N}}\lambda(v)+\smashoperator{\sum_{v\in V(\pts)}}\eta(v),
    &\text{ and }
      &&
    \cost(\pts)&:=\smashoperator{\sum_{v\in V(\pts)}}\kappa(v).      
  \end{align*}
  A set~\(\pts\) of paths in~\(G_x\)
  is a \emph{partial solution}
  of \emph{cost \(\cost(\pts)\)}
  for~\((S,\prtn{})\) if
  \begin{enumerate}[i)]
    
  \item $\Dz$ are exactly the vertices of zero-length paths \(P\in\pts\),
  \item\label{item:ends} $\De$ are exactly the end points of non-zero-length paths \(P\in\pts\),
  \item $\Di$ are exactly
    those vertices in~\(B_x\)
    that are inner vertices of paths~\(P\in\pts\),
  \item \(N(V(\pts))\cap B_x\subseteq \Ne\),
    that is,
    the vertices in~$B_x$
    that are neighbors of paths in~$\pts$ are in~$N$,
  \item \(\nei_x(\pts,N)\leq l\), and
  \item\label{item:paths} \(\pts{}\) consists of exactly \(|p|+|\Dz|\)~paths
    and each two vertices~\(u,v\in \De\)
    belong to the same path of~\(\pts{}\)
    if and only if they are in the same set
    of the partition~\(\prtn{}\).
  \end{enumerate}
  We will also say that~$\pts{}$ is a partial solution for~$G_x$,
  when the concrete solution signature is clear from context.
  For a solution signature~\((S,\prtn)\) at a node~\(x\), we denote
  \begin{align*}
    \feas_x(S,p)&:=\{\pts{}\mid \pts{}\text{ is a partial solution for~\((S,p)\)}\},\\
    \mcost_x(S,p)&:=\min\{\cost(\pts)\mid \pts\in\feas_x(S,p)\}.
  \end{align*}
\end{definition}

\noindent
Since the root bag~\(B_r=\{s,t\}\) by \cref{ass:st},
our input instance to \WsspAcr{}
is a yes\hyp instance if and only if
\begin{align}
\mcost_r((\emptyset,\{s,t\},\emptyset,\emptyset,\ell),\{\{s,t\}\})\leq k.\label{minKrsmall}
\end{align}
Thus,
our aim is computing this cost.
The naive dynamic programming approach is:
\begin{itemize}
\item compute \(\mcost_x(S,\prtn)\) for each solution signature~\((S,\prtn)\) and each leaf node~\(x\),
\item compute \(\mcost_x(S,\prtn)\) for each solution signature~\((S,\prtn)\) and each inner node~\(x\)
  under the assumption that
  \(\mcost_y(S',\prtn')\)~has already been computed
  for all solution signatures~\((S',\prtn')\)
  at children~\(y\) of~\(x\).
\end{itemize}

\noindent
However,
this approach is not suitable
to prove \cref{prop:wssptw},
since
the number of possible solution signatures
is too large:
the number of different partitions~\(\prtn{}\)
of \(\tw\)~vertices
is the \(\tw\)-th Bell number,
whose best known upper bound is \(O(\tw^\tw/\log\tw)\).
Thus,
we do not even have time to \emph{look}
at all solution signatures.

\subsubsection{Reducing the number of partitions}
To reduce the number of needed partitions,
we use an approach developed by \citet{BCKN15},
which has also been evaluated in experiments \citep{FBN15,BS19}.
We will replace the task of computing
\eqref{minKrsmall} for all possible partitions
by computing only sets of \emph{weighted partitions}
containing the needed information.
\begin{definition}[sets of weighted partitions]
  Let \(\Pi(U)\)~be
  the set of all partitions of~\(U\).
  A~\emph{set of weighted partitions} is a
  set~\(\mathcal A\subseteq\Pi(U)\times\mathbb N\).
  For a
  \emph{weighted partition}~\((p,w)\in\mathcal A\),
  we call \(w\)~its \emph{weight}.
\end{definition}

\noindent
Using sets of weighted partitions,
we can reformulate our task
of computing \(\mcost_x(S,\prtn)\)
for all bags~\(B_x\)
and all solution signatures~\((S,\prtn)\)
as follows.
Consider
a pre-signature~\(S=(\Dz,\Di,\De,N,l)\)
for a node~\(x\) of a tree decomposition
and
\begin{align}
  \mathcal A_x(S):=
  \Biggl\{
  \Bigl(
  p,
  \min_{\pts\in\feas_x(S,p)}K(\pts)
  \Bigr)
  \Biggm|
  \begin{aligned}
    p\in\Pi(\De),
    \feas_x(S,p)\ne\emptyset
  \end{aligned}
  \Biggr\}.
  \label{ax}
\end{align}
\begin{observation}\label[observation]{obs:pairs}
  By \cref{def:partsol}\eqref{item:ends},
  each path in a partial solution
  has both of its end points in~$\De$.
  Thus,
  from \cref{def:partsol}\eqref{item:paths},
  it follows that any partition~$\prtn{}\in  \mathcal A_x(S)$
  consists of pairs of vertices of~$\De$,
  since $\prtn$~allows for partial solutions.
  Thus,
  we can consider $\prtn{}$ as a perfect matching on~$\De$.
\end{observation}
\noindent
Now,
our problem of verifying \eqref{minKrsmall}
at the root node~\(r\)
of a tree decomposition
is equivalent to checking
whether
\(\mathcal A_r(\emptyset,\{s,t\},\emptyset,\emptyset,\ell)\)
contains a partition~\(\{\{s,t\}\}\)
of weight at most~\(k\).
Thus we can,
in a classical
dynamic programming manner
\begin{itemize}
\item compute \(\mathcal A_x(S)\) for each pre-signature~\(S\) and each leaf node~\(x\),
\item compute \(\mathcal A_x(S)\) for each pre-signature~\(S\) and each inner node~\(x\)
  under the assumption that~\(\mathcal A_y(S')\)
  has already been computed for all pre-signatures~\(S'\) at children~\(y\) of~\(x\).
\end{itemize}

\noindent
Yet we will not
work with the full sets~\(A_x(S)\)
but with ``representative''
subsets of size~\(2^{O(\tw)}\).
Since the number of pre-signatures is
\(2^{O(\tw)}\cdot\ell\),
this will allow us to %
prove \cref{prop:wssptw}.
In order to formally introduce
representative sets of weighted partitions,
we need some notation.
\begin{definition}[partition lattice]
  The set \(\Pi(U)\)
  is semi\hyp ordered
  by the coarsening relation~\(\sqsubseteq\),
  where \(p\sqsubseteq q\)
  if every set of~\(p\)
  is included in some set of~\(q\).
  We also say that \(q\)~is
  \emph{coarser} than~\(p\)
  and that \(p\)~is \emph{finer} than~\(q\).

  For two
  partitions~\(p,q\in \Pi(U)\),
  by
  \(p\sqcup q\) we denote
  the (unique) finest partition
  that is coarser
  than both~\(p\) and~\(q\).
\end{definition}

\noindent
To get an intuition for the \(p\sqcup q\)
operation,
recall from \cref{def:partsol}
that %
two end points of paths in partial solutions
belong to the same path (that is, connected component)
of a partial solution
if and only if
they are in the same set of~\(p\).
In these terms,
if \(p\in\Pi(U)\)~are the vertex sets
of the connected components
of a graph~\((U,E)\)
and \(q\in\Pi(U)\)~are the vertex sets
of the connected components
of a graph~\((U,E')\),
then \(p\sqcup q\)~are the vertex sets
of the connected components
of the graph~\((U,E\cup E')\).

Now,
assume that there is a solution~\(P\) to
\(\WsspAcr{}\)
in a graph~\(G\)
and consider an arbitrary node~\(x\)
of a tree decomposition.
Then,
the subpaths~\(\pts{}\) of~\(P\)
that lie in~\(G_x\)
are a partial solution
for some solution
signature~\((S,\prtn)\)
with $S=(\Dz,\De,\Di,N,l)$ at~\(x\).
The partition~\(\prtn{}\) of~\(\De\)
consists of the pairs of end points of
nonzero\hyp length paths in~\(\pts{}\).
Since,
in the overall solution~\(P\),
the vertices in~\(\De\) are all connected,
the vertices of~\(\De\)
are connected in~\(G\setminus E_x\)
according to a partition~\(q\) of~\(\De\)
such that \(p\sqcup q=\{\De\}\).
Now,
if in~\(P\), we replace
the subpaths~\(\pts\)
by any other partial solution~\(\pts'\)
to a solution signature
\((S,\prtn')\)
such that \(\cost(\pts')\leq \cost(\pts)\)
and \(p'\sqcup q=\{\De\}\),
then
we obtain a solution~\(P'\) for~\(G\)
with at most the cost of~\(P\).
Thus,
one of the two weighted partitions~\((p,K(p))\)
and~\((p',K(p'))\) in~\(\mathcal A_x(S)\)
is redundant.
This leads to the definition
of \emph{representative} sets of weighted partitions.

\begin{definition}[representative sets \citep{BCKN15}]
  For a set \(\mathcal A\subseteq\Pi(U)\times\N\)
  of weighted partitions
  and a partition~\(q\in\Pi(U)\), let
  \begin{align*}
    \opt(q,\mathcal A)&:=\min\{w\mid (p,w)\in\mathcal A\text{ and }p\sqcup q=\{U\}\}.
  \end{align*}
  Another set~\(\mathcal A'\subseteq\Pi(U)\times\mathbb N\) of weighted partitions is said
  to \emph{represent~\(\mathcal A\)}
  if
  \[\opt(p,\mathcal A)=\opt(p,\mathcal A')\text{ for all }p\in\Pi(U).\] 
  A function~\(f\colon 2^{\Pi(U)\times\N}\times Z\to
  2^{\Pi(U)\times\N}\)
  is said to \emph{preserve representation}
  if,
  for all~\(\mathcal A,\mathcal A'\subseteq\Pi(U)\times \N\) and all \(z\in Z\),
  it holds that, if \(\mathcal A'\)
  represents~\(\mathcal A\),
  then \(f(\mathcal A',z)\)
  represents~\(f(\mathcal A,z)\).
  Herein,
  \(Z\)~stands representative for further
  arguments to~\(f\).
\end{definition}
\noindent
Transferring this definition
to \WsspAcr{} and our sets \(\mathcal A_x(S)\)
in \eqref{ax},
\(\opt(q,\mathcal A_x(S))\)
is the minimum cost
of any partial solution
for any signature~\((S,p)\)
at a node~\(x\)
that leads to a connected overall
solution when the vertices
in~\(\De\)~are connected
in \(G\setminus E_x\) as described by partition~\(q\).
Moreover,
a subset~\(\mathcal A'\subseteq\mathcal A_x(S)\)
is representative
if this minimum cost for~\(\mathcal A'\)
is the same.
\begin{proposition}[\citet{BCKN15}]
  \label{prop:shrink}
  Given a set~\(\mathcal A\subseteq\Pi(U)\times \N\)
  of weighted partitions
  that form perfect matchings on~$U$,
  a representative subset~\(\mathcal A'
  \subseteq\mathcal A\)
  with \(|\mathcal A'|\leq 2^{|U|/2}\)
  is computable in
  \(2^{(\omega-1)\cdot|U|/2}\cdot|\mathcal A|\cdot|U|^{O(1)}\)~time,
  where \(\omega<2.2373\)
  is the matrix multiplication exponent.
\end{proposition}

\noindent
When computing
the sets~\(\mathcal A_x(S)\)
for pre-signatures~\(S\) at a node~\(x\),
we will first replace
the \(A_y(S')\) for all pre-signatures~\(S'\)
at the child nodes~\(y\)
by their representative sets
and thus work on sets of size~\(2^{O(\tw)}\).
More precisely,
we will compute~\(A_x(S)\)
from the children sets~\(A_y(S')\)
using the following operators,
which \citet{BCKN15}
have showed to preserve representation.
\begin{proposition}[operators on weighted partitions \citep{BCKN15}]
  \label[proposition]{prop:operators}
    Let \(U\)~be a set and~\(\mathcal A\subseteq \Pi(U)\times\mathbb N\).
    The following operations preserve representation.
\begin{itemize}[aaaaaaaaaaaaaaa]
\item[$\rmc(\mathcal A):=$]$\{(p,w)\mid \forall (p,w')\in\mathcal A:w'\geq w\}\subseteq\Pi(U)\times\N$
  simply removes duplicate partitions
  from the set,
  keeping the one with smallest weight.
\item[{$\mathcal A\mincup \mathcal B:=$}]$\rmc(A\cup B)\subseteq\Pi(U)\times\N$
  for some~\(\mathcal B\subseteq\Pi(U)\times\N\)
  takes all weighted partitions from
  \(\mathcal A\) and \(\mathcal B\),
  removing copies of larger weight.

\item[$\glue(\{u,v\},\mathcal A)\subseteq$] $\Pi(U\cup\{u,v\})\times\N$
  in all partitions
  merges the sets containing~\(u\) and~\(v\)
  into one.
  If necessary,
  $u$ and~$v$ are introduced into the universe
  as singletons first.
\item[$\shift(w',\mathcal A):=$]$\{(p,w+w')\mid (p,w)\in\mathcal A\}\subseteq
  \Pi(U)\times\N$ increases the weight of each partition in~\(\mathcal A\) by~\(w'\in\Z\)
  for \(w'\geq-\min\{w\mid (p,w)\in\mathcal A\}\).%
  \footnote{\citet{BCKN15} actually
    require \(w'\in\N\) here.
    Yet from the proof of their Lemma~3.6,
    it is easy to see that \(\shift\)
    preserves representation
    whenever it returns
    a set of partitions with nonnegative weights.}
\item[$\proj(X,\mathcal A)\subseteq$] $\Pi(U\setminus X)\times\N$
  for \(X\subseteq U\)
  removes the elements of~\(X\)
  from~\(U\),
  from all sets in~\(\mathcal A\),
  and removes partitions
  in which the removal of the elements of~\(X\) reduced
  the number of sets.
\item[$\join(\mathcal A,\mathcal B):=$]$\rmc(\{(p\sqcup q, w_1+w_2)\mid (p,w_1)\in\mathcal A, (q,w_2)\in\mathcal B\})$ for any \(\mathcal B\in\Pi(U)\).
\end{itemize}
\noindent
Moreover,
all operations \(\rmc,\mincup,\glue,\shift,\proj\)
can be executed in \(s\cdot |U|^{O(1)}\)~time,
where \(s\)~is the size of the input to the operations,
whereas \(\join\) can be executed in
\(|\mathcal A|\cdot|\mathcal B|\cdot|U|^{O(1)}\)time.
\end{proposition}

\subsubsection{The dynamic programming algorithm}

We will now show
how to use the operators from \cref{prop:operators}
to compute,
for each node~\(x\) of a tree decomposition
and each pre-signature~\(S\),
the weighted set of partitions~\(\mathcal A_x(S)\)
from~\eqref{ax}
assuming that \(\mathcal A_y(S')\)
has been computed for all children~\(y\)
of~\(x\) and all pre-signatures~\(S'\).
Using these operators guarantees that,
when applying them to representative subsets,
we will again get representative subsets.
We describe our algorithm independently for
leaf nodes,
forget nodes,
insert vertex nodes,
insert edge nodes,
and join nodes.

\paragraph{Leaf node \(x\).}
\label{sec:leafnode}
By \cref{ass:st},
\(B_x=\{s,t\}\).
Moreover,
\(G_x\)~has no edges.
Thus,
any partial solution
for any pre-signature~\(S=(\Dz,\De,\Di,N,l)\)
at~$x$
will contain~\(s\) and~\(t\)
as paths of length zero
and
\[
  \mathcal A_x(S)=
  \begin{cases}
    \Bigl\{\bigr(\emptyset,
    \kappa(s)+\kappa(t)\bigl)\Bigr\}&
    \text{if }\Dz=\{s,t\}, \De=\Di=N=\emptyset, \eta(s)+\eta(t)\leq l, \text{ and}\\
    \emptyset&\text{otherwise}.
  \end{cases}
\]

\paragraph{Introduce vertex~\(v\) node~\(x\) with child~\(y\).}
\label{sec:introducevertex}
For any %
pre-signature~\(S=(\Dz,\De,\Di,N,l)\),
we compute \(\mathcal A_x(S)\)
according to one of the following three cases.

\begin{enumerate}[({IV}1)]
\item
  \label{iv1}
  If \(v\in \Dz\cup\De\cup\Di\),
  then partial solutions for~\(S\)
  exist only if~\(\eta(v)\leq l\) and
  \(v\in\Dz\),
  since the newly introduced vertex
  has no edges in~\(G_x\) yet.
  Removing~$v$ from such a partial solution
  yields a partial solution~$\pts$ for~\(G_y\)
  with \(\Lambda_y(\pts,N)\leq l-\eta(v)\)
  and with \(\kappa(v)\) less cost.
\item
  \label{iv2}
  If \(v\in N\),
  then partial solutions for~\(S\)
  exist only if~\(l\geq\lambda(v)\).
  Such a partial solution
  is a partial solution~\(\pts{}\)
  for~\(G_y\)
  with the same cost and
  \(\Lambda_y(\pts,N\setminus\{v\})\leq l-\lambda(v)\).
  
\item
  \label{iv3}
  If neither of both,
  then any partial solution for~\(S\)
  is also one for~\(G_y\)
  of the same cost.
\end{enumerate}

\noindent
Thus,
\begin{align*}
  \mathcal A_x(S)=
  \begin{cases}
    \shift\Bigr(\kappa(v),
    \mathcal A_y(\Dz\setminus\{v\},\De,\Di,N,l-\eta(v))\Bigl)&\text{ if $v\in\Dz$ and $l\geq\eta(v)$},\\
    \mathcal A_y(\Dz,\De,\Di,N\setminus\{v\},l-\lambda(v))&\text{ if $v\in N$ and $l\geq\lambda(v)$},\\
    \mathcal A_y(\Dz,\De,\Di,N,l)&\text{ if $v\notin \Dz\cup\De\cup\Di\cup N$,}\\
    \emptyset&\text{ otherwise}.
  \end{cases}
\end{align*}

\paragraph{Introduce edge \(\{u,v\}\) node~\(x\)
  with child~\(y\).}
\label{sec:introduceedge}

\begin{table*}
  \centering
  \begin{tabular}{r|llllllllllllllll}
    \toprule
    $u\in$ & \Dz         & \Dz         & \Dz         & \Dz         & \Dz         & \De         & \De         & \De         & \De         & \Di         & \Di         & \Di         & $N$         & $N$         &  $B$ & \\
    $v\in$ & \Dz         & \De         & \Di         & $N$         & $B$         & \De         & \Di         & $N$         & $B$         & \Di         & $N$         & $B$         & $N$         & $B$           &  $B$ & \\
    \midrule
    Case IE  & \ref{case1} & \ref{case1} & \ref{case1} & \ref{case1} & \ref{case3} & \ref{case2} & \ref{case2} & \ref{case1} & \ref{case3} & \ref{case2} & \ref{case1} & \ref{case3} & \ref{case1} & \ref{case1} & \ref{case1} \\
    \bottomrule
  \end{tabular}
  \caption{All possibilities for containment of the vertices~$u$ and~$v$ in one of the sets~$\Dz,\De,\Di,N$ and $B:=B_x\setminus(\Dz\cup\De\cup\Di\cup N)$ (modulo symmetry)
    and which of the cases (IE\ref{case3})--(IE\ref{case1}) applies.}
  \label{tab:cases}
\end{table*}

For any pre-signature~\(S=(\Dz,\De,\Di,N,l)\),
we compute \(\mathcal A_x(S)\)
according to one of three cases for~\(u\) and~\(v\):

\begin{enumerate}[({IE}1)]
\item
  \label{case3}
  \looseness=-1
  If $v\in\Dz\cup\De\cup\Di$
  and $u\notin \Dz\cup\De\cup\Di\cup N$,
  then
  \(v\)~is required to be part of a path
  in our partial solution to~\(G_x\),
  whereas its neighbor~\(u\)
  is not allowed to be on any path
  nor neighbor of a path.
  There is no such feasible solution.
  The same holds when exchanging the roles of~$u$ and~$v$.
\item
  \label{case2}
  If \(\{u,v\}\subseteq \De\cup\Di\)
  then we have two choices:
  take edge~\(\{u,v\}\) into a path of a partial solution
  or not.
  A partial solution for~$S$
  not containing edge~$\{u,v\}$ on any path
  is also one for~$G_y$.
  A partial solution~$\pts$ for~$S$
  containing~$\{u,v\}$ on one of the paths
  gives rise to a partial solution~$\pts'$ for~$G_y$
  in which
  $u$ and~$v$
  must be in different connected components
  (since $\pts$ is not allowed to contain cycles),
  which, after adding edge~$\{u,v\}$,
  will be one connected component of~$\pts$.
  Moreover,
  if \(u\in\De\),
  then \(u\)~is a path of zero length in~$\pts'$;
  if \(u\in\Di\), then
  \(u\)~is an end point
  of a path in~$\pts'$.
  Symmetrically, this holds for~\(v\).

\item
  \label{case1}
  None of the above (see \cref{tab:cases}).
  Edge~\(\{u,v\}\)
  cannot be part of a path
  in a partial solution to~\(S\).
  Moreover,
  if it is incident to a path vertex,
  then both its vertices lie on a partial solution path
  or one of them is in~$N$.
  Thus,
  any partial solution for~\(G_y\)
  satisfying~\(S\)
  is a partial solution
  of the same cost for~\(G_x\).
\end{enumerate}

\noindent
Thus, we can compute $\mathcal A_x(S)$ as follows,
where the $\proj$ operation in the fourth case
replaces $u$ by~$v$ in all partitions,
in the fifth case replaces $v$ by $u$,
and in the third case
together with \cref{obs:pairs} ensures
that an edge is introduced only between vertices~$u$ and~$v$
in different partial solution paths:
if $u$ and $v$ in the same partial solution path,
then $\{u,v\}$ is a set of a partition~$\prtn{}$,
the $\glue$ operation does nothing to~$\prtn{}$,
and $\proj$  discards~$\prtn$ because
$\{u,v\}$ becomes empty after deleting~$u$ and~$v$.

\noindent
\begin{align*}
  \mathcal A_x(S)=
  \begin{cases}
    \emptyset&\text{ if $v\in\Dz\cup\De\cup\Di$
      and}\\
    &\quad u\notin \Dz\cup\De\cup\Di\cup N
    \text{, or symmetrically,}\\
    \mathcal A_y(S)\mincup
    \mathcal \glue(\{u,v\},
    \mathcal A_y(S[\{u,v\}_{\De\to \Dz}]))
    &\text{ if }\{u,v\}\subseteq\De,
    \\
    \mathcal A_y(S)\mincup
    \mathcal \proj(\{u,v\},\glue(\{u,v\},
    \mathcal A_y(S[\{u,v\}_{\Di\to \De}])))
    &\text{ if }\{u,v\}\subseteq\Di,
    \\
    \mathcal A_y(S)\mincup
    \mathcal \proj(\{u\},\glue(\{u,v\},
    \mathcal A_y(S[\{u\}_{\De\to\Dz},\{v\}_{\Di\to\De}])))
    &\text{ if }v\in\Di,u\in\De,
    \\
    \mathcal A_y(S)\mincup
    \mathcal \proj(\{v\},\glue(\{u,v\},
    \mathcal A_y(S[\{v\}_{\De\to\Dz},\{u\}_{\Di\to\De}])))
    &\text{ if }u\in\Di,v\in\De,\text{ and}
    \\
    \mathcal A_y(S)&\text{ otherwise,}
  \end{cases}
\end{align*}
where
\begin{align*}
  S[\{u,v\}_{\De\to \Dz}]
                 &:=(\Dz\cup\{u,v\},\De\setminus\{u,v\},\Di, N,l)),
  \\
  S[\{u,v\}_{\Di\to \De}]
  &:=
    (\Dz,\De\cup\{u,v\},\Di\setminus\{u,v\}, N,l)),\\
  S[\{u\}_{\De\to\Dz},\{v\}_{\Di\to\De}]
                 &:=(\Dz\cup\{u\},(\De\cup\{v\})\setminus\{u\},\Di\setminus\{v\}, N,l).
\end{align*}

\paragraph{Forget vertex~\(v\) node~\(x\)
  with child~\(y\).}
\label{sec:forgetnode}

By \cref{ass:st},
the forgotten vertex~\(v\notin\{s,t\}\).
For any pre-signature~\(S=(\Dz,\De,\Di,N,l)\),
we compute \(\mathcal A_x(S)\) as follows.
A partial solution~$\pts$ to~$S$ at~$G_x$
is also one at~$G_y$,
just its pre-signature depends on
the role of the forgotten vertex~\(v\)
in~$\pts$:
\begin{enumerate}[({F}1)]
\item If $v$~is a neighbor of a path in~$\pts$,
  then $\pts$~is a partial solution at~$y$
  with the pre-signature obtained from~$S$ by
  adding~$v$ to~$N$.
  Note that
  we can safely ignore \(v\) in~\(x\) and all parent nodes:
  we already accounted for the cost~\(\lambda(v)\)
  when \(v\)~was introduced
  and
  by \cref{def:tw}\eqref{tw2} and \eqref{tw3},
  no parent node of~\(x\)
  will ever introduce an edge incident to~\(v\).
  
\item If $v$~is part of a path in~$\pts$,
  then
  \(v\)~must be an \emph{inner} vertex of
  such path
  since
  \(v\)~must be an inner vertex
  of the overall solution
  and
  no parent node of~\(x\)
  will be able to introduce edges incident to~\(v\)
  due to \cref{def:tw}\eqref{tw2} and \eqref{tw3}.
  In this case,
  $\pts$~is a partial solution
  at~$y$
  with the pre-signature obtained from~$S$ by
  adding~$v$ to~$\Di$.

\item Neither of both.
  In this case,
  any partial solution for~$S$ at \(x\)
  is also one of the same cost for~$S$ at~$y$.
\end{enumerate}

\noindent
Thus,
\begin{align}
  \begin{split}
    \mathcal A_x(S)&=\mathcal A_y(\Dz,\De,\Di,N\cup\{v\},l)\mincup A_y(\Dz,\De,\Di\cup\{v\},N,l)\mincup \mathcal A_y(\Dz,\De,\Di,N,l).
    \label{funion}
  \end{split}
\end{align}

\paragraph{Join node~$x$ with children~\(y\) and~\(z\).}
\label{sec:joinnode}

For any
pre-signature~\(S=(\Dz,\De,\Di,N,l)\),
to compute \(\mathcal A_x(S)\),
we consider the roles of each vertex~\(v\in B_x\)
and the types of partial solutions $\pts_y$ for~$G_y$
and $\pts_z$ for~$G_z$
that a partial solution~$\pts$ for~$S$ decomposes into:

\begin{enumerate}[({J}1)]
\item
  \label{j1}
  If \(v\in N\),
  then~$\pts$
  allows~\(v\) as a neighbor.
  Thus,
  $\pts_y$ and~$\pts_z$
  also must allow~\(v\) as a neighbor.
\item
  \label{j2}
  If \(v\in\Dz\),
  then \(v\)~must be a path of length zero
  in~$\pts$ and therefore in~$\pts_y$ and~$\pts_z$.

\item
  \label{j3}
  \looseness=-1
  If \(v\in\De\),
  then $v$~is the end point of a path in~$\pts$.
  Thus,
  \(v\)~is a path of length zero in~$\pts_y$
  and an end point of a path of nonzero length
  in~$\pts_z$
  or vice versa.

\item
  \label{j4}
  If \(v\in\Di\),
  then $v$~is an inner vertex of a path in~$\pts$.
  Thus,
  \(v\) might be the end point
  of paths of non\hyp zero length
  in both~$\pts_y$ and~$\pts_z$,
  it might be an inner vertex of~$\pts_y$
  and a path of length zero of~\(\pts_z\),
  or vice versa.
  
\item
  \label{j5}
  \looseness=-1
  Otherwise,
  \(v\)~is not part of~$\pts$
  nor allowed to neighbor it.
  Thus,
  \(v\)~is also disallowed to be part of or to neighbor
  $\pts_y$ and $\pts_z$.
\end{enumerate}

\noindent
To compute the \(\mathcal A_x(S)\), let
\begin{align*}
  D&:=\Dz\cup\De\cup\Di,&
  \lambda_N&:=\sum_{v\in N}\lambda(v)+\sum_{v\in D}\eta(v),&
  \kappa_D&:=\sum_{v\in D}\kappa(v).
\end{align*}
By \cref{def:tw}\eqref{tw3},
$V_y\cap V_z\subseteq B_x$.
Thus,
the set of vertices common to~\(\pts_y\)
and \(\pts_z\) lies in~\(B_x\) and, hence, is precisely~\(D\).
Thus
\[
  K(\pts_y)+K(\pts_z)-\kappa_D=K(\pts)\geq 0.
\]
By \cref{def:tw}\eqref{tw2} and \eqref{tw3},
all common neighbors
of~\(\pts_y\) and~\(\pts_z\) also lie in~$B_x$,
and thus in~\(N\).
Thus,
\(\Lambda_y(\pts_y,N)\leq l_y\)
and \(\Lambda_z(\pts_z,N)\leq l_z\)
for some
\[
  l_y+l_z=l+\lambda_N.
\]
Thus, by (J\ref{j1})--(J\ref{j5}), one has
\begin{align}
  \mathcal A_x(S)&:=
  \,\,\smashoperator{\bigmincup\limits_{\substack{
  l_y+l_z=l+l_N\\
  \Dz^y\uplus\De^y\uplus\Di^y=D\\
  \Dz\subseteq\Dz^y\\
  \De^y\subseteq\De\cup\Di\\
  \Di^y\subseteq\Di
  }}}\,\,
  \shift\Biggl(-\kappa_D,
  \proj\Bigl((\De^y\cup\De^z)\setminus\De^x),
  \join\bigl(\mathcal A_y(\Dz^y,\De^y,\Di^y,N,l_y),
  \mathcal A_z(\Dz^z,\De^z,\Di^z,N,l_z)\bigr)\Bigr)\Biggr)
  \label{junion}
\end{align}
where the \(\Dz^z,\De^z\) and \(\Di^z\) are
fully determined by the choice of \(\Dz^y,\De^y\) and \(\Di^y\)
via
\begin{align*}
  \Dz^z&=\Dz\cup\Di^y\cup(\De\cap\De^y),&
  \De^z&=(\De\cap\Dz^y)\cup(\Di\cap\De^y),&
  \Di^z&=D\setminus(\Dz^z\cup\De^z),
\end{align*}
and where the $\proj$ operation
removes those vertices from the partitions
that are in $\De^y$ or $\De^z$ but
are not in $\De^x$.
Note that,
as in the case of introduce edge nodes,
the $\proj$ operation also ensures
that we do not create disconnected cycles:
when two vertices~$u,v\in\De^y\cap\De^z$
are joined to become part of~$\Di^y$
and $\{u,v\}$~is a set both
of a partition for~$y$ and a partition for~$z$,
then $\{u,v\}$ will be a set in the joined partition,
it will be empty after removing~$u$ and~$v$,
and therefore the partition will be discarded by~$\proj$.

\paragraph*{Wrapping up.}
\looseness=-1
Having described how to compute \eqref{ax}
for each node type of a nice tree decomposition,
we are now ready to prove \cref{prop:wssptw}
exploiting that we can efficiently
compute small representative subsets
of our families of weighted partitions
using \cref{prop:shrink}.
We will apply this shrinking procedure
to all intermediate sets computed in our algorithm.

\begin{proof}[Proof of \cref{prop:wssptw}]
  Our algorithm works as follows.
  It first preprocesses the given tree decomposition
  according to \cref{ass:st},
  which can be done in \(n\cdot\tw^{O(1)}\)~time
  and thus gives a tree decomposition
  with \(n\cdot\tw^{O(1)}\)~bags \citep{BCKN15}.
  Henceforth,
  we will be working on a tree decomposition
  of width at most~\(\tw+2\),
  that is,
  each bag has size at most~\(\tw+3\).

  The algorithm now computes
  \eqref{ax}
  for each node of the tree decomposition
  and each pre-signature~\(S\)
  as described in
  \cref{sec:leafnode,sec:introducevertex,sec:introduceedge,sec:forgetnode,sec:joinnode}.
  However,
  after computing \(\mathcal A_x(S)\)
  for some pre-signature~\(S\) at some node~\(x\),
  it will use \cref{prop:shrink}
  to store only a representative subset \(\mathcal A_x'(S)\)
  with
  \(|\mathcal A_x'(S)|\leq 2^{(\tw+3)/2}\).
  Since we compute \(\mathcal A_x(S)\)
  only using operators in \cref{prop:operators},
  the set \(\mathcal A_x'(S)\) represents \(\mathcal A_x(S)\)
  at each node~\(x\) of the tree decomposition,
  in particular at the root node~\(r\),
  where we can now verify whether
  \(\mathcal A_r'(\emptyset,\{s,t\},\emptyset,\emptyset,\ell)\)
  contains a partition \(\{\{s,t\}\}\) of weight at most~\(k\).

  We analyze the running time of this algorithm.
  For each pre\hyp signature~\(S\),
  each leaf node
  can be processed according to \cref{sec:leafnode}
  in constant time
  and the stored representative
  subset \(\mathcal A_x'(S)\subseteq\mathcal A_x(S)\)
  has constant size
  and can be computed in constant time.
  According to \cref{def:partsol},
  there are at most \(5^{\tw+3}\cdot\ell\)
  pre\hyp signatures,
  since the bags of our tree decomposition
  have size at most~\(\tw+3\).
  Thus,
  each leaf node can
  be processed in \(O(5^\tw\cdot \ell)\)~time.

  \looseness=-1
  For each pre\hyp signature~\(S\)
  with \(|\De|=i\),
  the most expensive operation
  when processing
  introduce vertex,
  introduce edge,
  and forget nodes
  is the union operation in \eqref{funion},
  which is applied to
  three sets of weighted partitions,
  each of a size upper bounded by \(2^{i/2}\).
  By \cref{prop:operators},
  this union can be computed in
  \(3\cdot 2^{i/2}\cdot i^{O(1)}=2^{i/2}\cdot i^{O(1)}\)~time.
  The resulting
  set of weighted partitions
  therefore has size at most \(2^{i/2}\cdot i^{O(1)}\).
  Thus, shrinking using \cref{prop:shrink}
  works in \(2^{(\omega-1)i/2}\cdot 2^{i/2}
  \cdot i^{O(1)}\)~time,
  which is the most expensive operation
  for each fixed pre\hyp signature.
  Since there are at most
  \(\binom{\tw+3}{i} \cdot 4^{\tw+3-i}\cdot \ell\)~pre\hyp signatures 
  with \(|\De|=i\),
  each
  introduce vertex,
  introduce edge,
  and forget node
  is processed in
  \begin{align*}
    &\sum_{i=1}^{\tw+3}\binom{\tw+3}{i}\cdot 4^{\tw+3-i}\cdot \ell\cdot 2^{\omega i/2} \cdot i^{O(1)}\leq
      \ell\cdot\tw^{O(1)}\cdot
      \sum_{i=1}^{\tw+3}\binom{\tw+3}{i}\cdot (2^{\omega/2})^i\cdot 4^{\tw+3-i}=\ell\cdot\tw^{O(1)}\cdot (4+2^{\omega/2})^{\tw+3}\text{ time.}
  \end{align*}

  \noindent
  For each pre\hyp signature~\((\Dz,\De,\Di,N,l)\)
  with \(|\Dz|=\iz,|\De|=\ie,|\Di|=\ii\),
  processing a join node
  according to \cref{sec:joinnode}
  is more costly.
  By (J\ref{j1})--(J\ref{j5}),
  the union operator in \eqref{junion}
  has up to \(1^{\iz}\cdot 2^{\ie}\cdot 3^{\ii}\cdot l\)~operands.
  Each operand
  is a join of two sets of size at most \(2^{(\ie+\ii)/2}\)
  and, by \cref{prop:operators},
  takes \(2^{\ie+\ii}\cdot\tw^{O(1)}\)~time
  to compute.
  Thus,
  the union operator
  is applied to sets
  whose total size is bounded
  by \(2^{\ie}\cdot 3^{\ii}\cdot 2^{\ie+\ii}\cdot l\).
  Shrinking according to
  \cref{prop:shrink}
  thus works in
  \(2^{(\omega-1)\ie/2}\cdot 2^{\ie}\cdot 3^{\ii}\cdot 2^{\ie+\ii}\cdot l\cdot\tw^{O(1)}\)~time
  and, again,
  is the most expensive operation
  in the computation of a join node
  for a fixed pre\hyp signature.
  There are
  \[\binom{\tw+3}{i}\sum_{\iz+\ie+\ii=i}\binom{i}{\iz,\ie,\ii}\cdot 2^{\tw+3-i}\cdot\ell\]
  pre\hyp signatures~\((\Dz,\De,\Di,N,l)\)
  with \(|\Dz|=\iz,|\De|=\ie,|\Di|=\ii\),
  and \(\iz+\ie+\ii=i\),
  where
  \(\binom{i}{\iz,\ie,\ii}\)
  is the multinomial coefficient---%
  the number of ways
  to throw \(i\)~items
  into three distinct bins such that
  the first bin gets \(\iz\)~items,
  the second gets \(\ie\)~items,
  and the third gets \(\ii\)~items.
  Thus,
  each join node is processed in a total time of
  \begin{align*}
    &\sum_{i=1}^{\tw+3}\binom{\tw+3}{i}\sum_{\iz+\ie+\ii=i}\binom{i}{\iz,\ie,\ii}\cdot 2^{\tw+3-i}\cdot\ell \cdot
2^{(\omega-1)\ie/2}\cdot 2^{\ie}\cdot 3^{\ii}\cdot 2^{\ie+\ii}\cdot l\cdot\tw^{O(1)}\\
    &\leq\ell^2\cdot \tw^{O(1)}\cdot \sum_{i=1}^{\tw+3}\binom{\tw+3}{i}
      \left(\sum_{\iz+\ie+\ii=i}\binom{i}{\iz,\ie,\ii}\cdot 1^{\iz}\cdot (2^{(\omega+3)/2})^\ie\cdot 6^{\ii}\right)\cdot 2^{\tw+3-i}\\
    \intertext{which, by the multinomial theorem, is}
    &=\ell^2\cdot \tw^{O(1)}\cdot \sum_{i=1}^{\tw+3}\binom{\tw+3}{i}
      \cdot (7+2^{(\omega+3)/2})^i\cdot 2^{\tw+3-i}=\ell^2\cdot\tw^{O(1)}\cdot(9+2^{(\omega+3)/2})^{\tw+3}.
  \end{align*}
  Since our tree decomposition has \(n\cdot\tw^{O(1)}\)~bags,
  we conclude that
  we can solve \WsspAcr{} in \(n\cdot\ell^2\cdot\tw^{O(1)}\cdot (9+2^{(\omega+3)/2})^{\tw}\) time.  
\end{proof}

\subsection{Limits of data reduction}
\label{sec:notwkern}
In the previous section,
we presented a fixed\hyp parameter algorithm
for \sspAcr{} parameterized by treewidth.
Now
we prove that,
unless~$\unlessPK$,
\sspAcr{} has no problem kernel
of size polynomial in~$k$, $\ell$, and the treewidth of the input graph combined.
In fact,
we prove the following:

\begin{theorem}
  \label{thm:nopktwell}
  \sspTsc{}
  has no problem kernel
  with size polynomial in~$\tw+k+\ell$, even on planar graphs
  with maximum degree six,
  where \(\tw\)~is the treewidth,
  unless $\unlessPK$.
\end{theorem}

\noindent
To prove \cref{thm:nopktwell},
we use a special kind of reduction
called \emph{cross composition}
\citep{BJK14}.

\begin{definition}[cross composition]
  \label[definition]{def:crossco}
  A \emph{polynomial equivalence relation}~$\sim$ is
  an equivalence relation over~$\Sigma^*$ such that
  \begin{itemize}
  \item there is an algorithm that decides~$x\sim y$ in polynomial time
    for any two instances~$x,y\in\Sigma^*$, and such that
  \item the index of~$\sim$ over any \emph{finite} set~$S\subseteq
    \Sigma^*$ is polynomial in $\max_{x\in S}|x|$.
  \end{itemize}

  \noindent
  A language~$K\subseteq\Sigma^*$ \emph{cross-composes} into a
  parameterized language~$L\subseteq\Sigma^*\times\mathbb N$ if there
  is a polynomial\hyp time algorithm, called \emph{cross composition},
  that, given a sequence~$x_1,\ldots,x_p$ of $p$~instances that are
  equivalent under some polynomial equivalence relation, outputs
  an instance~$(x^*,k)$ such that
  \begin{itemize}
  \item $k$~is bounded by a polynomial in~$\max^p_{i=1}|x_i|+\log p$ and
  \item $(x^*,k)\in L$ if and only if there is an~$i\in\{1,\dots,p\}$
    such that~$x_i\in K$.
  \end{itemize}
\end{definition}

\noindent
Cross compositions can be used
to rule out
problem kernels of polynomial size
using the following result:

\begin{proposition}[\citet{BJK14}] \label{thm:Bod-No-Poly-Kernel}
  If an NP-hard language~$K\subseteq \Sigma^*$ cross-composes into the
  parameterized language~$L\subseteq \Sigma^*\times \mathbb{N}$, then
  there is no polynomial-size problem kernel for~$L$ unless $\unlessPK$.
\end{proposition}

\noindent
Using a cross composition,
\citet{LF18} proved
that \sspAcr{} on planar graphs of maximum degree six
does not admit
a problem kernel with size polynomial in~$k+\ell$.
To prove \cref{thm:nopktwell},
we show that the graph
created by their cross composition
has treewidth at most~$n+3$,
where~$n$ is the number of vertices
in each input instance to their cross composition.
To this end,
we briefly describe
their composition 
(see \cref{fig:LFconstr} for an illustrative example):
 
\begin{construction}
  \label[construction]{constr:twnopkern}
  \begin{figure}[t]
  \centering
   \begin{tikzpicture}

    \tikzset{decorate sep/.style 2 args=
    {decorate,decoration={shape backgrounds,shape=circle,shape size=#1,shape sep=#2}}}

    \tikzstyle{xnode}=[circle, scale=1/2, fill, draw];
    \tikzstyle{tnode}=[circle, scale=1/2, fill, draw];
    \tikzstyle{xedge}=[thick,-];

    \def\xr{1}
    \def\yr{1}
    \def\txr{1.5}
    \def\tyr{1}

    \newcommand{\LFbintree}[4]{
      \node (#21) at (0,0)[tnode,label=-90:{$#4=#2_1$}]{};
      \node (#22) at (1*#1*\txr,-1*\tyr)[tnode,label=-90:{$#2_2$}]{};
      \node (#31) at (2*#1*\txr,-1.5*\tyr)[tnode,label=-90:{$#2_3=#3_1$}]{};
      \node (#32) at (2*#1*\txr,-0.5*\tyr)[tnode,label=-90:{$#2_4=#3_2$}]{};
      \node (#25) at (1*#1*\txr,1*\tyr)[tnode,label=-90:{$#2_5$}]{};
      \node (#33) at (2*#1*\txr,0.5*\tyr)[tnode,label=-90:{$#2_6=#3_3$}]{};
      \node (#34) at (2*#1*\txr,1.5*\tyr)[tnode,label=-90:{$#2_7=#3_4$}]{};
      \draw[xedge] (#21) -- (#22) -- (#31);
      \draw[xedge] (#22) -- (#32);
      \draw[xedge] (#21) -- (#25) -- (#33);
      \draw[xedge] (#25)-- (#34);
    }

    \newcommand{\LFinsta}[1]{
    \node at (0,0)[minimum width=3*\xr cm,ellipse, draw]{$G_#1$};
    \node (s#1) at (-1.5*\xr,0)[xnode,label=-90:{$s_#1$}]{};
    \node (t#1) at (1.5*\xr,0)[xnode,label=-90:{$t_#1$}]{};
    }

    \LFbintree{1}{g}{a}{s};
    \begin{scope}[xshift=12*\xr cm]\LFbintree{-1}{h}{b}{t};\end{scope}
    \begin{scope}[xshift=6*\xr cm,yshift=-1.5*\yr cm]\LFinsta{1};\end{scope}
    \begin{scope}[xshift=6*\xr cm,yshift=-0.5*\yr cm]\LFinsta{2};\end{scope}
    \begin{scope}[xshift=6*\xr cm,yshift=0.5*\yr cm]\LFinsta{3};\end{scope}
    \begin{scope}[xshift=6*\xr cm,yshift=1.5*\yr cm]\LFinsta{4};\end{scope}

    \foreach \x in {1,...,4}{
      \draw[xedge] (a\x) -- (s\x);
      \draw[xedge] (t\x) -- (b\x);
    }
    \end{tikzpicture}
    \caption{Illustrative example of~\cref{constr:twnopkern} with~$p=4$ instances.}
    \label{fig:LFconstr}
  \end{figure}
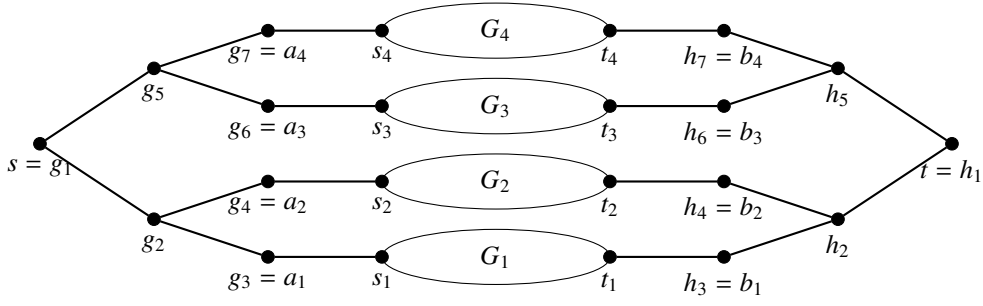
  For \(i\in\{1,\dots,p\}\),
  let
  \((G_i,s_i,t_i,k_i,\ell_i)\)~be
  instances of \sspAcr{}
  such that
  each \(G_i\) is a planar graph
  of maximum degree five
  and has a planar embedding with
  \(s_i\) and~\(t_i\) on the outer face.
  Without loss of generality,
  \(p\)~is a power of two
  (otherwise, we pad the list of input instances with no\hyp instances),
  the vertex sets
  of the graphs~\(G_1,\dots,G_p\) are pairwise disjoint,
  and, for all \(i\in\{1,\dots,p\}\),
  one has \(|V(G_i)|=n\), \(\ell_i=\ell\), and \(k_i=k\)
  (this is a polynomial equivalence relation).
  We construct an instance \((G,s,t,k',\ell')\)
  of \sspAcr{},
  where
  \begin{align*}
    k'&:=k+2\log p, & \ell'&:=\ell + 2\log p-1,
  \end{align*}
  and the graph~\(G\) is as follows.
  Graph~\(G\) consists of
  \(G_1,\dots,G_p\)
  and
  two rooted balanced binary trees~$T_s$ and~$T_t$
  with roots~$s$ and~$t$, respectively, each having~$p$ leaves.
  Let $g_1,\ldots,g_{2p-1}$ and $h_1,\ldots,h_{2p-1}$~denote
  the vertices of~$T_s$ and~$T_t$
  enumerated by a depth-first search starting at~$s$ and~$t$,
  respectively.
  Moreover,
  let $a_1,\ldots,a_p$ and $b_1,\ldots,b_p$~denote
  the leaves of~$T_s$ and~$T_t$
  as enumerated in each depth-first search mentioned before.
  Then,
  for each~$i\in\{1,\dots,p\}$,
  graph~\(G\) contains
  the edges~$\{a_i,s_i\}$ and~$\{b_i,t_i\}$.
  This finishes the construction.
\end{construction}

\begin{proof}[Proof of \cref{thm:nopktwell}]
 \citet{LF18} already proved
 that \cref{constr:twnopkern}
 is a correct cross composition.
 Moreover,
 obviously,
 \(k',\ell'\in O(n+\log p)\).
 Thus,
 to prove \cref{thm:nopktwell},
 we show
 $\tw(G)\leq 3n+3$
 for the graph~\(G\)
 constructed by \cref{constr:twnopkern}.
 To this end,
 we give a tree decomposition of width at most~\(3n+3\) for~$G$
 (recall \cref{def:tw} of tree decompositions)
 as illustrated in \cref{fig:tdc}:
 \begin{figure*}[t]
  \centering
  \begin{tikzpicture}
    \usetikzlibrary{decorations.pathreplacing,calc}
    \tikzstyle{tnode}=[circle, fill, scale=1/2,draw]
    \def\xr{0.88}
    \def\yr{0.4}

    \newcommand{\tree}[4]{
	    \node (r) at (#1,#2)[tnode,label=-90:{$#3$},label=90:{\fbox{$\{#3\}$}}]{};
	    \node (rl) at (#1-2*\xr,#2-2*\yr)[tnode,label=-90:{$#4_2$},label=135:{\fbox{$\{#3,#4_2\}$}}]{};
	    \node (rr) at (#1+2*\xr,#2-2*\yr)[tnode,label=-90:{$#4_5$},label=45:{\fbox{$\{#3,#4_5\}$}}]{};
	    \draw (r) -- (rl);		\draw (r) -- (rr);
	    \node (rll) at (#1-3*\xr,#2-4*\yr)[tnode,label=135:{$#4_3$},label=-90:{\fbox{$\{#4_2,#4_3\}$}}]{};
	    \node (rlr) at (#1-1*\xr,#2-4*\yr)[tnode,label=45:{$#4_4$},label=-90:{\fbox{$\{#4_2,#4_4\}$}}]{};
	    \draw (rl) -- (rll);		\draw (rl) -- (rlr);
	    \node (rrl) at (#1+1*\xr,#2-4*\yr)[tnode,label=135:{$#4_6$},label=-90:{\fbox{$\{#4_5,#4_6\}$}}]{};
	    \node (rrr) at (#1+3*\xr,#2-4*\yr)[tnode,label=45:{$#4_7$},label=-90:{\fbox{$\{#4_5,#4_7\}$}}]{};
	    \draw (rr) -- (rrl);		\draw (rr) -- (rrr);

    }

    \newcommand{\treeX}[6]{
	    \node (r) at (#1,#2)[tnode,label=-90:{$#3$},label=90:{\fbox{$\{#3,#5\}$}}]{};
	    \node (rl) at (#1-4*\xr,#2-2*\yr)[tnode,label=-90:{$#4_2$},label=135:{\fbox{$\{#3,#4_2,#5,#6_2\}$}}]{};
	    \node (rr) at (#1+4*\xr,#2-2*\yr)[tnode,label=-90:{$#4_5$},label=45:{\fbox{$\{#3,#4_5,#5,#6_5\}$}}]{};
	    \draw (r) -- (rl);		\draw (r) -- (rr);
	    \node (rll) at (#1-6*\xr,#2-4*\yr)[tnode,label=135:{$#4_3$},label=-90:{\fbox{$\{#4_2,#4_3,#6_2,#6_3,V_1\}$}}]{};
	    \node (rlr) at (#1-2*\xr,#2-4*\yr)[tnode,label=45:{$#4_4$},label=-90:{\fbox{$\{#4_2,#4_4,#6_2,#6_4,V_2\}$}}]{};
	    \draw (rl) -- (rll);		\draw (rl) -- (rlr);
	    \node (rrl) at (#1+2*\xr,#2-4*\yr)[tnode,label=135:{$#4_6$},label=-90:{\fbox{$\{#4_4,#4_6,#6_4,#6_6,V_3\}$}}]{};
	    \node (rrr) at (#1+6*\xr,#2-4*\yr)[tnode,label=45:{$#4_7$},label=-90:{\fbox{$\{#4_4,#4_7,#6_4,#6_7,V_4\}$}}]{};
	    \draw (rr) -- (rrl);		\draw (rr) -- (rrr);
    }
    \begin{scope}[scale=0.66,transform shape]
      \tree{0}{0}{s}{g}
    \end{scope}
    \begin{scope}[xshift=9*\xr cm,scale=0.66,transform shape]
      \tree{0}{0}{t}{h}
    \end{scope}
    \node at (-3*\xr,1.5*\yr)[]{(a)};
    \node at (6*\xr,1.5*\yr)[]{(b)};
    \node at (-3*\xr,-5*\yr)[]{(c)};
    \treeX{4.5*\xr}{-6.25*\yr}{s}{g}{t}{h}
  \end{tikzpicture}
  \caption{Overview on the tree decompositions (sets in boxes refer to the bags), exemplified for~$p=4$ input instances. (a) and (b) display the tree decomposition for~$T_s$ and~$T_t$, respectively.
  (c) displays the tree decomposition~$\bbT$ (and~$\bbT_{st}$ when removing~$V_i^*$ for \(i\in\{1,\dots,4\}\)). 
  Here, $V_i$~represents the set of vertices in the input graph~$G_i$.
  }
  \label{fig:tdc}
 \end{figure*}
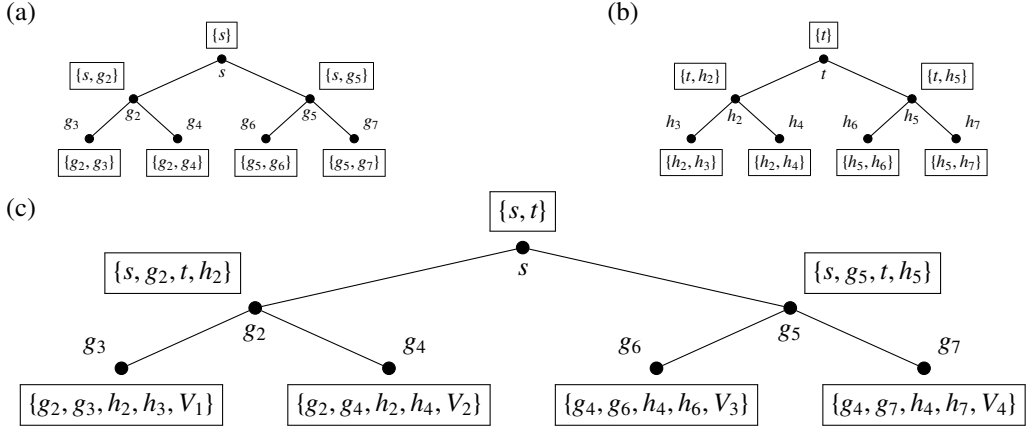

 First,
 we construct a tree decomposition~$\bbT_s=(T_s,\beta_s)$
 of~$T_s$ with bags as follows.
 Let~$\parent_{T_s}(v)$ denote the parent of~$v\in V(T_s)$
 (where~$\parent_{T_s}(s)=s$).
 For each~$v\in V(T_s)$,
 let~$\beta_s(v):=\{v,\parent_{T_s}(v)\}$.
 Then $\bbT_s$~is a tree decomposition of width one.
 Let $\bbT_t=(T_t,\beta_t)$~be
 the tree decomposition for~$T_t$ constructed analogously.

 We now construct
 a tree decomposition~$\bbT_{st}=(T,\beta_{st})$
 for the disjoint union of~$T_s$ and~$T_t$
 as follows:
 take \(T=T_s\)
 and,
 for each \(i\in\{1,\dots,2p-1\}\),
 let \(\beta_{st}(g_i):=\beta_s(g_i)\cup \beta_t(h_i)\),
 where \(g_i\) and \(h_i\)
 are the vertices of~\(T_s\) and~\(T_t\)
 according to the depth-first labeling
 in \cref{constr:twnopkern}.
 As~$\bbT_s$ and~$\bbT_t$
 are tree decompositions of
 two vertex\hyp disjoint
 trees~$T_s$ and~$T_t$, respectively,
 and \(\{g_i,g_j\}\)~is an edge of~\(T_s\)
 if and only if
 \(\{h_i,h_j\}\)~is an edge of~\(T_t\),
 $\bbT_{st}$~is a tree decomposition
 for the disjoint union of~$T_s$ and~$T_t$.
 The width of \(\bbT_{st}\) is three.
 
 Now,
 recall that,
 for \(i\in\{1,\dots,p\}\),
 the graph~$G_i$ in~$G$
 is adjacent to exactly one leaf~\(a_i\) of~$T_s$
 and one leaf~\(b_i\) of~$T_t$
 (via paths on $k$~vertices each).
 Hence,
 we obtain a tree decomposition~$\bbT$ of~$G$
 from~$\bbT_{st}$ by,
 for each \(i\in\{1,\dots,p\}\),
 adding~$V(G_i)$ 
 to bag~$\beta(a_i)$,
 which contains both~\(a_i\) and~\(b_i\).
 The width of~$\bbT$ is at 
 most~$n+3$, 
 and hence, we have~$\tw(G)\leq n+3$.
\end{proof}

\subsection{Subexponential-time algorithms using treewidth}
\label{sec:subexp}
\label{sec:plansubexp}
In this section,
we briefly describe the implications of \cref{thm:twsingexp}
to subexponential\hyp time algorithms for \sspAcr{}
on restricted graph classes.
For example, road networks
have small crossing number.
We can use \cref{thm:twsingexp}
to solve \sspAcr{} in subexponential time
on graphs with \emph{constant} crossing number,
which
are \emph{\(H\)-minor free}
for some graph~\(H\) \citep{BFM06}.

\begin{definition}[graph minor]
  A graph~\(H\)
  is a \emph{minor}
  of a graph~\(G\)
  if
  \(H\)~can be obtained
  from~\(G\)
  by a sequence of
  vertex deletions,
  edge deletions,
  and edge contractions.
  If a graph~\(G\)
  does not contain~\(H\)
  as a minor,
  then \(G\)~is
  said to be
  \emph{\(H\)-minor free}.
\end{definition}

\noindent
\citet{BFM06} showed that,
if a graph~\(G\)
contains \(K_{r,r}\)~as a minor,
then the crossing number of~\(G\)
is \(\crn(G)\geq \frac12(r-2)^2\).
Thus,
any graph~\(G\)
is \(K_{r,r}\)-minor free
for \(r>\sqrt{2\crn(G)}+2\),
which goes in line
with the well\hyp known fact
that planar graphs
are \(K_{3,3}\)-minor free \citep{Wag37}.
\citet{DH08} showed that,
for any graph~\(H\),
all \(H\)-minor free graphs
have treewidth~\(\tw\in O(\sqrt n)\).%
\footnote{In fact,
  they showed \(\tw\in O(\sqrt q)\) for any graph parameter~\(q\)
  that is \(\Omega(p)\) on a \((\sqrt p\times\sqrt p)\)-grid
  and does not increase when taking minors.
For example, the vertex cover number or feedback vertex number.}
Hence,
\cref{thm:twsingexp}
immediately yields:

\begin{corollary}
  \label{thm:plansubexp}
  \sspTsc{} is solvable
  in \(2^{O(\sqrt{n})}\)~time
  on graphs
  with constant crossing number
  (and in all graphs excluding some fixed minor).
\end{corollary}
\noindent
Complementing \cref{thm:plansubexp},
we now show that,
unless the Exponential Time Hypothesis (ETH) fails,
\cref{thm:plansubexp}
can be neither significantly improved
in planar graphs
nor generalized to general graphs.

\begin{hypothesis}[Exponential Time Hypothesis (ETH),
  \citet{IPZ01}]
  There is a constant~\(c\) such that
  \(n\)-variable \textsc{3-Sat}
  cannot be solved in \(2^{c(n+m)}\)~time.
\end{hypothesis}

\noindent
The ETH was introduced by
\citet{IPZ01}
and since then
has been used
to prove running time lower bounds
for various NP-hard problems
(we refer to \citet[Chapter~14]{CFK+15}
for an overview).
We use it to prove the following.

\begin{observation}
  \label{thm:ethlb}
  Unless the Exponential Time Hypothesis fails,
  \sspTsc{} has no \(2^{o(\sqrt n)}\)-time algorithm
  in planar graphs
  and no \(2^{o(n+m)}\)-time algorithm in general.
\end{observation}

\begin{proof}
  Assume that there is
  (i) a \(2^{o(\sqrt n)}\)-time algorithm
  for \sspAcr{} in planar graphs
  or 
  (ii) a \(2^{o(n+m)}\)-time algorithm
  for \sspAcr{} in general graphs.
  \citet{LF18} give
  a polynomial\hyp time many\hyp one reduction
  from \textsc{Hamiltonian Path}
  to \sspAcr{}
  that maintains planarity
  and increases
  the number of vertices and edges
  by at most a constant.
  Thus,
  in case of (i) 
  we get a \(2^{o(\sqrt{n})}\)-time algorithm
  for \textsc{Hamiltonian Path}
  in planar graphs
  and in case of (ii) 
  we get a \(2^{o(n+m)}\)-time algorithm
  in general graphs.
  This contradicts ETH
  \citep[Theorems~14.6 and~14.9]{CFK+15}.
\end{proof}

\section{Graphs with small feedback sets}
\label{sec:fes}
In the previous section,
we studied the parameterized complexity of \sspAcr{}
with respect to the treewidth parameter.
and showed that \sspAcr{}
has no problem kernel with size polynomial
in the treewidth of the input graph.
In this section,
we study polynomial\hyp size
kernelizability of \sspAcr{}
with respect to other parameters that measure
the tree\hyp likeness of a graph:
the \emph{feedback vertex number~\(\fvs\)}
and \emph{feedback edge number~\(\fes\)}.
Graphs in which these parameters are small
arise as waterways:
ignoring the few man-made canals,
natural river networks form forests \citep{Gia10}.

In \cref{sec:feskern},
we prove problem kernels with size 
polynomial in~\(\fvs+k+\ell\) (\cref{sssec:fvskell}) and 
polynomial in~\(\fes\) (\cref{sssec:fes}).
In \cref{sec:fvsnokern},
we show that
\sspAcr{} does not allow for problem kernels
of size polynomial in~\(\fvs+\ell\)
unless $\unlessPK$. %

\subsection{Efficient data reduction}
\label{sec:feskern}
In this section,
we prove problem kernels with size 
polynomial in~\(\fvs+k+\ell\) (\cref{sssec:fvskell}) and 
polynomial in~\(\fes\) (\cref{sssec:fes}).
Our data reduction rules
will delete leaves from trees
and shrink their 
paths consisting of degree-two vertices,
storing information about the changes in vertex weights.
That is,
the result of our data reduction will be
an instance of \WsspTsc{}.
These instances,
which we will call \emph{\simple{}},
satisfy a set of properties
that allow us to strip them of weights efficiently.

\begin{definition}
  \label{def:simplevwssp}
 An instance $(G,s,t,k,\ell,\lambda,\kappa,\eta)$
 of \WsspAcr{} with \(G=(V,E)\) is called \emph{\simple} 
 if there is a set~\(A\subseteq V\) such that
 \begin{enumerate}[(i)]
 \item $\kappa(s)=\kappa(t)=1$,
 \item $\lambda(v)=1$ for all~$v\in V$,

 \item $\eta(v)>\ell$ and~$\kappa(v)=1$ for all~$v\in A$, and
 \item in \(G-A\), every vertex~$v$ with $\kappa(v)>1$
   has exactly two neighbors~\(u\) and~\(w\),
   which have degree at most two, are distinct from~$s$ and~$t$, and~\(\kappa(u)=\kappa(w)=1\).
 \end{enumerate}
\end{definition}

\noindent
For the sake of readability, for any set~$W$ of vertices we also write~$\kappa(W)$ for~$\sum_{v\in W}\kappa(v)$ (analogously for~$\eta$).
Next we show that for any given simple instance of \WsspAcr{}, we can compute in linear time an equivalent instance of \sspAcr{} whose number of vertices only depends on~$\kappa$ and~$\eta$.
\begin{proposition}
  \label{prop:Hssp2ssp}
  Any \simple{} instance~$(G,s,t,k,\ell,\lambda,\kappa,\eta)$
  of \WsspAcr{}
  with \(G=(V,E)\) and given~\(A\subseteq V\)
  can be reduced to an equivalent instance of \sspAcr{}
  with at most $M:= \kappa(V) + \eta(V)$~vertices
  in time linear in~$M+|E|$.
\end{proposition}
  
\noindent
To prove \cref{prop:Hssp2ssp}, we use the following construction.

\begin{construction}
  \label{constr:Hssp2ssp}
  Let $(G,s,t,k,\ell,\lambda,\kappa,\eta)$~be
  a \simple{} instance of \WsspAcr{}
  with \(G=(V,E)\)
  and given set~$A\subseteq V$ as in~\cref{def:simplevwssp}.
  Construct an instance~$(G',s',t',k,\ell)$ of \sspAcr{}
  as follows (see \cref{fig:Hssp2ssp} for an illustrative example).
  Let $G'$~be initially a copy of~$G$.
  For each~$v\in V$ with~$\kappa(v)>1$,
  let~$\{v',v''\}=N_{G-A}(v)$,
  replace~$v$ by a path~$P_{v}$ with
  $\kappa(v)$~vertices,
  make one endpoint adjacent to~$v'$,
  and the other endpoint adjacent to~$v''$.
  Next, for each~$v\in V$, add a set~\(U_v\) of
  $\eta(v)$~vertices.
  If~$\kappa(v)=1$ make each~$u\in U_v$ only adjacent to~$v$.
  If~$\kappa(v)>1$ make each~$u\in U_v$ only adjacent to some vertex~$x$ on~$P_{v}$.
  Finally, for each~$v\in V\setminus(A\cup\{s,t\})$ with~$\kappa(v)>1$ and~$A_v:=N_G(v)\cap A\neq \emptyset$, make each~$w\in A_v$ adjacent with some vertex~$x$ on~$P_v$.
  This finishes the construction of~$G'$.
  Observe that the construction can be done in time linear
  in~$M+|E|$ and~$(G',s,t,k,\ell)$ consists of~$M$ vertices.
\end{construction}
\begin{figure}
  \centering
 \begin{tikzpicture}
    \tikzstyle{anode}=[rectangle, fill,scale=0.6, draw]
    \tikzstyle{xnode}=[circle, fill, scale=0.6, draw]
    \tikzstyle{xxnode}=[rectangle, fill, scale=0.6, draw]
    \tikzstyle{knode}=[rectangle, scale=0.7, draw]
    \tikzstyle{kxnode}=[circle, scale=0.7, draw]
    \tikzstyle{xedge}=[-,thick]

    \def\xr{1.0}
    \def\yr{0.9}

    \newcommand{\tovwmakenodes}{
      \draw[rounded corners,gray] (0,0.25) rectangle (3*\xr,0.75*\yr);
      \foreach \x in {1,3,5}\node (a\x) at (0.5*\x*\xr,0.5*\yr)[anode]{};
        \node (x0) at (-1*\xr,1*\yr)[xxnode,label=180:{$s$}]{};	
        \node (x1) at (-1*\xr,0*\yr)[xnode]{};	
      \node (x3) at (0*\xr,-2*\yr)[xnode]{};
      \node (x4) at (0.5*\xr,-2*\yr)[xnode]{};	
      \node (x5) at (1.5*\xr,-0.5*\yr)[xnode]{};
      \node (x6) at (1.5*\xr,-2*\yr)[xxnode]{};	
      \node (x7) at (2.5*\xr,-2*\yr)[xnode]{};
      \node (x8) at (3*\xr,-2*\yr)[xnode]{};	
      \node (x10) at (4*\xr,0*\yr)[xnode]{};
      \node (x11) at (4*\xr,1*\yr)[xnode,label=180:{$t$}]{};
    }
    \newcommand{\tovwmakeedges}{
      \draw[xedge] (x0) -- (x1) -- (x22) -- (x3) -- (x4) -- (x5) -- (x7) -- (x8) -- (x92) -- (x10) -- (x11);
      \draw[xedge] (x1) -- (a1);
      \draw[xedge] (x4) -- (x6) -- (x7);
      \draw[xedge] (a3) -- (x5) -- (a5); 
      \draw[xedge] (a5) -- (x10) ;
    }

    \begin{scope}[xshift=-1*\xr cm]
      \tovwmakenodes{};
      \node (x22) at (-0.5*\xr,-1*\yr)[knode]{};
      \node (x92) at (3.5*\xr,-1*\yr)[kxnode]{};
      \tovwmakeedges{}
      \draw[xedge] (a1) -- (x22) -- (a3);
      
      \begin{scope}[xshift=4.5*\xr cm, yshift=0.25*\yr cm]
        \draw[dashed, gray, rounded corners] (0,-0.7) rectangle (2.3*\xr, 1.4*\yr);
        \node at (0.25*\xr,1.*\yr)[xnode,label=0:{$\eta=0$, $\kappa=1$}]{};
        \node at (0.25*\xr,0.5*\yr)[xxnode,label=0:{$\eta>0$, $\kappa=1$}]{};
        \node at (0.25*\xr,0*\yr)[kxnode,label=0:{$\eta=0$, $\kappa>1$}]{};
        \node at (0.25*\xr,-0.5*\yr)[knode,label=0:{$\eta>0$, $\kappa>1$}]{};
        
      \end{scope}
    \end{scope}
    \node at (5.5*\xr,-1*\yr)[scale=1.5]{$\leadsto$};
    \tikzstyle{xxnode}=[xnode]
    \tikzstyle{anode}=[xnode]
    \tikzstyle{mygray}=[fill=gray!50!white]
    \newcommand{\tovwthesets}[5]{
      \node (y1) at ($(#1)+(#2)$)[scale=3/4,color=gray!50!black]{#5};
      \node (y2) at ($(y1)+(#3)$)[xnode,mygray]{};
      \node (y3) at ($(y1)+(#4)$)[xnode,mygray]{};
      \draw[xedge] (y2) -- (#1) -- (y3);
    }
    \newcommand{\tovwlargerkappa}{
      \foreach \x in {2,3,5,6}\node (x2\x) at (-1.075*\xr+0.15*\x*\xr,0.15*\yr-0.3*\x*\yr)[xnode,mygray]{};
      \foreach \x in {1,2,3}\node (x2x\x) at (-0.5*\xr-0.075*\xr+0.05*\x*\xr,-1*\yr-0.1*\x*\yr+0.15*\yr)[xnode,mygray,scale=0.4]{};
      \foreach \x in {2,3,5,6}\node (x9\x) at (2.925*\xr+0.15*\x*\xr,-2.15*\yr+0.3*\x*\yr)[xnode,mygray]{};
      \foreach \x in {1,2,3}\node (x9x\x) at (3.5*\xr-0.075*\xr+0.05*\x*\xr,-1*\yr+0.1*\x*\yr-0.15*\yr)[xnode,mygray,scale=0.4]{};
    }
    \begin{scope}[xshift=9*\xr cm]
      \tovwmakenodes{};
      \tovwlargerkappa{}
      \tovwmakeedges{};
      \tovwlargerkappa{}
      \draw[xedge] (a1) -- (x23) -- (a3);
        \tovwthesets{x0}{0,0.75*\yr}{-0.35*\xr,0}{0.35*\xr,0}{$\cdots$};
        \tovwthesets{x25}{-0.75*\xr,-0.25*\yr}{-0*\xr,0.25*\yr}{0*\xr,-0.4*\yr}{$\vdots$};
        \tovwthesets{a1}{0,0.75*\yr}{-0.35*\xr,0}{0.35*\xr,0}{$\cdots$};
        \tovwthesets{a3}{0,0.75*\yr}{-0.35*\xr,0}{0.35*\xr,0}{$\cdots$};
        \tovwthesets{a5}{0,0.75*\yr}{-0.35*\xr,0}{0.35*\xr,0}{$\cdots$};
        \tovwthesets{x6}{0,0.5*\yr}{-0.35*\xr,0}{0.35*\xr,0}{$\cdots$};
    \end{scope}
    \end{tikzpicture}
    \caption{Illustrative example to \cref{constr:Hssp2ssp}.
    On the left-hand side: an input graph with vertex weights,
    on the right-hand side: the graph obtained after applying \cref{constr:Hssp2ssp} (gray vertices indicate added vertices).
    Vertices enclosed in the gray solid rectangle form the set~$A$.
    }
    \label{fig:Hssp2ssp}
\end{figure}
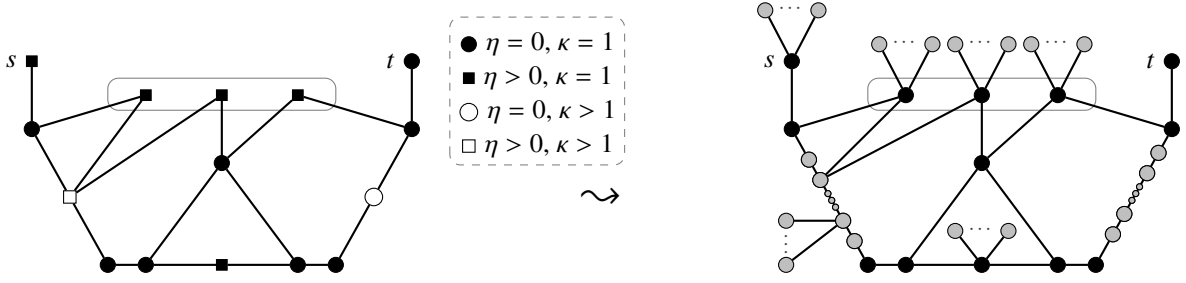

\begin{proof}[Proof of \cref{prop:Hssp2ssp}]
  Let $(G,s,t,k,\ell,\lambda,\kappa,\eta)$~be
  a \simple{} instance of \WsspAcr{}
  with \(G=(V,E)\)
  and given set~$A\subseteq V$ as in~\cref{def:simplevwssp}.
  Apply \cref{constr:Hssp2ssp} to compute instance $I'=(G',s,t,k,\ell)$ of~\sspAcr{} with at most~$M:= \kappa(V) + \eta(V)$ vertices in time linear in~$M+|E|$.
 We claim that~$I$ is a \yes-instance if and only if~$I'$ is a \yes-instance.

 \RD{}
 Let~$I$ be a \yes-instance and $P:=(v_1,v_2,\ldots,v_q)$
 with \(v_1=s\) and \(v_q=t\)~be
 a solution $s$-$t$-path.
 Let $W\subseteq V(P)$~denote the vertices in~$P$ with~$\kappa(v)>1$.
 We claim that the path~$P'$ obtained from~$P$ by replacing each vertex~$v\in W$ by~$P_v$ is a solution $s$-$t$-path to~$I'$.
 First, observe that~$|V(P')|=|V(P)\setminus W|+\kappa(W)\leq k$.
 It remains to prove (recall that $\lambda(v)=1$ for all~$v\in V$)
 \begin{align*}
  |N_{G'}(V(P'))|	&= |N_{G}(V(P))|+\smashoperator{\sum_{v\in V(P')}} |U_{v}| = |N_{G}(V(P))|+\smashoperator{\sum_{v\in V(P)}} \eta(v) \leq \ell.
 \end{align*}
 To this end,
 it is enough to prove
 $N_{G'}(V(P'))= N_{G}(V(P))\uplus \biguplus_{v\in V(P')} U_{v}$.
 First observe that
 no vertex in~\(A\) is in~$V(P)$
 since~$\eta(v)>\ell$ for all~$v\in A$.
 Thus,
 no vertex in~$A$ is contained in~$V(P')$.
 For each~$v\in W$ let~$v'$ and~$v''$
 be the only two neighbors of~\(v\) in~$G-A$.
 Then,
 for each~$v\in W$,
 we have $N_{G'}(V(P_{v}))\setminus \{v',v''\} = N_G(v)\setminus  \{v',v''\}$,
 since the neighbors of~$v$ in~$A$ coincide with the neighbors of~$V(P_{v})$ in~$A$.
 Thus,
 \begin{align*}
  N_{G'}(V(P'))	&= \Biggl(N_{G'}(V(P)\setminus W)\setminus \bigcup_{v\in W} V(P_{v})\Biggr)
                    \cup \bigcup_{v\in W} \left(N_{G'}(V(P_{v}))\setminus \{v',v''\}\right) 
                    \uplus \biguplus_{v\in V(P')} U_{v} \\
                &= \left(N_{G}(V(P)\setminus W)\setminus W\right)\cup \bigcup_{v\in W} \left(N_{G}(v)\setminus \{v',v''\}\right) \uplus \biguplus_{v\in V(P')} U_{v} \\\
                &= N_{G}(V(P))\uplus \biguplus_{v\in V(P')} U_{v}.
 \end{align*}
 
 \LD{}
 Let $I'$~be a \yes-instance and
 let $P'$~be a solution $s$-$t$-path.
 Note that all vertices in~$P_{v}$
 for $v\in V$ with~$\kappa(v)>1$
 are of degree two in~$G'-(A\cup U_v)$ and distinct from~$s$ and~$t$.
 Hence, if a vertex of~$P_v$ is contained in~$P'$, then all vertices from~$P_v$ are contained in~$P'$.
 Let~$W\subseteq V$ denote the set of vertices~$v$ with~$\kappa(v)>1$ such that~$P_v$ is a subpath of~$P'$.
 We claim that the path~$P$ obtained from~$P'$ by replacing each path~$P_v$ by~$v\in W$
 is a solution $s$-$t$ path for~$I$.
 First, observe that~$\kappa(V(P))=|V(P')|-\sum_{v\in W}|V(P_v)|+\kappa(W)=|V(P')|\leq k$.
 Second, similarly as in the above
 direction, since~$I$ is \simple{},
 we have
 \begin{align*}
  |N_{G}(V(P))|+\sum_{v\in V(P)} \eta(v)=|N_{G}(V(P))|+\sum_{v\in V(P)} |U_{v}|=|N_{G'}(V(P'))|\leq \ell. &\qedhere
 \end{align*}
\end{proof}

\subsubsection{A problem kernel with $O(\fvs\cdot(k+\ell)^2)$ vertices}
\label{sssec:fvskell}
We show that \sspAcr{} admits a
problem kernel with a number of vertices
cubic in the  parameter~$\fvs+\ell+k$.

\begin{theorem}
  \label{thm:ssppk-fvskell}
  \sspTsc{} admits a problem kernel
  of size polynomial in $\fvs+k+\ell$
  with~$O(\fvs\cdot(k+\ell)^2)$ vertices.
\end{theorem}

\noindent
In a nutshell, we will construct a simple instance of \WsspAcr{} by shrinking the number and sizes of the trees in the graph after removing a feedback vertex set.
Herein, we store information on each shrinking step in the vertex weights.
To shrink the sizes of the trees, we delete leaves (as their number upper-bounds the number of vertices of degree at least three) and replace maximal paths consisting of degree-two vertices by shorter paths.
The number of vertices and edges as well as the vertex weights will be upper-bounded polynomially in~$\fvs+k+\ell$.
Finally, we employ \cref{prop:Hssp2ssp} on the simple instance of \WsspAcr{} to compute an instance of~\sspAcr{} where the number of vertices and edges is upper-bounded polynomially in~$\fvs+k+\ell$.

Let $G=(V,E)$~be the input graph
with $V=F\uplus W$
such that $F$~is a feedback vertex set
with $s,t\in F$ (hence, $G[W]$ is a forest).
Let~$\beta:=|F|$.
We distinguish the following
types of vertices of~$G$
(see \cref{fig:rrrulesfvskell} for an illustration).

\begin{itemize}[$Y\subseteq W$]
\item[$R\subseteq F$] is the subset of vertices in~$F$ with more than~$k+\ell$ neighbors or more than~$\ell$ degree\hyp one neighbors.
  Since no vertex of~$R$ is part of any solution path, we refer to the vertices in~$R$ as~\emph{forbidden}.
  
\item[$Y\subseteq W$] is the subset of vertices in~$W$ containing all vertices~$v$ with 
  $N(v) \cap F \nsubseteq R$,
  that is, vertices that have at least one neighbor in~$F$ that is not forbidden.
  We call the vertices in~$Y$ \emph{good}.
  
\item[$\calT$] is the set of
  connected components of~$H:=G[W]$,
  all of which are trees.
\end{itemize}
\noindent
Towards proving \cref{thm:ssppk-fvskell},
we will first prove the following,
and then strip the weights using \cref{prop:Hssp2ssp}.

\begin{proposition}
  \label{prop:sspToHssp}
  For any instance of~\sspAcr{} we can compute in polynomial time an equivalent \simple{} instance of~\WsspAcr{} with~$O(\fvs\cdot(k+\ell))$~vertices, $O(\fvs^2\cdot (k+\ell))$~edges and vertex weights $O(k+\ell)$.
\end{proposition}

\noindent
We will interpret the input \sspAcr{}
instance as an instance of \WsspAcr{}
with unit weight functions~\(\kappa\) and~\(\lambda\)
and the zero weight function~$\eta$.
For an exemplified illustration of the
following reduction rules,
we refer to \cref{fig:rrrulesfvskell}.
The first reduction rule ensures that each forbidden vertex remains forbidden
throughout our application of all reduction rules.
It is clearly applicable in linear time.

\begin{rrule}
  \label{rr:redisred}
  For each~$v\in R$, set~$\eta(v)=\ell+1$.
\end{rrule}

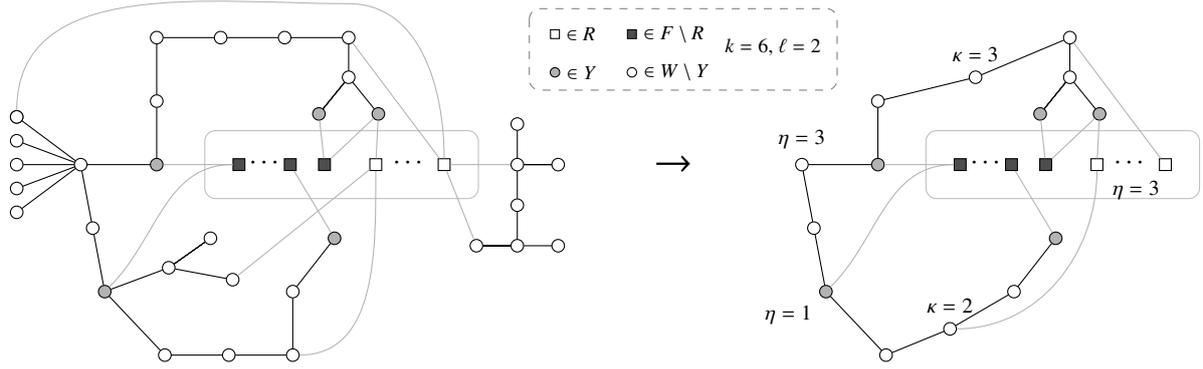
\begin{figure}
 \centering
 \begin{tikzpicture}

  \usetikzlibrary{arrows,patterns,calc}

  \def\xr{0.9}
  \def\yr{0.9}

  \tikzstyle{xnode}=[circle, scale=1/2, fill, draw];
  \tikzstyle{rnode}=[scale=2/3, fill=white, draw];
  \tikzstyle{nrnode}=[scale=2/3, fill=black!70!white, draw=black];
  \tikzstyle{gnode}=[xnode,fill=black!30!white];
  \tikzstyle{nnode}=[xnode,fill=white];
  \tikzstyle{bnode}=[xnode,fill=white];
  \tikzstyle{gedge}=[color=black!30!white];
  \tikzstyle{bedge}=[color=black!30!white];%

  \begin{scope}[xshift=-30*\xr,yshift=0*\yr]%
	  \node (r) at (0,0)[gnode]{};	
	  \node (x1) at (-1,0)[nnode]{};
	  \draw (r) -- (x1);
	  \foreach \x in {-20,-10,...,20}{
		  \node (x11) [left of=x1,xshift=5*\xr,yshift=\x*\yr,nnode]{};
		  \draw (x11) -- (x1);
	  }
	  \node (x11) [left of=x1,xshift=5*\xr,yshift=20*\yr,bnode]{};

	  \node (x2) [below of=x1,xshift=5*\xr,yshift=5*\yr,nnode]{};
	  \node (x3) [below of=x2,xshift=5*\xr,yshift=5*\yr,gnode]{};
		  \node (x31) [right of=x3,,xshift=-5*\xr,yshift=10*\yr,nnode]{};
		  \node (x32) [above right of=x31,xshift=-5*\xr,yshift=-10*\yr,nnode]{};
		  \node (x33) [right of=x31,xshift=-5*\xr,yshift=-5*\yr,bnode]{};
	  \node (x4) [below of=x3,xshift=25*\xr,yshift=5*\yr,nnode]{};
	  \node (x5) [right of=x4,xshift=-5*\xr,yshift=0*\yr,nnode]{};
	  \node (x6) [right of=x5,xshift=-5*\xr,yshift=0*\yr,bnode]{};
	  \node (x7) [above of=x6,xshift=0*\xr,yshift=-5*\yr,nnode]{};
	  \node (x8) [above right of=x7,xshift=-5*\xr,yshift=0*\yr,gnode]{};
	  \draw (x1) -- (x2) -- (x3) -- (x4) -- (x5) -- (x6) -- (x7) -- (x8);
	  \draw (x3) -- (x31) -- (x32) -- (x31) -- (x33);
	  \node (y1) [above of=r,,yshift=-5*\yr,nnode]{};
	  \node (y2) [above of=y1, yshift=-5*\yr,nnode]{};
	  \node (y3) [right of=y2,xshift=-5*\xr,nnode]{};
	  \node (y4) [right of=y3,xshift=-5*\xr,nnode]{};
	  \node (y5) [right of=y4,xshift=-5*\xr,bnode]{};
	  \node (y6) [below of=y5,yshift=15*\xr,nnode]{};
	  \node (y61) [below left of=y6,xshift=10*\xr,yshift=7*\yr,gnode]{};
	  \node (y62) [below right of=y6,xshift=-10*\xr,yshift=7*\yr,gnode]{};
	  \draw (r) -- (y1) -- (y2) -- (y3) -- (y4) -- (y5) -- (y6) -- (y61) -- (y6) -- (y62);

  \end{scope}

  \begin{scope}[xshift=120*\xr,node distance=17*\xr and 17*\yr]
	  \node (d1) at (0,0)[bnode]{};
	  \node (d2) [above of=d1,nnode]{};
	  \node (d3) [right of=d1,nnode]{};
	  \node (d4) [below of=d1,nnode]{};
	  \node (d5) [below of=d4,nnode]{};
	  \node (d6) [left of=d5,bnode]{};
	  \node (d7) [right of=d5,nnode]{};
	  \draw (d2) -- (d1) -- (d3) -- (d1) -- (d4) -- (d5) -- (d6) -- (d5) -- (d7);
  \end{scope}

  \begin{scope}[xshift=-10*\xr]
	  \draw[rounded corners,lightgray] (0,-0.5*\yr) rectangle (4*\xr,0.5*\yr);
	  \node (g1) at (0.5*\xr,0*\yr)[nrnode]{};
	  \node (g2) at (0.9*\xr,0*\yr)[]{$\cdots$};
	  \node (gx1) at (1.25*\xr,0*\yr)[nrnode]{};
	  \node (gx) at (1.75*\xr,0*\yr)[nrnode]{};
	  \node (r1) at (2.5*\xr,0*\yr)[rnode]{};
	  \node (r2) at (3*\xr,0*\yr)[]{$\cdots$};
	  \node (rx) at (3.5*\xr,0*\yr)[rnode]{};
	  \draw[gedge] (y61) -- (gx) -- (y62);
	  \draw[gedge] (gx1) -- (x8);
	  \draw[gedge] (g1) to [out=180,in=45](x3);
	  \draw[gedge] (g1) to (r);
	  \draw[bedge] (r1) to (x33);
	  \draw[bedge] (r1) to (y62);
	  \draw[bedge] (r1) to [out=-90,in=0](x6);
	  \draw[bedge] (d6) -- (rx)  -- (d1);
	  \draw[bedge] (rx) to (y5);
	  \draw[bedge] (rx) to [out=90,in=0]($(y5)+(0*\xr,0.5*\yr)$) to [out=180,in=90](x11);
  \end{scope}

  \node at (6.5*\xr,0)[scale=1.5]{$\to$};

  \begin{scope}[xshift=-30*\xr+300*\xr,yshift=0*\yr]%
	  \node (r) at (0,0)[gnode]{};	
	  \node (x1) at (-1,0)[nnode,label=90:{\footnotesize $\eta=3$}]{};
	  \draw (r) -- (x1);

	  \node (x2) [below of=x1,xshift=5*\xr,yshift=5*\yr,nnode]{};
	  \node (x3) [below of=x2,xshift=5*\xr,yshift=5*\yr,gnode,label=-135:{\footnotesize $\eta=1$}]{};
	  \node (x4) [below of=x3,xshift=25*\xr,yshift=5*\yr,nnode]{};
	  \node (x5) [right of=x4,xshift=-5*\xr,yshift=0*\yr]{};
	  \node (x6) [right of=x5,xshift=-5*\xr,yshift=0*\yr]{};
	  \node (x56) [right of=x4,xshift=-5*\xr,yshift=11*\yr,bnode,label=90:{\footnotesize $\kappa=2$}]{};
	  \node (x7) [above of=x6,xshift=0*\xr,yshift=-5*\yr,nnode]{};
	  \node (x8) [above right of=x7,xshift=-5*\xr,yshift=0*\yr,gnode]{};
	  \draw (x1) -- (x2) -- (x3) -- (x4) -- (x56) -- (x7) -- (x8);
	  \node (y1) [above of=r,,yshift=-5*\yr,nnode]{};
	  \node (y2) [above of=y1, yshift=-5*\yr]{};
	  \node (y3) [right of=y2,xshift=-5*\xr]{};
	  \node (y4) [right of=y3,xshift=-5*\xr]{};
	  \node (y3x) [below of=y3,xshift=14*\xr,yshift=15*\yr,nnode,label=90:{\footnotesize $\kappa=3$}]{};
	  \node (y5) [right of=y4,xshift=-5*\xr,nnode]{};
	  \node (y6) [below of=y5,yshift=15*\xr,nnode]{};
	  \node (y61) [below left of=y6,xshift=10*\xr,yshift=7*\yr,gnode]{};
	  \node (y62) [below right of=y6,xshift=-10*\xr,yshift=7*\yr,gnode]{};
	  \draw (r) -- (y1) -- (y3x) -- (y5) -- (y6);
	  \draw (y6) -- (y61) -- (y6) -- (y62);

  \end{scope}

  \begin{scope}[xshift=-10*\xr+300*\xr]
	  \draw[rounded corners,lightgray] (0,-0.5*\yr) rectangle (4*\xr,0.5*\yr);
	  \node (g1) at (0.5*\xr,0*\yr)[nrnode]{};
	  \node (g2) at (0.9*\xr,0*\yr)[]{$\cdots$};
	  \node (gx1) at (1.25*\xr,0*\yr)[nrnode]{};
	  \node (gx) at (1.75*\xr,0*\yr)[nrnode]{};
	  \node (r1) at (2.5*\xr,0*\yr)[rnode,label=-45:{\footnotesize $\eta=3$}]{};
	  \node (r2) at (3*\xr,0*\yr)[]{$\cdots$};
	  \node (rx) at (3.5*\xr,0*\yr)[rnode]{};
	  \draw[gedge] (y61) -- (gx) -- (y62);
	  \draw[gedge] (gx1) -- (x8);
	  \draw[gedge] (g1) to [out=180,in=45](x3);
	  \draw[gedge] (g1) to (r);
	  \draw[bedge] (r1) to (y62);
	  \draw[bedge] (r1) to [out=-90,in=0](x56);
	  \draw[bedge] (rx) to (y5);
  \end{scope}

  \begin{scope}[xshift=125*\xr,yshift=-20*\yr,scale=0.75,transform shape]
	  \draw[rounded corners,dashed, gray] (0,4*\yr) rectangle (6*\xr,2.4*\yr);
      \node at (0.5*\xr,3.5*\yr)[rnode,label=0:{$\in R$}]{};
	  \node at (2.0*\xr,3.5*\yr)[nrnode,label=0:{$\in F\setminus R$}]{};
	  \node at (2.0*\xr,2.75*\yr)[nnode,label=0:{$\in W\setminus Y$}]{};
	  \node at (0.5*\xr,2.75*\yr)[gnode,label=0:{$\in Y$}]{};
	  \node at (4.75*\xr,3.25*\yr)[]{$k=6$, $\ell=2$};
  \end{scope}

  \end{tikzpicture}
  \caption{Exemplified illustration of
    the partition into
    forbidden (white square, adjacent degree-one vertices are omitted)
    and good (gray round) vertices,
    and for the application of \cref{rr:killingtree,rr:burntrees,rr:shrinkedgypaths}. 
    The vertices enclosed in the light-gray rectangle
    are all vertices in the feedback vertex set~$F$.}
  \label{fig:rrrulesfvskell}
\end{figure}

\noindent
Since each vertex in~$F\setminus R$
has degree at most~$k+\ell$,
by the definition of good vertices,
we have the following.
\begin{observation}
  \label{obs:fewgoodvertices}
  The number of good vertices is $|Y|\leq \beta(k+\ell)$.
\end{observation}

\noindent
Since a solution path
has neither vertices
nor neighbors
in any tree~\(T\in\calT\)
that does not contain vertices of~\(Y\),
we delete such trees.

\begin{rrule}
  \label{rr:killingtree}
  Delete all trees~$T\in\calT$
  with~$V(T)\cap Y=\emptyset$.
\end{rrule}

\noindent
Note that if
\cref{rr:killingtree} is not applicable,
then each tree in~\(\calT\)
contains a vertex from~\(Y\),
which gives $|\calT|\leq \beta(k+\ell)$
together with \cref{obs:fewgoodvertices}.

The following
data reduction rule
deletes degree-one vertices
in trees
that are not in~\(Y\),
since they
cannot be part of a solution path
(yet can neighbor it).

\begin{rrule}
 \label{rr:burntrees}
 If there is a tree~$T\in\calT$
 and $v\in V(T)\setminus Y$
 with $N_T(v)=\{w\}$,
 then
 set $\eta(w):=\min\{\ell+1,\eta(w)+1\}$ and delete~$v$. 
\end{rrule}

\noindent
Note that updating~$\eta(w)$ to the minimum of~$\ell+1$ and~$\eta(w)+1$ is correct: if a vertex has any weight at least~$\ell+1$, the vertex is equally excluded from any solution path as having weight~$\ell+1$.

\begin{proof}[Correctness proof]
 Let~$I=(G,s,t,k,\ell,\lambda,\kappa,\eta)$ be an instance of~\WsspAcr{} and let~$I'=(G',s,t,k,\ell,\lambda,\kappa,\eta')$ be the instance of \WsspAcr{} obtained from applying \cref{rr:burntrees}.
 Let~$v\in V(T)\setminus Y$ with~$T\in \calT$ be the vertex deleted by the application of \cref{rr:burntrees}.
 We claim that~$I$ is a \yes-instance if and only if~$I'$ is a \yes-instance.

 \looseness=-1
 \RD{}
 Let $P$~be a solution $s$-$t$-path in~$G$.
 Since $v\not\in Y$, we know that~$v$ is different from~$s$ and~$t$ and that~$\{w\}=N_T(v)=N_G(v)$.
 Hence, $v\not\in V(P)$.
 If $w\not\in V(P)$, then $P$~is a solution $s$-$t$ path in~$G'$.
 If $w\in V(P)$, then $\eta'(w)=\eta(w)+1\leq \ell$
 and~$\sum_{x\in V(P)} \eta'(x)+|N_{G'}(V(P))| = \sum_{x\in V(P)\setminus\{w\}} \eta(x)+\eta(w)+1+|N_{G}(V(P))| -1 \leq \ell$. 
 Hence, $P$ is a solution $s$-$t$ path in~$G'$.
 
 \LD{}
 Let~$P$ be a solution~$s$-$t$ path in~$G'$.
 Since~$G'= G-\{v\}$, we know that~$P$ is an $s$-$t$ path in~$G$ with~$\sum_{v\in V(P)} \kappa(v)\leq k$.
 If~$w\not\in V(P)$, then~$P$ is a solution $s$-$t$ path in~$G$.
 If~$w\in V(P)$, then~$\eta(w)+1=\eta'(w)\leq \ell$ and~$\sum_{x\in V(P)} \eta(x)+|N_{G}(V(P))| = \sum_{x\in V(P)} \eta'(x) + 1 +|N_{G'}(V(P))| -1 \leq \ell$. 
 Hence, $P$ is a solution $s$-$t$ path in~$G$.
\end{proof}

\begin{lemma}
 \label{obs:firstrrulesfast}
 \cref{rr:killingtree,rr:burntrees} are exhaustively applicable in linear time.
\end{lemma}

\begin{proof}
  For each tree~$T$ in~$G-F$, do the following.
  As long as there is a degree\hyp one vertex~$v\in V(T)\setminus Y$,
  that is, it is not good, delete~$v$.
  This is clearly doable in~$O(|V(T)|)$-time.
  When no vertex remains, apply \cref{rr:killingtree}.
  Otherwise, \cref{rr:burntrees} is exhaustively applied on~$T$.
  Since~$\sum_{T\in\calT} |V(T)|=|W|$, the claim follows.
\end{proof}

\noindent
If none of \cref{rr:killingtree,rr:burntrees} is applicable, the leaves of each tree are all good vertices.
Recall that the number of leaves in a tree upper-bounds the number of vertices of degree at least three in the tree.
Hence, to upper-bound the number of vertices in the trees, it remains to upper-bound the number of degree-two vertices in the tree.
The next data reduction rule deletes these degree-two vertices by shrinking so-called \emph{maximal\hyp edgy paths}.
\begin{definition}
  \label{def:edgy}
  We call an $a$-$b$-path in a tree~$T$
  \emph{edgy} if it contains no good vertex and no vertex~$w$ with~$\deg_T(w)\geq 3$.
  We call an $a$-$b$-path~$Q$ \emph{maximal\hyp edgy} if there is no edgy path containing~$Q$ with more vertices than~\(Q\).
\end{definition}

\begin{rrule}
 \label{rr:shrinkedgypaths}
 Let~$T\in\calT$ and let~$Q\subseteq T$ be a maximal\hyp edgy $a$-$b$-path in~$T$ with~$|V(Q)|>3$.
 Let~$K:=V(Q)\setminus\{a,b\}$.
 Then,
 add a vertex~$x$ and the edges~$\{x,a\}$ and~$\{x,b\}$.
 Set $\kappa(x):=\min\{k+1,\kappa(K)\}$
 and $\eta(x):=\min\{\ell+1,\eta(K)\}$.
 For each~$w\in N_G(K)\cap R$, add the edge~$\{x,w\}$.
 Delete all vertices in~$K$.
\end{rrule}

\begin{proof}[Correctness proof]
 Let~$I=(G,s,t,k,\ell,\lambda,\kappa,\eta)$ be an instance of~\WsspAcr{} and let~$I'=(G',s,t,k,\ell,\lambda,\kappa',\eta')$ be the instance of \WsspAcr{} obtained from applying \cref{rr:shrinkedgypaths}.
 Let~$T\in\calT$ and let~$Q\subseteq T$ be the maximal\hyp edgy $a$-$b$-path in~$T$ being changed to the maximal\hyp edgy $a$-$b$-path~$Q'$ by the application of \cref{rr:shrinkedgypaths}.
 We claim that~$I$ is a \yes-instance if and only if~$I'$ is a \yes-instance.

 (\(\Rightarrow\))
 Let~$I$ be a \yes-instance and $P$~be a solution path in~$G$.
 Note that, by construction of~$G'$,
 for each~$X\subseteq V(G)\setminus (R\cup V(Q))$, we have~$N_G(X)=N_{G'}(X)$.
 Thus,
 if $V(Q)\cap V(P)=\emptyset$,
 then $P$~is also a solution path in~$G'$.
 Hence, assume that~$V(Q)\cap V(P)\neq\emptyset$.
 Since~$Q$ contains no good vertex and no vertex of degree at least three in~$T$, it follows that~$V(Q)\subseteq V(P)$.
 Moreover, we have~$\kappa'(x)=\kappa(K)\leq k$
 and~$\eta'(x)=\eta(K)\leq\ell$.
 For the path~\(P'\) in~$G'$ obtained from~$P$ by replacing~$V(Q)$ by~$V(Q')$,
 we have
\begin{align}
  \begin{aligned}
    \kappa'(V(P'))	&=\kappa'(V(P')\setminus \{x\})+\kappa'(x)=\kappa(V(P)\setminus K)+\kappa(K)
    =\kappa(V(P)),\\
    N_{G'}(V(P'))	&= \left(N_{G'}(V(P')\setminus\{x\})\setminus\{x\}\right) \cup \left(N_{G'}(x)\setminus\{a,b\}\right)
    \\&= \left(N_{G}(V(P)\setminus K)\setminus K\right) \cup \left(N_{G}(K)\setminus\{a,b\}\right)
    = N_G(V(P)),\text{ and}\\
    \eta'(V(P'))	&= \eta'(x)+\eta'(V(P')\setminus\{x\})= \eta(K) + \eta(V(P)\setminus K) 
    =\eta(V(P)).
    \label{eta}
  \end{aligned}
 \end{align}

 (\(\Leftarrow\))
 Let~$I'$ be a \yes-instance and $P'$~be a solution path in~$G'$.
 If $V(Q')\cap V(P')=\emptyset$,
 then $P'$~is also a solution path in~$G$.
 Hence, assume that~$V(Q')\cap V(P')\neq\emptyset$.
 Since~$Q'$ contains no good vertex and no vertex of degree at least three in~$T$, it follows that~$V(Q')\subseteq V(P')$.
 Moreover, we have~$\kappa'(x)=\kappa(K)\leq k$
 and~$\eta'(x)=\eta(K)\leq\ell$.
 Let~$P$ be the path in~$G$ obtained from~$P'$ by replacing~$V(Q')$ by~$V(Q)$.
 We have \(\kappa(V(P))=\kappa'(V(P'))\),
 \(N_G(V(P))=N_{G'}(V(P'))\),
 and \(\eta(V(P))= \eta'(V(P'))\)
 by \eqref{eta}.
\end{proof}

\begin{lemma}
 \label{obs:lintimelastrr}
 If \cref{rr:killingtree,rr:burntrees} are not applicable,
 then \cref{rr:shrinkedgypaths} is exhaustively applicable in linear time.
 Moreover, no application of \cref{rr:shrinkedgypaths} makes \cref{rr:killingtree} or \cref{rr:burntrees} applicable again.
\end{lemma}

\begin{proof}
  If \cref{rr:killingtree,rr:burntrees} are not applicable,
  then every maximal path of degree\hyp two vertices
  in~$G-F$ not containing good vertices
  is a maximal\hyp edgy path.
  Hence, employ the following.
  Let $Z$~be the set of all degree\hyp two vertices
  in~$G-F$
  and $Z'$~be a working copy of~$Z$.
 As long as~$Z'\neq \emptyset$, do the following.
 Select any vertex~$v\in Z'$
 and start a breadth-first search
 that stops when a good vertex or a vertex of degree at least three is found.
 Apply \cref{rr:shrinkedgypaths} on the just identified maximal\hyp edgy path (if it contains more than three vertices).
 Delete all the vertices found in the iteration from~$Z'$.
 
 Since no application of \cref{rr:shrinkedgypaths} deletes a good vertex or creates a vertex of degree one, no application of \cref{rr:shrinkedgypaths} makes \cref{rr:killingtree} or \cref{rr:burntrees} applicable.
\end{proof}

\noindent
We prove next that if none of \cref{rr:killingtree,rr:burntrees,rr:shrinkedgypaths} is applicable, the trees are small in the sense that the number of vertices in the tree is linear in the number of good vertices.

\begin{lemma}
  \label{obs:verysmalltrees}
 Let~$T\in\calT$ such that none of \cref{rr:killingtree,rr:burntrees,rr:shrinkedgypaths} is applicable.
 Let~$Y_T:=Y\cap V(T)$ denote the set of good vertices in~$T$.
 Then~$T$ has~$O(|Y_T|)$ vertices, each of weight~$O(k+\ell)$.
\end{lemma}

\begin{proof}
  We first show that,
  if none of \cref{rr:killingtree,rr:burntrees} is applicable,
  then $V(T)= Y_T\uplus V_3\uplus \biguplus_{Q\in \calQ} V(Q)$,
  where
  $V_3$ denotes the set of all vertices~$w$ not in~$Y_T$ with~$\deg_T(w)\geq 3$
 and $\calQ$~denotes the set of all maximal\hyp edgy paths in~$T$.
 Note that the sets~$Y_T$, $V_3$, and~$\biguplus_{Q\in \calQ} V(Q)$ are pairwise disjoint
  (by \cref{def:edgy},
  no edgy path contains a good vertex
  or a vertex of degree at least three).
  Suppose $V(T)=Y_T\uplus V_3\uplus \biguplus_{Q\in \calQ} V(Q)\uplus X$, we show that~$X= \emptyset$.
  Due to \cref{rr:killingtree,rr:burntrees}, the only vertices in~$T$ of degree
  one are good vertices.
  It follows that \(X\)~contains only degree\hyp two vertices,
  none of which are good.
  Since every vertex in~$V(T)\setminus Y_T$
  of degree two is contained in a maximal\hyp edgy path,
  it follows that $X$~is empty.
  It follows that~$V(T)=Y_T\uplus V_3\uplus \biguplus_{Q\in \calQ} V(Q)$.
  To finish the proof we upper-bound the number of vertices
  in~$V_3$ and in all paths in~$\calQ$
  linearly in~$|Y_T|$.

 Again, due to \cref{rr:killingtree,rr:burntrees}, every degree\hyp one vertex is in~$Y_T$.
 Hence there are at most $|Y_T|$~degree\hyp one vertices in~$T$, and thus~$|V_3|\leq |Y_T|$.
 Moreover,~$|\calQ|\leq 2|Y_T|$.
 Due to \cref{rr:shrinkedgypaths}, for every~$Q\in\calQ$ we have~$|V(Q)|\leq 3$.
 It follows that~$|V(T)|\leq 2|Y_T|+6|Y_T|=8|Y_T|$.
 Due to \cref{rr:burntrees,rr:shrinkedgypaths},
 each vertex~$v$ in~$T$ has~$\kappa(v)\leq k+1$ and~$\eta(v)\leq \ell+1$.
\end{proof}

\noindent
We are now ready to prove \cref{prop:sspToHssp}.
In a nutshell,
we approximate a minimum feedback vertex set
in linear time~\cite{BGNR98},
then apply \cref{rr:killingtree,rr:burntrees,rr:shrinkedgypaths} exhaustively
in linear time
(\cref{obs:firstrrulesfast,obs:lintimelastrr}),
and finish the proof using \cref{obs:verysmalltrees}.

\begin{proof}[Proof of \cref{prop:sspToHssp}]
  Compute a feedback vertex set~$F$
  of size~$\beta\leq 4\fvs$ in linear time~\cite{BGNR98}.
 Apply all reduction rules exhaustively in linear time: 
 first apply \cref{rr:killingtree,rr:burntrees}
 exhaustively in linear time (\cref{obs:firstrrulesfast}), 
 then \cref{rr:shrinkedgypaths}
 exhaustively in linear time (\cref{obs:lintimelastrr}).

 Now,
 consider a graph~\(G\) to which no data reduction
 rules are applicable and
 let $T_1,\ldots,T_h$~denote the trees in~$G-F$.
 By \cref{obs:verysmalltrees}, each~$T_i$ has $O(|Y_{T_i}|)$~vertices, each of maximal weight~$O(k+\ell)$, where~$Y_{T_i}=Y\cap V(T_i)$.
 Thus, the number of vertices and edges in~$G-F$ is
 \begin{align*}
    \sum_{i=1}^h O(|Y_{T_i}|) &=   \sum_{i=1}^h O(|Y_{T_i}|) = O(|Y|)\subseteq O(\beta\cdot(k+\ell)),
 \end{align*}
 where the last inclusion follows from \cref{obs:fewgoodvertices}.
 It follows that there are~$O(\beta^2\cdot(k+\ell))$ edges in~$G$.
 Altogether, $G$~has $O(\beta\cdot (k+\ell))$~vertices, each of weight~$O(k+\ell)$, and $O(\beta^2\cdot(k+\ell))$~edges.
 Moreover, the obtained instance is \simple{} (with~$A=R$, see \cref{def:simplevwssp}).
\end{proof}
  
\noindent
Combining \cref{prop:sspToHssp} with \cref{prop:Hssp2ssp},
we now prove the main result of this section. 

\begin{proof}[Proof of \cref{thm:ssppk-fvskell}]
 Let~$I=(G,s,t,k,\ell)$ be an instance of~\sspAcr{}.
 First, employ \cref{prop:sspToHssp} to obtain \simple{} instance~$I'=(G',s,t,k,\ell,\lambda,\kappa,\eta)$ of~\WsspAcr{}.
 Then employ \cref{prop:Hssp2ssp} to obtain instance~$I''=(G'',s',t',k',\ell)$ of~\sspAcr{}.
 We know that~$G'$ has $O(\fvs\cdot(k+\ell))$~vertices, $O(\fvs^2\cdot (k+\ell))$~edges and vertex weights at most~$O(k+\ell)$.
 Hence, 
 \begin{align*}
    |V(G'')|	= \kappa(V(G')) + \eta(V(G')) 
		&\in O(\fvs\cdot(k+\ell)\cdot k+ \fvs\cdot(k+\ell)\cdot \ell)\subseteq O(\fvs\cdot(k+\ell)^2).\qedhere
 \end{align*}
\end{proof}

\subsubsection{A problem kernel with $O(\fes)$~vertices and edges}
\label{sssec:fes}

In this section,
we show two data reduction algorithms that reduce \sspAcr{}
to an equivalent instance of \WsspAcr{} with \(O(\fes)\)~vertices
and allow a trade\hyp off between the running time
and the size of the resulting instance.
The first runs in linear time,
yet creates vertices with weights in~$O(k+\ell)$,
thus not bounding the overall size
of the reduced instance by a polynomial in~\(\fes\).
The second takes polynomial time
and creates vertex weights encodable using $O(\fes^3)$~bits.
Thus, when finally reducing back to~\sspAcr{} using~\cref{prop:Hssp2ssp},
we obtain a problem kernel of size~$O(\fes\cdot(k+\ell))$
using the first algorithm
and a problem kernel of size 
polynomial in~$\fes$
using the second algorithm.

\begin{theorem}
 \label{thm:bikernelfes}
 \sspTsc{} admits a problem kernel
 \begin{enumerate}[(i)]
  \item\label{thm:biki} with size $O(\fes\cdot(k+\ell))$ that is computable in linear\hyp time, and
  \item\label{thm:bikii} with size polynomial in~$\fes$ that is computable in polynomial time.
 \end{enumerate}
 Herein,
 $\fes$ denotes the feedback edge number of the input graph.
\end{theorem}

\noindent
We first prove \eqref{thm:biki}.
Concretely, we prove the following:

\begin{proposition}
  \label{prop:sspToHsspFES}
  For any instance of \sspAcr{}
  we can compute in linear time an equivalent \simple{} instance
  of \WsspAcr{} with
  $16\fes+9$~vertices, $17\fes+8$~edges
  and vertex weights in~$O(k+\ell)$.
\end{proposition}

\noindent
Let $\Ef{}$~be a feedback edge set
of size~\(\fes\) in~$G=(V,E)$.
By \cref{rrule:onecomponent},
we may assume $G$~to be connected.
Thus, $T:=G-\Ef{}$~is a tree.
Let $Y:=\{v\in V\mid v\in e\in \Ef{}\}\cup\{s,t\}$~denote
the set of vertices containing~$s$ and~$t$ and all endpoints of the edges in~$\Ef{}$.
We call the vertices in~$Y$ \emph{good}.
\looseness=-1
In the following,
we will interpret the input \sspAcr{}
instance as an instance of \WsspAcr{}
with unit weight functions~\(\kappa\) and~\(\lambda\)
and the zero weight function~$\eta$.
Our two reduction rules we state next are simplified version of \cref{rr:burntrees,rr:shrinkedgypaths} with~$\calT=\{T\}$ and~$R=\emptyset$.
\begin{rrule}
 \label{rr:burntreesfes}
 If there is~$v\in V(T)\setminus Y$
 with $N_T(v)=\{w\}$,
 set $\eta(w):=\min\{\ell+1,\eta(w)+1\}$ and delete~$v$. 
\end{rrule}

\begin{rrule}
 \label{rr:shrinkedgypathsfes}
 Let $Q\subseteq T$~be a maximal\hyp edgy
 $a$-$b$-path with~$|V(Q)|>3$ in~$T$
 and
 let $K:=V(Q)\setminus\{a,b\}$.
 Then,
 add a vertex~$x$ and the edges~$\{x,a\}$ and~$\{x,b\}$.
 Set~$\kappa(x):=\min\{k+1,\kappa(K)\}$
 and~$\eta(x):=\min\{\ell+1,\eta(K)\}$.
 Delete all vertices in~$K$.
\end{rrule}
\noindent
The correctness of \cref{rr:burntreesfes,rr:shrinkedgypathsfes} follow immediately from the correctness of \cref{rr:burntrees,rr:shrinkedgypaths}.
Moreover, due to~\cref{obs:firstrrulesfast,obs:lintimelastrr}, 
we can first apply \cref{rr:burntreesfes} exhaustively in linear time, 
and then apply \cref{rr:shrinkedgypathsfes} exhaustively in linear time without making \cref{rr:burntreesfes} applicable again.
After applying
\cref{rr:burntreesfes,rr:shrinkedgypathsfes} exhaustively,
we can show:

\begin{observation}
  \label{obs:verysmalltreesfes}
 Let~$T$ be such that none of \cref{rr:burntreesfes,rr:shrinkedgypathsfes} is applicable.
 Then~$G$ has at most~$8|Y|-7$ vertices and~$8|Y|-8+|\Ef{}|$~edges, where each vertex is of weight~$O(k+\ell)$.
\end{observation}

\begin{proof}
  Due to \cref{rr:burntreesfes},
  every leaf of~$T$ is in~$Y$.
  Hence, there are at most $2|Y|-1$~vertices
  in~$T$ of degree not equal to two.
  Since~$T$ is a tree,
  there are at most~$2|Y|-2$ paths connecting two vertices being good or of degree at least three.
 Due to~\cref{rr:shrinkedgypathsfes}, these paths contain at most three vertices.
 It follows that there are at most~$8|Y|-7$ vertices in~$T$,
 each of weight~$O(k+\ell)$,
 and,
 consequently,
 at most~$8|Y|-8$ edges in~$T$.
 As~$T$ only differs from~$G$ by~$\Ef{}$,
 it follows that $G$~has at most~$8|Y|-8+|\Ef{}|$ edges.
\end{proof}

\noindent
We are ready to prove~\cref{prop:sspToHsspFES}.

\begin{proof}[Proof of~\cref{prop:sspToHsspFES}]
 Let $I=(G,s,t,k,\ell)$~be an instance of~\sspAcr{}.
 Compute a minimum feedback vertex set~$\Ef{}$
 of size $\fes:=|\Ef{}|$ in~$G$ in linear time
 (just take the complement of a spanning tree).
 Compute the set~$Y$ of good vertices. 
 First apply \cref{rr:burntreesfes} exhaustively in linear time.
 Next, apply \cref{rr:shrinkedgypathsfes} exhaustively
 in linear time.
 Let $I'=(G',s,t,k,\ell,\lambda,\kappa,\eta)$~denote
 the obtained instance of \WsspAcr.
 Observe that due to~\cref{rr:shrinkedgypathsfes},~$I'$ is \simple{} (with~$A=\emptyset$, see~\cref{def:simplevwssp}).
 Due to \cref{obs:verysmalltreesfes},
 we know that $G'$~has at most $8|Y|-7$~vertices
 and $8|Y|-8+\fes$~edges,
 where each vertex is of weight~$O(k+\ell)$.
 Note that~$|Y|\leq 2\fes+2$.
 Hence,
 $G'$~has at most $16\fes+9$~vertices,
 $17\fes+8$~edges,
 and vertex weights in~$O(k+\ell)$.
\end{proof}

\noindent
Having shown
\cref{prop:sspToHsspFES},
we can now prove
\cref{thm:bikernelfes}.
Herein,
to strip our shrunk \WsspAcr{} instances of weights,
we will employ \cref{prop:Hssp2ssp} for~\cref{thm:bikernelfes}\eqref{thm:biki}
and \cref{lem:weightreduce} for~\cref{thm:bikernelfes}\eqref{thm:bikii}.

\begin{proof}[Proof of~\cref{thm:bikernelfes}]
 Let~$I=(G,s,t,k,\ell)$ be an instance of~\sspAcr{}.
 Employ~\cref{prop:sspToHsspFES} to obtain simple instance $I'=(G',s,t,k,\ell,\lambda,\kappa,\eta)$ of~\WsspAcr{},
 where~$G'$ has at most~$O(\fes)$ vertices and edges, where each vertex is of weight~$O(k+\ell)$.
 Employing~\cref{prop:Hssp2ssp} yields an instance~$\I''=(G'',s',t',k'',\ell'')$ of \sspAcr{} in time
 \[ 
  \kappa(V(G')) + \eta(V(G')) + |E(G')| \in O(\fes\cdot (k+\ell)).
  \]
 Due to~\cref{prop:Hssp2ssp}, it follows that~$G''$ has at most~$M$ vertices, yielding \eqref{thm:biki}.
 
 For statement~(ii), apply \cref{lem:weightreduce} to obtain an instance~$\I''=(G',s,t,k',\ell',\lambda',\kappa',\eta')$ of \WsspAcr{} with $k'$, $\ell'$, and all weights encoded with~$O(\fes^3)$ bits.
 Since \WsspAcr{}~is \NP-complete,
 there is a polynomial\hyp time many\hyp one
 reduction to \sspAcr{}.
 Employing such a polynomial\hyp time many\hyp one
 reduction on instance~$\I''$ yields statement \eqref{thm:bikii},
 since it can blow up the instance size
 by at most a polynomial.
\end{proof}

\subsection{Limits of data reduction} 
\label{sec:fvsnokern}

In \cref{sssec:fvskell}, we proved a problem kernel for
\sspAcr{} with size polynomial in~$\fvs+k+\ell$.
By \cref{rem:vckr},
we know that,
unless \unlessPK{},
we cannot drop~$\ell$ here,
as a problem kernel with size polynomial in~$\fvs+k$
would also be polynomial in~$\vc+k$.
In this section,
we prove that,
unless \unlessPK{},
we cannot drop~$k$ either:

\begin{theorem}
 \label{thm:sspNoPKfvsell}
 Unless $\unlessPK$, \sspTsc{} admits no problem kernel with size polynomial in~$\fvs+\ell$. 
\end{theorem}
\noindent
To prove \cref{thm:sspNoPKfvsell},
we use a cross composition (\cref{def:crossco})
of \mcclique{} (\cref{prob:mcclique})
into \sspAcr{}.
In fact,
we will reduce from the NP-hard~\cite{FHRV09,CFK+15} special case
of \mcclique{}
where
instances~\(G=(V_1,V_2,\dots,V_k,E)\)
with $E_{\{i,j\}}:=\{\{u,v\}\in E\mid u\in V_i, v\in V_j\}$
satisfy
$|V_i|= |V_j|$ for~$1\leq i<j\leq k$,
$|E_{\{i,j\}}|= |E_{\{i',j'\}}|$ for all~$1\leq i<j\leq k$
and $1\leq i'<j'\leq k$, and
have at least $k+1$~vertices.
For the cross composition,
we use the following polynomial equivalence relation
on instances of \mcclique{}.

\begin{lemma}
  \label{obs:nopkfvsellRper}
  Let two \mcclique{}
  instances~$G=(V_1,V_2,\dots, V_k,E)$
  and~$G'=(V_1',V_2',\dots, V'_{k'},E')$
with \(V=V_1\uplus\dots\uplus V_k\)
and \(V'=V_1'\uplus\dots\uplus V'_{k'}\)
  be \emph{$\calR$\hyp equivalent}
  if and only if~$|V|=|V'|$, $|E|=|E'|$,
  and \(k=k'\).
  Then, $\calR$ is a polynomial equivalence relation.
\end{lemma}

\begin{proof}
  Deciding whether $G$ and~$G'$
  are $\calR$\hyp equivalent
  is doable in~$O(|V|+|V'|+|E|+|E'|)$ time.
  Now, let
  $S\subseteq \Sigma^*$~be a set of instances
  and $n:=\max_{x\in S} |x|$.
  There are at most~$n^{O(1)}$
  different vertex set sizes,
  edge set sizes,
  and partition sizes of the vertex sets,
  resulting in at most
  $n^{O(1)}$~equivalence classes.
\end{proof}

\noindent
We next describe the cross composition.

\newcommand{\zvertex}{h}
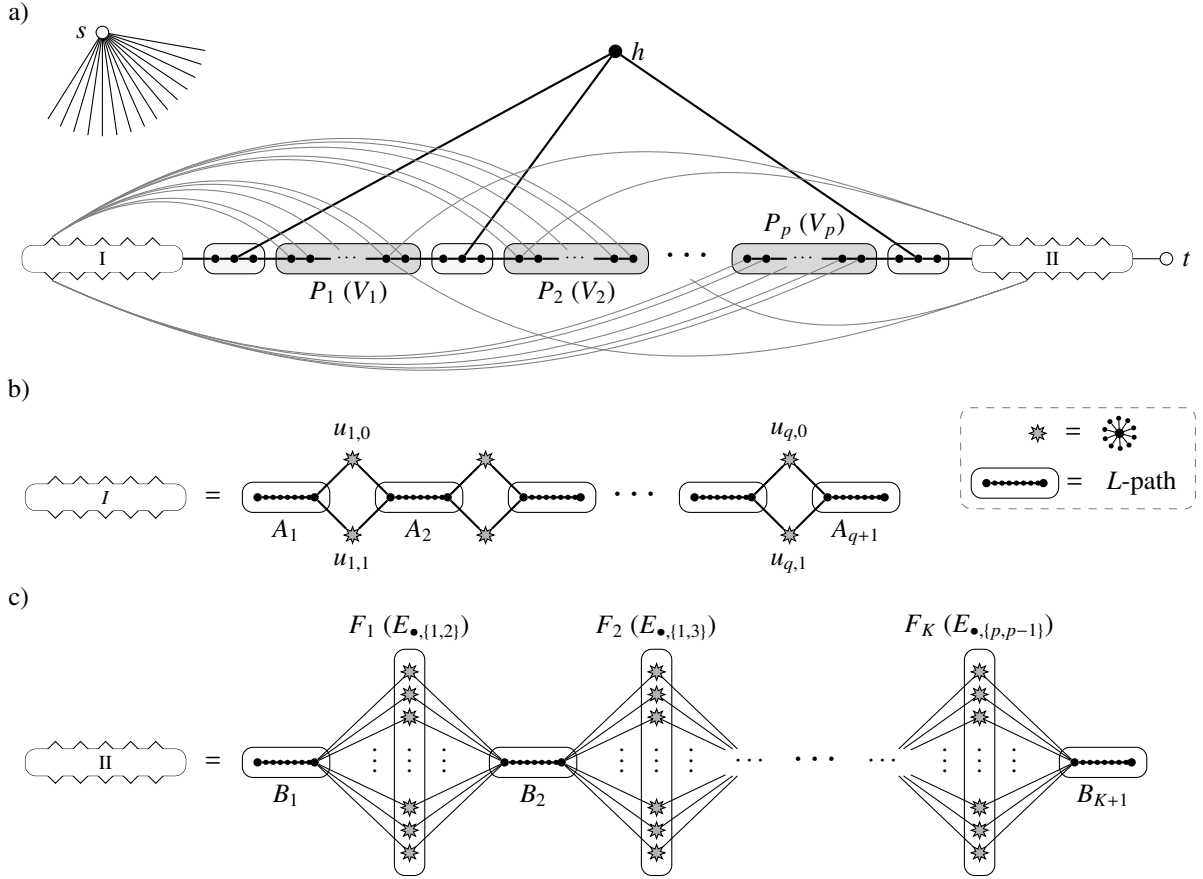
\begin{figure*}[t!]
 \centering
  \begin{tikzpicture}

\usetikzlibrary{decorations.pathreplacing,calc,shapes,positioning}
	\tikzstyle{pnode}=[fill,circle,scale=1/3]
	\tikzstyle{stnode}=[fill=white,circle,draw,scale=1/2]
	\tikzstyle{starnode}=[draw=black,fill=black!30!white,star,star points=8,star point ratio=2,scale=1/3]%
	\tikzstyle{lnode}=[fill,circle,scale=1/5]
	\tikzstyle{ppnode}=[fill=white,circle,draw,scale=1/3]
      \def\xr{1}%
      \def\yr{1}
      \def\xsh{1*\xr}

	\newcommand{\bstar}[3]{
	
		\def\noleaves{10};
		\def\distleaves{0.2};
	
		\node (#1) at (#2,#3)[ppnode,fill=black]{};
 		 \foreach \j in {1,...,\noleaves}{
			   \node (a\j) at ($(#1)+ (165+\j*360/\noleaves:\distleaves cm)$)[lnode]{};
				\draw (#1) -- (a\j);
 		 }
	}
   
	\newcommand{\edgesets}[7]{
		
		\draw[rounded corners] (#3-0.2,#4-1.5*\yr) rectangle (#3+0.2*\xr,#4+1.5*\yr)node[above,xshift=-6pt]{#7};
		\foreach \x in {1,2,3,7,8,9}{
    			\node (x#1#2x\x) at  (#3,\x*0.3*\yr-1.5*\yr-1*#4)[starnode]{};
    		}
		\node at (x#1#2x7.south)[yshift=-10*\yr]{$\vdots$};
		 \foreach \x in {1,2,3,7,8,9}{
 		   \draw[] (x#1#2x\x) to (#5);
    		   \draw[] (x#1#2x\x) to (#6);
    		}
		\node at (x#1#2x7.south)[yshift=-10*\yr,xshift=13*\xr]{$\vdots$};	
		\node at (x#1#2x7.south)[yshift=-10*\yr,xshift=-13*\xr]{$\vdots$};
}
    
	\newcommand{\longpath}[4]{
		\draw[rounded corners] (#1-0.2*\xr,#2-0.2*\yr) rectangle  node[below,yshift=-5*\yr]{#4}(#1+0.95*\xr,#2+0.2*\yr);
		\node (l#3-1) at (#1,#2)[pnode]{};
		\node (l#3-4) at (#1+0.75*\xr,#2)[pnode]{};
		\draw[thick] (l#3-1) to (l#3-4);
		\draw[line width=2pt, line cap=round, dash pattern=on 0pt off 1.3\pgflinewidth] ($(l#3-1.east)+(0.05*\xr,0)$) to (l#3-4);
	}

	\newcommand{\vpath}[6]{
		\draw[rounded corners,fill=black!15!white] (#1-0.2*\xr,#2-0.2*\yr) rectangle (#1+1.7*\xr,#2+0.2*\yr) node[midway,#5,yshift=#6]{#4};%
		\node (v#3-1) at (#1,#2)[pnode]{};
		\node (v#3-2) at (#1+0.25*\xr,#2)[pnode]{};
		\node (v#3-x) at (#1+0.75*\xr,#2)[scale=2/3]{$\cdots$};
		\draw[thick] (v#3-1) to (v#3-2) to (v#3-x);
		\node (v#3-3) at (#1+1.25*\xr,#2)[pnode]{};
		\node (v#3-4) at (#1+1.5*\xr,#2)[pnode]{};
		\draw[thick] (v#3-x) to (v#3-3) to (v#3-4);
	}

	\newcommand{\selecgadget}[7]{
		\node (s#3-o) at (#1,#2+0.5*\yr)[starnode,label=90:{#6}]{};
		\node (s#3-u) at (#1,#2-0.5*\yr)[starnode,label=-90:{#7}]{};
		\draw[thick]  (#4) to (s#3-o) to (#5);
		\draw[thick]  (#4) to (s#3-u) to (#5);
	}

	 \node (s) at (2*\xr,3*\yr)[stnode,label=180:{$s$}]{};
		\foreach \j in {1,...,15}{
			   \node (atemp\j) at ($(s)+ (230+\j*120/15:1.5*\yr)$)[]{};
				\draw (s) -- (atemp\j);
 		 }
	
	\node (skip) at  (8.75*\xr,2.75*\yr)[fill=black,draw, circle, scale=1/2,label=0:{$\zvertex$}]{};
	
	\newcommand{\skipgadget}[5]{
		\draw[rounded corners,fill=black!3!white] (#1-0.4*\xr,#2-0.2*\yr) rectangle (#1+0.4*\xr,#2+0.2*\yr);%
		\node (i#3-l) at (#1-0.25*\xr,#2)[pnode]{};
		\node (i#3-c) at (#1,#2)[pnode]{};
		\node (i#3-r) at (#1+0.25*\xr,#2)[pnode]{};
		\draw[thick]  (#4) to (i#3-l) to (i#3-c) to (i#3-r) to (#5);
		\draw[thick]  (i#3-c) to (skip);
	}
	\node (uno) at (2*\xr,0)[rounded corners, draw, minimum width=80*\xr,scale=0.75]{$I$};
	\node (uno1) at (2*\xr-0.66*\xr,0)[fill=white,diamond,draw,scale=1.2]{};
	\node at (2*\xr-0.33*\xr,0)[fill=white,diamond,draw,scale=1.2]{};
	\node at (2*\xr-0.0*\xr,0)[fill=white,diamond,draw,scale=1.2]{};
	\node at (2*\xr+0.33*\xr,0)[fill=white,diamond,draw,scale=1.2]{};
	\node at (2*\xr+0.66*\xr,0)[fill=white,diamond,draw,scale=1.2]{};
	\node (uno) at (2*\xr,0)[rounded corners, fill=white, minimum width=80*\xr,scale=0.75]{I};
	
	\node (due) at (14.5*\xr,0)[rounded corners, draw, minimum width=80*\xr,scale=0.75]{II};
	\node (due1) at (14.5*\xr-0.66*\xr,0)[fill=white,diamond,draw,scale=1.2]{};
	\node (due2) at (14.5*\xr-0.33*\xr,0)[fill=white,diamond,draw,scale=1.2]{};
	\node at (14.5*\xr-0.0*\xr,0)[fill=white,diamond,draw,scale=1.2]{};
	\node at (14.5*\xr+0.33*\xr,0)[fill=white,diamond,draw,scale=1.2]{};
	\node at (14.5*\xr+0.66*\xr,0)[fill=white,diamond,draw,scale=1.2]{};
	\node (due) at (14.5*\xr,0)[rounded corners,  fill=white, minimum width=80*\xr,scale=0.75]{II};
	
	\node (t) at  (16*\xr,0*\yr)[stnode,label=0:{$t$}]{};
	\draw (due) to (t);
	
	\begin{scope}[xshift=-100*\xr,yshift=0]
	\vpath{11*\xr-3*\xsh}{0}{1}{$P_1$~$(V_1)$}{below}{-5*\yr};
	\skipgadget{10.25*\xr-3*\xsh}{0}{1}{uno}{v1-1};
	\vpath{14*\xr-3*\xsh}{0}{2}{$P_2$~$(V_2)$}{below}{-5*\yr};
	\skipgadget{13.25*\xr-3*\xsh}{0}{1}{v1-4}{v2-1};
	\node (dots2) at (16.25*\xr-3*\xsh,0)[scale=1.5]{$\cdots$};
	\vpath{17*\xr-3*\xsh}{0}{3}{$P_p$~$(V_p)$}{above}{4*\yr};		
	\skipgadget{19.25*\xr-3*\xsh}{0}{1}{v3-4}{due};
	\end{scope}

	\foreach \x in {1,2,3,4,x}{
		\draw[color=gray] (uno1.north) to [out=35,in=135](v1-\x);
		\draw[color=gray] (uno1.north)  to [out=35,in=135](v2-\x);
		\draw[color=gray] (uno1.south)  to [out=-25,in=-155](v3-\x);
	}

	\draw[color=gray] (v2-1) to [out=45,in=160](due1.north) to [out=160,in=45](v1-3);
	\draw[color=gray] (v1-3) to [out=-35,in=-160](due2.south) to [out=-160,in=-35](dots2.south);

	\begin{scope}[xshift=115*\xr,yshift=-10*\yr]

	\begin{scope}[xshift=-410*\xr,yshift=-160*\yr]
	\draw[dashed,gray,rounded corners] (23.66*\xr,2.66*\yr) rectangle (26.75*\xr,4*\yr);
	\node (xxx) at (24.675*\xr,3.66*\yr)[starnode]{};
	\node at (xxx.east)[xshift=10*\xr]{$=$};
	\bstar{xxx}{25.75*\xr}{3.66*\yr};
	\longpath{24*\xr}{3*\yr}{xxx}{};
	\node at (lxxx-4.east)[xshift=10*\xr]{$=$};
	\node at (lxxx-4.east)[xshift=35*\xr]{$L$-path};
	\end{scope}

  \begin{scope}[xshift=0,yshift=-80*\yr]
	\node (uno) at (-2*\xr,0)[rounded corners, draw, minimum width=80*\xr,scale=0.75]{$I$};
	\node at (-2*\xr-0.66*\xr,0)[fill=white,diamond,draw,scale=1.2]{};
	\node at (-2*\xr-0.33*\xr,0)[fill=white,diamond,draw,scale=1.2]{};
	\node at (-2*\xr-0.0*\xr,0)[fill=white,diamond,draw,scale=1.2]{};
	\node at (-2*\xr+0.33*\xr,0)[fill=white,diamond,draw,scale=1.2]{};
	\node at (-2*\xr+0.66*\xr,0)[fill=white,diamond,draw,scale=1.2]{};
	\node (uno) at (-2*\xr,0)[rounded corners, fill=white, minimum width=80*\xr,scale=0.75]{$I$};
	
	\node at (uno.east)[anchor=center,xshift=10*\xr]{$=$}; 
	\longpath{0*\xr}{0}{1}{$A_1$};
	\longpath{1.75*\xr}{0}{2}{$A_2$};
	\selecgadget{1.25*\xr}{0*\yr}{1}{l1-4}{l2-1}{$u_{1,0}$}{$u_{1,1}$};
	\longpath{3.5*\xr}{0}{3}{};
	\selecgadget{3*\xr}{0*\yr}{2}{l2-4}{l3-1}{}{};
	\node (dots1) at (5*\xr,0)[scale=1.5]{$\cdots$};	
	\longpath{5.75*\xr}{0}{4}{};
	\longpath{7.5*\xr}{0}{5}{$A_{q+1}$};
	\selecgadget{7*\xr}{0*\yr}{2}{l4-4}{l5-1}{$u_{q,0}$}{$u_{q,1}$};
  \end{scope}

	\begin{scope}[xshift=0,yshift=-180*\yr]
	\node (due) at (-2*\xr,0)[rounded corners, draw, minimum width=80*\xr,scale=0.75]{II};
	\node at (-2*\xr-0.66*\xr,0)[fill=white,diamond,draw,scale=1.2]{};
	\node at (-2*\xr-0.33*\xr,0)[fill=white,diamond,draw,scale=1.2]{};
	\node at (-2*\xr-0.0*\xr,0)[fill=white,diamond,draw,scale=1.2]{};
	\node at (-2*\xr+0.33*\xr,0)[fill=white,diamond,draw,scale=1.2]{};
	\node at (-2*\xr+0.66*\xr,0)[fill=white,diamond,draw,scale=1.2]{};
	\node (due) at (-2*\xr,0)[rounded corners, fill=white, minimum width=80*\xr,scale=0.75]{II};
	\node at (due.east)[anchor=center,xshift=10*\xr]{$=$}; 
	\longpath{0*\xr}{0}{5}{$B_1$};
	\longpath{3.25*\xr}{0}{6}{$B_2$};
	\edgesets{1}{2}{2*\xr}{0}{l5-4}{l6-1}{$F_1$ ($E_{\bullet,\{1,2\}}$)};
	\node (dotsL1) at (6.5*\xr,0)[scale=1]{$\cdots$};
	\node (dotsL) at (7.375*\xr,0)[scale=1.5]{$\cdots$};
	\node (dotsL2) at (8.25*\xr,0)[scale=1]{$\cdots$};
	\edgesets{1}{3}{5.25*\xr}{0}{l6-4}{dotsL1}{$F_2$ ($E_{\bullet,\{1,3\}}$)};
	\longpath{10.75*\xr}{0}{7}{$B_{K+1}$};
	\edgesets{p-1}{p}{9.5*\xr}{0}{dotsL2}{l7-1}{$F_K$ ($E_{\bullet,\{p,p-1\}}$)};
	\end{scope}
	\end{scope}
    \node at (0.9*\xr,3.25*\yr)[]{a)};
    \node at (0.9*\xr,-1.75*\yr)[]{b)};
    \node at (0.9*\xr,-4.5*\yr)[]{c)};
\end{tikzpicture}
 \caption{Illustration to \cref{constr:fvsellnokern}.
 a) A high-level sketch of the construction with some illustrative edges between the gadgets.
 b) Details for the gadget labeled~I in a).
 c) Details for the gadget labeled~II in a).
 For~$P_a$ and~$F_y$, in parentheses
 we indicate to which sets of the input graphs they correspond to (where~$\bullet$ is a placeholder for every element in~$\{1,\ldots,p\}$).
 }
 \label{fig:fvsellnoPK}
\end{figure*}

\begin{construction}
  \label{constr:fvsellnokern}
  Let~$G_1=(V_{1,1},\dots, V_{1,k},E_1),\ldots,G_p=(V_{p,1},\dots, V_{p,k},E_p)$ be
  $p=2^q$ \mcclique{} instances
  equivalent under~$\calR$,
  where \(q\in\N\).
  Then,
  we can denote
  by $n$~the number of vertices and
  by $m$~the number of edges in each instance.
  Moreover,
  let
  $V_{a,i}=\{v_{a,i}^1,\dots,v_{a,i}^r\}$,
  and
  $E_a=\biguplus_{1\leq i<j\leq k}E_{a,\{i,j\}}$
  with
  $E_{a,\{i,j\}}=\{e_{a,\{i,j\}}^1,\dots,e_{a,\{i,j\}}^x\}$
  for all~\(a\in\{1,\dots,p\}\).
  Construct the following \sspAcr{}
  instance~$(G,s,t,k',\ell)$ with graph~$G$ (refer to \cref{fig:fvsellnoPK} for an illustration).
 Let~$G$ be initially empty,
 $M:=k+|E|-K+q+2$, 
 $L:=p\cdot n\cdot m+2\cdot(M+ K\cdot M^2 + q\cdot M^2)$, and~$K:=\binom{k}{2}$.
 \begin{enumerate}
  \item Add~$q+1$ paths~$A_1,\dots,A_{q+1}$, where~$V(A_y)=\{a_{y,1},\dots,a_{y,L}\}$ and $a_{y,1},a_{y,L}$~are the end points.
    For each~$y\in\{1,\ldots,q\}$, 
    add the vertex set~$U_y=\{u_{y,0},u_{y,1}\}$,
    and make each vertex of~$U_y$ adjacent to~$a_{y,L}$
    and~$a_{y+1,1}$.
    Define~$U=\bigcup_{y=1}^{q} U_y$. 
    See \cref{fig:fvsellnoPK}b).
  \item Add~$K+1$ paths~$B_1,\ldots,B_{K+1}$, where~$V(B_z)=\{b_{z,1},\dots,b_{z,L}\}$ and~$b_{z,1},b_{z,L}$ are the
    end points.
    For each~$z\in\{1,\ldots,K\}$, add the vertex set~$F_z:=\{e_z^1,\ldots,e_z^x\}$
    and make each vertex in~$F_z$ adjacent
    to~$b_{z,L},b_{z+1,1}$.
    Define~$F=\bigcup_{z=1}^K F_z$.
    Choose an arbitrary
    bijection~$\pi\colon\{1,\ldots,K\}\to \{\{i,j\}\mid 1\leq i<j\leq k\}$.
    We say that $e_y^z$~corresponds
    to the $z$-th edge~$e^z_{a,\pi(y)}\in E_{a,\pi(y)}$
    for all~$a\in\{1,\ldots,p\}$. 
    See \cref{fig:fvsellnoPK}c).
  \item Add~$p$ paths~$P_1,\dots,P_p$ such that~$P_a$
    has vertex set~$\{v_{a,i}^d\mid i\in\{1,\ldots,k\},\ d\in\{1,\ldots,r\}\}$ and
    edge set~$\{\{v_{a,i}^r,v_{a,i+1}^{1}\}\mid i\in\{1,\ldots,k-1\}\} \cup \bigcup_{1\leq i\leq k} \{\{v_{a,i}^d,v_{a,i}^{d+1}\}\mid d\in\{1,\ldots,r-1\}\}$.
    We say~$P_a$ corresponds to the vertices in~$V_a$ in the $a$-th graph~$G_a$.
    Next, for each~$a\in\{1,\ldots,p+1\}$,
    add a path of three vertices~$w_{a,1},w_{a,2},w_{a,3}$ with edges~$\{w_{a,1},w_{a,2}\},\{w_{a,2},w_{a,3}\}$.
    Make~$w_{1,1}$ adjacent to~$a_{q+1,L}$,
    and~$w_{p+1,3}$ adjacent to~$b_{1,1}$.
    For each~$1<a\leq p+1$, make~$w_{a,1}$ adjacent
    to~$v_{a-1,k}^r$.
    For each~$1\leq a<p+1$, make~$w_{a,3}$ adjacent to~$v_{a,1}^1$.
  \item Add one vertex~$\zvertex$ and for each~$a\in\{1,\ldots,p+1\}$ make~$w_{a,2}$ adjacent to~$\zvertex$.
  \item For each~$a\in\{1,\ldots,p\}$, make each~$v\in V(P_a)$ adjacent to the vertex in~$F$ corresponding to an incident edge.
    That is, if~$v_{a,i}^d$ is incident with edge~$e_{a,\{i,j\}}^{x'}$, make~$v_{a,i}^d$ adjacent to vertex~$e_{z}^{x'}$ where~$z=\pi^{-1}(\{i,j\})$.
  \item For each~$a\in\{1,\ldots,p\}$, make each~$v\in V(P_a)$ adjacent to the vertices in~$U$ as follows:
    Let~$a_1a_2\cdots a_q$ be the 0-1-string
    of length~$q$ encoding the number~$a-1$ in binary.
    Then, make each~$v\in V(P_a)$ adjacent to each vertex in the set~$\{u_{i,a_i}\mid i\in\{1,\ldots,q\}\}$.
    Note that for each~$u\in U$, we have that if~$N(u)\cap V(P_a)\neq \emptyset$ for some~$a\in\{1,\ldots,p\}$, then~$N(u)\supseteq V(P_a)$.
    Moreover, for each~$u\in U$ we have~$|\{a\in\{1,\ldots,p\}\mid N(u)\cap V(P_a)\neq \emptyset\}|=p/2$.
  \item Add~$s$ and~$t$.
    Make~$t$ adjacent to~$b_{K+1,L}$.
    Make~$s$ adjacent to all vertices except the vertices in~$\bigcup_{a=1}^p V(P_a)$.
  \item For each vertex~$v\in F\cup U$, add~$M^2$ vertices only
    adjacent to~$v$.
  \item Set~$k'=(q+K+2)L+q+(p-1)n+3(p+1)+K+1$ and~$\ell=M+ K\cdot M^2 + q\cdot M^2$.
 \end{enumerate}
\end{construction}
\noindent
Before we prove that the instance~$\I$ obtained from \cref{constr:fvsellnokern} is a \yes-instance if and only if at least one input instance is a \yes-instance,
we prove some crucial properties of solutions in~$\I$ in the case that~$\I$ is a \yes-instance.

\begin{lemma}
 \label{lem:fvsell-pathstr}
 Let~$(G,s,t,k',\ell)$ be the \sspAcr{}-instance obtained from \cref{constr:fvsellnokern} and let~$(G,s,t,k',\ell)$ be a \yes-instance.
 Let~$P$ be a solution $s$-$t$~path in~$G$.
 Then the following hold:
 \begin{compactenum}[(i)]
  \item $P$ contains each path~$Q\in\{A_2,\dots,A_{q+1},B_1,\dots,B_{K+1}\}$ and a subpath of~$A_1$ as subpath.
	Moreover, the first vertex on~$P$ after~$s$ is in~$V(A_1)\setminus\{a_{1,L}\}$.
     \label{lem:fvsell-pathstr-subpaths}
  \item $|V(P)\cap U_y|=|V(P)\cap F_z|=1$ for all~$y\in\{1,\ldots,q\}$, $z\in\{1,\ldots,K\}$.\label{lem:fvsell-pathstr-onepergadget}
  \item Let~$v$ be a vertex of some~$U_y$ ($F_z$) contained in~$P$,
    and let~$(\{v',v,v''\},\{\{v',v\},\{v,v''\}\})$ be a subpath of~$P$ where the distance from~$v'$ to~$s$ in~$P$ is smaller than the one from~$v$ or~$v''$.
    Then~$v'=a_{y,L}$ and~$v''=a_{y+1,1}$ ($v'=b_{z,L}$ and~$v''=b_{z+1,1}$).\label{lem:fvsell-pathstr-gadgetnghbrs}
 \end{compactenum}
\end{lemma}

\begin{proof}
  (\ref{lem:fvsell-pathstr-subpaths}): 
  From each path~$Q\in\{A_2,\dots,A_{q+1},B_1,\dots,B_{K+1}\}$, at least~$L-\ell>\ell$ vertices must be contained.
  Since the inner vertices of~$Q$ are only adjacent to vertices in~$Q$ and~$s$, it follows that~$Q$ is a subpath of~$P$.
  Moreover, also at least~$L-\ell>\ell$ vertices from~$A_1$ must be contained in~$P$.
  Hence, a subpath of~$A_1$ is a subpath of~$P$.
  From the latter, we observe that the first vertex on~$P$ after~$s$ is in~$V(A_1)\setminus\{a_{1,L}\}$.
  
  (\ref{lem:fvsell-pathstr-onepergadget}): 
  From~(\ref{lem:fvsell-pathstr-subpaths}), we know that each path~$Q\in\{A_2,\dots,A_{q+1},B_1,\dots,B_{K+1}\}$ is a subpath of~$P$, and the first vertex on~$P$ after~$s$ is in~$V(A_1)\setminus\{a_{1,L}\}$.
  If~$Q=A_y$, $2\leq y\leq q+1$, we know that~$a_{y,1}$ is only incident with vertices in~$U_{y-1}\cup\{s\}\cup\{a_{y,2}\}$. 
  It follows that for each~$U_y$ at least one vertex is contained in~$P$. 
  If~$Q=B_z$, $2\leq z\leq K+1$, we know that~$b_{z,1}$ is only incident with vertices in~$F_{z-1}\cup\{s\}\cup\{b_{z,2}\}$. 
  It follows that for each~$F_z$ at least one vertex is contained in~$P$.
  Suppose there is a set~$X\in\{U_1,\dots,U_q,F_1,\dots,F_K\}$ such that at least two vertices from~$X$ are contained in~$P$.
  Recall that by construction, each vertex in~$U\cup F$ has~$M^2$ degree\hyp one neighbors.
  Then~$P$ has at least~$M^2\cdot\binom{k}{2}+M^2\cdot q+M^2>\ell$ neighbors, yielding a contradiction.
  Hence, we know that for each~$U_i$ and~$F_j$ exactly one vertex is contained in~$P$.

  (\ref{lem:fvsell-pathstr-gadgetnghbrs}): 
  Let~$v\in U_y$ for some~$y\in\{1,\ldots,q\}$.
  Suppose that~$v'\neq a_{y,L}$ (for~$v''$, this works analogously).
  We know that~$P$ contains~$A_i$ as a subpath. 
  Hence,~$a_{y,L}$ is adjacent to the other vertex in~$U_y\setminus\{v\}$ on~$P$, yielding a contradiction to~(\ref{lem:fvsell-pathstr-onepergadget}).
  
  In the same way, we can prove the claim for~$v\in F_z$ for some~$z\in\{1,\ldots,K\}$.
\end{proof}

We proceed proving that the instance obtained from \cref{constr:fvsellnokern} is a \yes-instance if and only if at least one input instance is a \yes-instance.

\begin{lemma}
 \label{lem:multicliquesspfvsell}
 Let~$(G_a)_{a=1,\ldots,p}$ be~$p=2^q$
 instances of \mcclique{}
 that are $\calR$\hyp equivalent,
 where \(q\in\N\).
 Let $(G,s,t,k',\ell)$~be the \sspAcr{}-instance obtained from \cref{constr:fvsellnokern}.
 Then at least one instance~$G_a$
 is a \yes-instance if and only if $(G,s,t,k',\ell)$ is \yes-instance for~\sspAcr.
\end{lemma}

\begin{proof}
 \RD{}
 Let~$G_a$ be a \yes-instance
 for some~\(a\in\{1,\dots,p\}\)
 and let $C$~be a $k$-clique in~$G_a$.
 Construct an $s$-$t$~path~$P$ as follows:
 $P$ starts at~$s$,
 then goes to~$a_{1,1}$,
 follows along the vertices only in~$A_1,\dots,A_{q+1}$ and~$U$ until~$a_{q+1,L}$,
 while selecting the vertices in~$U$
 such that only the vertices corresponding to~$V(G_a)$ are not in the neighborhood yet.
 This is possible since, for each~$b\in\{1,\ldots,p\}$, only one of~$u_{y,0}$ and~$u_{y,1}$ is adjacent to
 the vertices in~$V(P_b)$. 
 Next, follow the vertices in~$V(P_1),\dots,V(P_p)$,
 avoiding the vertices in~$V(P_a)$
 by using~$w_{a,2},\zvertex,w_{a+1,2}$.
 Then follow, starting at~$b_{1,1}$ towards~$b_{K+1,L}$ and then to~$t$ by only selecting the vertices corresponding to the edges in~$C$.
 This path contains~
 \begin{itemize}[{$(3(p+1)-1)$}]
  \item[2] vertices~$s$ and~$t$,
  \item[$(q+1)\cdot L$] vertices which are all vertices from the set~$A_1\uplus\cdots\uplus A_{q+1}$,
  \item[$q$] vertices from the set~$U$,
  \item[$(p-1)\cdot n$] vertices which are all vertices from $\biguplus_{b\in\{1,\ldots,p\}\setminus\{a\}} V(P_a)$,
  \item[$(3(p+1)-1)$] vertices which are all vertices from $\biguplus_{b\in\{1,\ldots,p\}\setminus\{a,a+1\}} \{w_{b,1},w_{b,2},w_{b,3}\}$, $w_{a,1}$, $w_{a,2}$, $h$, $w_{a+1,2}$, and~$w_{a+1,3}$,
  \item[$(K+1)\cdot L$] vertices which are all vertices from the set~$B_1\uplus\cdot\uplus B_{K+1}$, and
  \item[$K$] vertices, one from each~$F_z$, $z\in\{1,\ldots,K\}$.
 \end{itemize}
 That is,~$P$ contains
 \[ 2+(q+1)L+q+(p-1)n+(3(p+1)-1)+(K+1)L+K\leq k'\]  
 vertices.
 Moreover, path~$P$ is neighboring
 \begin{itemize}[{$|E|-K$}]
  \item[$q\cdot M^2$] degree-one vertices neighboring~$U$, i.e., $M^2$ degree-one vertices from each of the~$q$ vertices from~$U$ in~$P$,
  \item[$K\cdot M^2$] degree-one vertices neighboring~$F$, i.e., $M^2$ degree-one vertices from each of the~$K$ vertices from~$F$ in~$P$,
  \item[$k$] vertices on the path~$P_a$ (those corresponding to the vertices of clique~$C$),
  \item[$|E|-K$] vertices in~$F$,
  \item[$q$] vertices from~$U$, and
  \item[$2$] vertices~$w_{a,3}$ and~$w_{a+1,1}$.
 \end{itemize}
 That is, $P$ is neighboring 
 \[ q\cdot M^2 + K\cdot M^2 + k + |E|-K +q +2 = q\cdot M^2 + K\cdot M^2 + M \leq \ell \]
 vertices.
 Hence,~$P$ is a solution~$s$-$t$ path in~$G$.
 
 \LD{}
 Let~$(G,s,t,k',\ell)$ be a \yes-instance for~\sspAcr.
 Let~$P$ be a solution $s$-$t$~path.
 We claim that if~$P$ contains a vertex in~$V(P_a)$
 for some~$a\in\{1,\ldots,p\}$,
 then it contains all vertices in~$V(P_a)$.
 Suppose not, that is,
 there is an~$a\in\{1,\ldots,p\}$ such that~$1\leq |V(P)\cap V(P_a)|<n$.
 Note that $N(V(P_a))\subseteq U\cup F\cup \{w_{a,3},w_{a+1,1}\}$.
 Since $1\leq |V(P)\cap V(P_a)|<n$, there is a vertex~$v\in V(P_a)\cap V(P)$ such that at least one of its neighbors in~$V(P_a)$ is not contained in~$V(P)$.
 It follows that in~$P$,~$v$ is adjacent to
 a vertex in~$U\cup F$.
 This contradicts \cref{lem:fvsell-pathstr}(\ref{lem:fvsell-pathstr-gadgetnghbrs}).
 
 From \cref{lem:fvsell-pathstr}(\ref{lem:fvsell-pathstr-onepergadget}) and (\ref{lem:fvsell-pathstr-gadgetnghbrs}), 
 we know that~$P$ contains~$|E|-\binom{k}{2}+q+M^2\cdot\binom{k}{2}+M^2\cdot q$ neighbors not contained in~$A_1\cup \bigcup_{a=1}^p V(P_a)$.
 By the values of~$k'$ and~$\ell$, we know that either exactly one~$P_a$ is not contained in~$P$, or there are~$n+2$ vertices from~$A_1$ being not contained in~$P$.
 In the latter case, we have at least~$n+2+|E|-\binom{k}{2}+q+M^2\cdot\binom{k}{2}+M^2\cdot q>M+M^2\cdot\binom{k}{2}+M^2\cdot q=\ell$ neighbors (recall that~$n>k$), yielding a contradiction.
 It follows the former case: there is exactly one~$P_a$ being not contained in~$P$.
 It follows that~$\zvertex\in V(P)$ and~$w_{a,3},w_{a+1,1}\in N(V(P))$.
 
 By \cref{lem:fvsell-pathstr}(\ref{lem:fvsell-pathstr-onepergadget}), from each~$F_z$ there is exactly one vertex contained in~$P$.
 Moreover, for each~$z\in\{1,\ldots,K\}$ and for each~$v\in F_z$ it holds true that~$|V(P_a) \cap N(v)|\geq 2$.
 Hence,~$|N(V(P))\cap V(P_a)|\geq k$, as~$K$ edges cannot be distributed among less than~$k$ vertices.
 It follows that~$A_1$ is a subpath of~$P$.
 
 Since~$|N(V(P))\setminus V(P_a)|=|E|-\binom{k}{2}+q+M^2\cdot\binom{k}{2}+M^2\cdot q+2$, it follows that there must be exactly~$k$ vertices in~$V(P_a)$ neighboring~$P$.
 This witnesses a $k$-clique in~$G_a$, and the statement follows.
\end{proof}

We are ready to prove the main result of this section.

\begin{proof}[Proof of \cref{thm:sspNoPKfvsell}]
 Due to \cref{obs:nopkfvsellRper}, we know that~$\calR$ is a polynomial equivalence relation on the instances of \mcclique{}.
 Let~$G_1,\ldots,G_p$
 be~$p=2^q$ instances of \mcclique{}
 that are $\calR$\hyp equivalent,
 where $q\in\N$.
 We construct an instance $(G,s,t,k',\ell)$ of~\sspAcr{} by applying \cref{constr:fvsellnokern} in time polynomial in~$\sum_{a=1}^p |G_a|$.
 By \cref{lem:multicliquesspfvsell}, we have that $(G,s,t,k',\ell)$ is a \yes-instance if and only if~$(G_a,k)$ is a \yes-instance for some~$a\in\{1,\ldots,p\}$.
 The set~$W:=U\cup F\cup \{s,\zvertex,t\}$ forms a feedback vertex set with~$|W|\leq 2\log p+K\cdot x$, that is,~$|W|$ is upper-bounded by a polynomial in~$G_a+\log p$ for any~$a\in\{1,\ldots,p\}$.
 Moreover, $\ell=M+ M^2\cdot\binom{k}{2} + M^2\log p$, where~$M:=k+|E|-\binom{k}{2}+\log p+2$ is upper-bounded by a polynomial in~$|G_a|+\log p$.
 Altogether, we described a cross composition from \mcclique{} into~\sspAcr{} parameterized by~$\fvs+\ell$, and the statement follows.
\end{proof}

\section{Conclusion}
Concluding,
we point out that our algorithms
for \WsspAcr{} on graphs of
bounded treewidth (\cref{thm:twsingexp})
can easily be generalized to
a problem variant where also edges
have a weight counting towards the path length,
and so can our subexponential\hyp time
algorithms (\cref{thm:plansubexp}).
Moreover,
the technique of \citet{BCKN15}
that our algorithm is based on
has experimentally been proven
to be practically implementable \citep{FBN15,DKTW17}.

In contrast,
we observed \sspAcr{}
to be a problem
for which provably effective
polynomial\hyp time data reduction
is rather hard to obtain
(\cref{thm:wk1hard,thm:nopktwell,thm:sspNoPKfvsell}).
Therefore,
studying relaxed models
of data reduction with performance guarantees
like approximate \citep{LPRS17,FKRS18}
or randomized kernelization \citep{KW14}
seems worthwhile.

Indeed,
some of our positive results
on kernelization,
in particular
our problem kernels
of size \(\vc^{O(r)}\)
in \(K_{r,r}\)\ssfree{} graphs
and of size \(\fes^{O(1)}\)
in graphs of feedback edge number~\(\fes\)
for \sspAcr{}
(\cref{thm:kernelkrr,thm:bikernelfes}),
for now,
can be mainly seen as a proof of concept,
since they employ
the quite expensive weight reduction algorithm
of \citet{FT87}
and we have no ``direct'' way
of reducing \WsspAcr{} back to \sspAcr{}.
On the positive side,
our solution algorithms
also work for \WsspAcr{},
so that they can be applied
to the linear\hyp time computable
weighted shrunk instances
and stripping the weights is not necessary
from a practical point of view.

\paragraph{Acknowledgments.}
This research was initiated during a research visit of René van Bevern and Oxana Yu.\ Tsidulko at TU Berlin.
The authors are grateful to anonymous reviewers of \emph{Networks} for their constructive feedback.

\paragraph{Funding.}
Till Fluschnik was
supported by the German Research Foundation,
grant NI~369/18.
René van Bevern and Oxana Yu.\ Tsidulko
were supported by the
Russian Foundation for Basic Research (RFBR),
grant~18-501-12031 NNIO\textunderscore a,
while working on
\cref{sec:apgs,sec:fes,sec:notwkern}.
While working on \cref{sec:plansubexp,sec:tw},
René van Bevern was supported by RFBR grant~16-31-60007~mol\textunderscore a\textunderscore dk
and Oxana Yu.\ Tsidulko
was supported by the
Ministry of Science and Higher Education of the Russian Federation under the 5-100 Excellence Programme.

\bibliographystyle{plainnat}
\bibliography{main-secluded}

\end{document}